%% file: main.tex
\documentclass[11pt]{article}

\usepackage[margin=1in]{geometry}
\usepackage{setspace}
\usepackage{microtype}
\usepackage{lmodern}
\usepackage[T1]{fontenc}
\usepackage{bm}
\usepackage{enumitem}

\usepackage{amsmath,amssymb,amsthm,mathtools}

\usepackage{bbm}

\theoremstyle{plain}
\newtheorem{proposition}{Proposition}[section]
\newtheorem{lemma}{Lemma}[section]

\theoremstyle{definition}

\theoremstyle{remark}

\usepackage{booktabs}
\usepackage{array}
\usepackage{colortbl}
\usepackage{graphicx}
\usepackage{caption}
\usepackage{subcaption}
\usepackage{adjustbox}
\usepackage{tikz}
\usetikzlibrary{positioning,calc,arrows.meta,automata,fit}
\usepackage{ifthen}
\usepackage{xcolor}
\definecolor{darkgreen}{HTML}{008f00}
\usepackage{makecell}
\usepackage[skins]{tcolorbox}
\tcbuselibrary{breakable}

\usepackage{hyperref}
\hypersetup{
    colorlinks=true,
    linkcolor=magenta,
    filecolor=magenta, 
    citecolor=blue,      
    urlcolor=blue
}
\definecolor{forestgreen}{RGB}{34,139,34}
\definecolor{bblue}{HTML}{1E88E5} 
\definecolor{bgreen}{HTML}{004D40}
 \definecolor{bred}{HTML}{D81B60}
 
\usepackage[noblocks]{authblk}

\usepackage[
    backend=biber,
    style=nature,
    natbib=true,
    date=year,
    doi=true,
    isbn=false,
    url=false,
    eprint=false
]{biblatex}

\AtEveryBibitem{%
  \clearfield{note}%
}
\AtEveryCitekey{\clearlist{publisher}}
\AtEveryBibitem{\clearlist{publisher}}
\renewbibmacro{in:}{}
\newif\ifsubmissionsnapshot
\submissionsnapshottrue

\ifsubmissionsnapshot
  \newcommand{\resultspath}{results/}
\else
  \newcommand{\resultspath}{../results/}
\fi

\title{How should we select test-negative controls? A causal perspective in the era of multiplex respiratory testing}
\author[1,2]{Christopher B. Boyer\thanks{Corresponding author: \href{mailto:boyerc5@ccf.org}{boyerc5@ccf.org}}}
\author[3]{Kendrick Qijun Li}
\author[4]{Xu Shi}
\author[5,6,7]{Marc Lipsitch}
\author[8]{Eric J. Tchetgen Tchetgen}
\affil[1]{Department of Quantitative Health Sciences, Cleveland Clinic, Cleveland, OH, USA}
\affil[2]{Department of Medicine, Cleveland Clinic Lerner College of Medicine of Case Western Reserve University, Cleveland, OH, USA}
\affil[3]{Department of Biostatistics, St. Jude Children's Research Hospital, Memphis, TN, USA}
\affil[4]{Department of Biostatistics, University of Michigan, Ann Arbor, MI, USA}
\affil[5]{Division of Infectious Diseases and Geographic Medicine, Department of Medicine, Stanford University, Stanford, CA, USA}
\affil[6]{Department of Biology, Stanford University, Stanford, CA, USA}
\affil[7]{Center for International Security and Cooperation, Freeman-Spogli Institute, Stanford University, Stanford, CA, USA}
\affil[8]{Department of Statistics and Data Science, University of Pennsylvania, Philadelphia, PA, USA}
\date{\today}

\begin{document}
\maketitle

\newpage
\clearpage

\begin{abstract}
\singlespacing
The test-negative design (TND) is widely used to estimate vaccine effectiveness (VE) for respiratory pathogens by comparing vaccination odds among test-positive cases versus test-negative controls. A central yet underexplored design element is which test-negative illnesses constitute valid controls. With rapid multiplex PCR panels, investigators can now identify specific non-focal pathogens among test-negative patients, allowing for better characterization of ``test-negative illness'', but also revealing a mixture of control outcomes that may each satisfy or violate causal assumptions. We synthesize recent causal identification results for the TND and show that they imply two distinct interpretations of control selection: 1) a sampling view in which controls represent the source population and 2) a bias-correction view in which controls function as negative control outcomes under equi-confounding. Building on these interpretations, we develop a framework for multiplex-informed control selection. We propose a taxonomy that distinguishes controls that serve primarily as \emph{exposure proxies} (sharing unmeasured determinants of infection) from those that serve as \emph{testing proxies} (sharing unmeasured determinants of care-seeking), derive implications for pathogen-specific and pooled estimators, and suggest three practical principles for control selection: vaccine irrelevance, avoidance of entanglement with other interventions, and testing-process comparability. We also formalize nuances introduced by multiplex panels, including co-detections and pan-negative episodes, and outline when standard pooled estimators remain valid versus when alternative estimators are needed. In simulations across 9 scenarios, we demonstrate violations concentrated in a single control pathogen can substantially bias pooled TND estimates, whereas a pre-specified pathogen screening estimator remained unbiased. 
\end{abstract}
\vspace{0.5em}
\newpage

\section{Introduction}
The test-negative design (TND) has become a mainstay of post-licensure vaccine evaluation for respiratory pathogens \citep{jacksonTestnegativeDesignEstimating2013,sullivanTheoreticalBasisTestNegative2016,lewnardMeasurementVaccineDirect2018,boyerIdentificationEstimationVaccine2026}. Because it can be embedded in routine clinical testing and may mitigate confounding by healthcare-seeking behavior relative to other observational approaches, the TND has been widely used for evaluating influenza \citep{frutos2025interim}, COVID-19 \citep{andrews2022covid}, and RSV vaccines \citep{payne2024respiratory}. In the canonical design, symptomatic individuals who seek care are tested for a focal pathogen; those testing positive serve as cases and those testing negative serve as controls. Vaccine effectiveness (VE) is then estimated from the odds ratio comparing vaccination coverage among test-positive versus test-negative patients.

Despite this popularity, TND validity relies on strong assumptions, and subtle design choices can introduce bias \citep{sullivanPotentialTestnegativeDesign2014}. A persistent ambiguity concerns the role of controls: if test-negative individuals are symptomatic but lack the focal pathogen, what are they sick with, and does it matter for identification? A few influenza studies have attempted to characterize common sources of test-negative infection \citep{chuaUseTestnegativeControls2020a,fengAssessmentVirusInterference2017,fengEstimatingInfluenzaVaccine2018}, but most have treated ``test-negative illness'' as a single composite outcome (often by necessity), implicitly assuming the specific causes of non-focal illness are unimportant or equally valid so long as the symptom screen is met.

This simplification is increasingly untenable. Rapid multiplex PCR assays can simultaneously detect many respiratory viruses---and sometimes bacteria---from a single specimen, making the true heterogeneity of ``test-negative illness'' explicit \citep{huangMultiplexPCRSystem2018}. However, these pathogens may themselves vary in terms of their validity as test-negative controls. For instance some may be affected by focal vaccination, while others are affected by a vaccine or preventative intervention that's correlated with focal vaccination; some may provide nonspecific or short-term immunity against focal infection, while others may have different symptom profiles or test-seeking pathways. 

Multiplex testing therefore raises a crucial question: having revealed the sources of test-negative illness which pathogens should be included as controls? One the one hand, pathogens that violate assumptions may be identified and excluded. At the same time, naive exclusion risks ``collider bias'' \citep{wacholder1992selection}. Multiplex testing also introduces novel complications, such as co-detections (individuals testing positive for multiple pathogens) and pan-negatives (individuals testing negative for all pathogens), that must be carefully considered. What is needed is a clear methodological framework to guide these choices.

In this article, we review recent causal frameworks for the TND (Section~\ref{sec:review}) and show that they imply two contrasting views on control selection (Section~\ref{sec:two_views}). Building on the negative-control view, we examine the implications of multiplex testing and their relationship to unmeasured confounding (Section~\ref{sec:multiplex_implications}), introduce a taxonomy of controls as exposure versus testing proxies (Section~\ref{sec:proxy_taxonomy}), propose practical selection principles (Section~\ref{sec:three_principles}), evaluate them in simulation (Section~\ref{sec:sim}), and discuss implications and limitations (Section~\ref{sec:discussion}).

\section{A review of causal models of the TND}\label{sec:review}

\subsection{Setup and notation}
Let $X$ represent a vector of baseline covariates, $V$ an indicator of vaccination status for the focal vaccine (1: vaccinated, 0: unvaccinated), and $Y$ a tri-level outcome at the end of follow up where:
$$Y:= \begin{cases}
    Y = 1 & \text{test-positive case} \implies (I_{+} = 1, T = 1) , \\
    Y = -1 & \text{test-negative control} \implies (I_{-} = 1, T = 1) , \\
    Y = 0 & \text{not tested} \implies (T = 0).
\end{cases}
$$
The outcome $Y$ represents \emph{medically attended illness} and is only meaningful given context about the population's testing regime, which we define through additional variables. First, $I_{+}$ and $I_{-}$ indicate symptomatic infection due to the focal pathogen and any test-negative pathogen, respectively, where ``symptomatic'' means symptoms matching the pre-specified set used to screen for inclusion in the TND. For simplicity we begin by assuming $I_{+}$ and $I_{-}$ are mutually exclusive ($I_+ = 1 \implies I_- = 0$, and vice versa)---no individual is simultaneously symptomatically infected by both a focal and a non-focal pathogen within the testing window---which is justified when focal--control co-infection is rare or a clear attribution rule identifies the source of current symptoms. We retain mutual exclusivity throughout Sections~\ref{sec:multiplex_implications}--\ref{sec:three_principles}, because most of its control-selection implications are unchanged, and treat co-infections in Section~\ref{sec:coinfection} and, in depth, Appendix~\ref{sec:app_codetection}. Second, $T$ indicates receipt of a clinical test (1: tested, 0: not tested). Together $I_{+}$, $I_{-}$, and $T$ determine $Y$: in particular, $Y = -1$ requires testing negative for the focal pathogen \emph{and} positive for at least one non-focal pathogen on the panel. Symptomatic, tested individuals negative for all pathogens (``pan-negatives'') are discussed in Section~\ref{sec:pan_negatives}. 

\subsection{A DAG for the TND}
Figure~\ref{fig:tnd-dags} shows directed acyclic graphs (DAGs) representing vaccination, illness processes, and selection into the tested population. The key feature is that TND analyses condition on $T=1$, a post-exposure variable influenced by illness and healthcare behavior; without additional structure, conditioning on $T$ can induce selection or ``collider-stratification'' bias \citep{sullivanTheoreticalBasisTestNegative2016,lewnardMeasurementVaccineDirect2018}.

\begin{figure}[!ht]
\centering
\begin{subfigure}[t]{0.48\textwidth}
\centering
\begin{tikzpicture}[> = stealth, shorten > = 1pt, auto, node distance = 2.5cm, inner sep = 0pt, minimum size = 0.5pt, thick]
  \tikzstyle{every state}=[draw=none, fill=none]
  \node[state] (x) {$X$};
  \node[state] (v) [right of=x] {$V$};
  \node[state] (i) [right of=v] {$I_+$};
  \node[state] (t) [right of=i] {$T$};
  \node[state] (i1) [below of=i] {$I_-$};
  \node[state] (u) [below of=v] {$U$};

  \path[->] (x) edge node {} (v);
  \path[->] (x) edge [out=45, in=135] node {} (i);
  \path[->] (v) edge node {} (i);

  \path[->] (i) edge node {} (t);
  \path[->] (i1) edge node {} (t);
  \path[->] (x) edge [out=45, in=135] node {} (t);

  \path[->] (u) edge node {} (x);
  \path[->] (u) edge node {} (i);
  \path[->] (u) edge node {} (t);
  \path[->] (u) edge node {} (i1);
\end{tikzpicture}
\subcaption{No unmeasured confounding (Section~\ref{sec:identification_1}): $U \to V$ is absent.\label{fig:base-dag}}
\end{subfigure}
\hfill
\begin{subfigure}[t]{0.48\textwidth}
\centering
\begin{tikzpicture}[> = stealth, shorten > = 1pt, auto, node distance = 2.5cm, inner sep = 0pt, minimum size = 0.5pt, thick]
  \tikzstyle{every state}=[draw=none, fill=none]
  \node[state] (x) {$X$};
  \node[state] (v) [right of=x] {$V$};
  \node[state] (i) [right of=v] {$I_+$};
  \node[state] (t) [right of=i] {$T$};
  \node[state] (i1) [below of=i] {$I_-$};
  \node[state] (u) [below of=v] {$U$};

  \path[->] (x) edge node {} (v);
  \path[->] (x) edge [out=45, in=135] node {} (i);
  \path[->] (v) edge node {} (i);

  \path[->] (i) edge [draw=gray!40] node {} (t);
  \path[->] (i1) edge [draw=gray!40] node {} (t);
  \path[->] (x) edge [out=45, in=135] node {} (t);

  \path[->] (u) edge node {} (x);
  \path[->] (u) edge node {} (v);
  \path[->] (u) edge [draw=gray!40] node {} (i);
  \path[->] (u) edge [draw=gray!40] node {} (t);
  \path[->] (u) edge [draw=gray!40] node {} (i1);
\end{tikzpicture}
\subcaption{Equi-confounding (Section~\ref{sec:identification_2}): $U$ may affect $V$ directly; gray arrows are restricted by the equi-confounding assumption.\label{fig:equiconfounding}}
\end{subfigure}
\caption{Directed acyclic graphs for a test-negative study under two identification strategies. Panel~(a): under no unmeasured confounding, all unmeasured factors $U$ affect $V$ only through $X$. Panel~(b): under equi-confounding, $U \to V$ is present, but arrows $U \rightarrow I_+$ and $U \rightarrow I_-$ require equivalent effects of $U$ on $I_+$ and $I_-$ on the multiplicative scale, and arrows $U \rightarrow T$, $I_+ \rightarrow T$, and $I_- \rightarrow T$ require equivalent effects of $U$ on test-seeking for $I_+$ and $I_-$. $V$: vaccination; $I_+$: test-positive illness; $I_-$: test-negative illness; $T$: clinical testing; $X$: measured covariates; $U$: unmeasured factors.}
\label{fig:tnd-dags}
\end{figure}
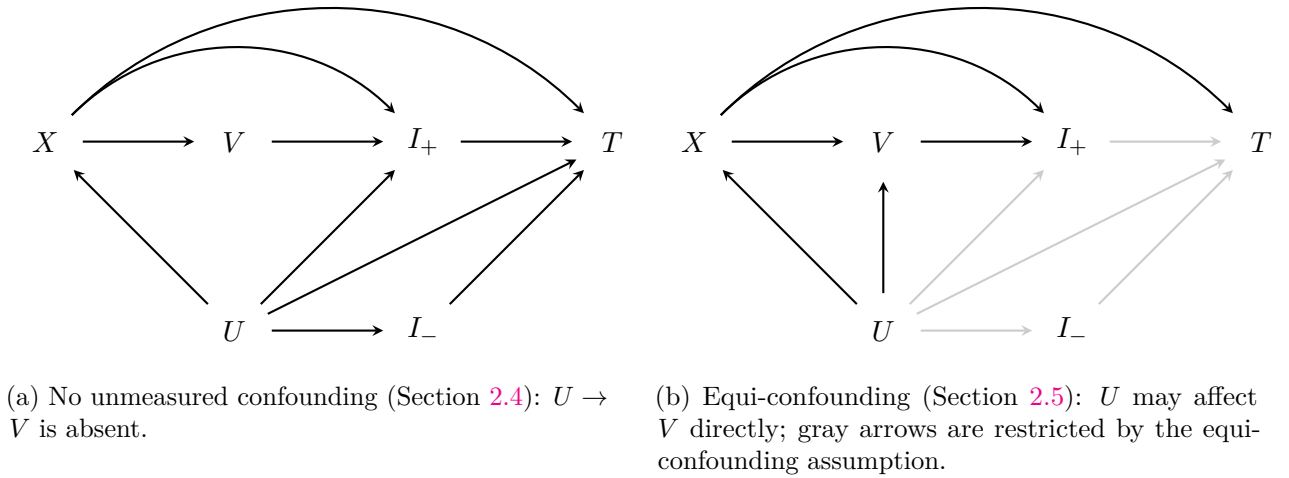

\subsection{Target causal estimand}
A typical estimand is the conditional causal risk ratio for medically attended test-positive illness under vaccination versus no vaccination:
\begin{equation}
RR(X) \;\equiv\; 
\frac{\Pr\!\left(Y^{v=1}=1 \mid X\right)}
{\Pr\!\left(Y^{v=0}=1 \mid X\right)},
\end{equation}
where $Y^{v}$ is the potential outcome under vaccination status $v$, and conditional vaccine efficacy is $VE(X)=1-RR(X)$. The marginal risk ratio is often also of interest \citep{schnitzerEstimandsEstimationCOVID192022}; we focus on the conditional estimand to simplify exposition, but all control-selection insights apply equally to the marginal target (Appendix~\ref{sec:app_marginal}).

\subsection{Identification under no unmeasured confounding}\label{sec:identification_1}
As formalized by \textcite{schnitzerEstimandsEstimationCOVID192022}, under a set of assumptions including consistency, positivity, no causal effect of the focal vaccine on test-negative illness, and no unmeasured confounding conditional on $X$, $RR(X)$ is identified by\footnote{The ``no causal effect'' restriction takes a slightly different form under the two identification strategies: Section~\ref{sec:identification_1} requires the \emph{marginal} equality $\Pr(Y^{v=1}=-1\mid X) = \Pr(Y^{v=0}=-1\mid X)$ (Appendix~\ref{sec:app_identification_1}), whereas Section~\ref{sec:identification_2} requires only the \emph{within-vaccinated} equality $\Pr(Y^{v=1}=-1\mid V=1,X) = \Pr(Y^{v=0}=-1\mid V=1,X)$ (Appendix~\ref{sec:app_identification_2}). The latter is strictly weaker.}
\begin{equation}\label{eqn:or}
OR_{T=1}(X) \equiv \frac{\Pr(Y=1 \mid T=1, V=1, X)\;\Pr(Y=-1 \mid T=1, V=0, X)}
{\Pr(Y=-1 \mid T=1, V=1, X)\;\Pr(Y=1 \mid T=1, V=0, X)}.
\end{equation}
which is the odds ratio among the tested targeted by the standard TND logistic regression (assumptions and proofs in Appendix~\ref{sec:app_identification_1}). Under this view, the TND sampling design functions like incidence-density sampling in case-control studies: it recovers the exposure (or propensity-score) ratio in the source population under so-called ``control exchangeability'' \citep{schnitzerEstimandsEstimationCOVID192022, jiangTNDDREfficientDoubly2023}.

\subsection{Identification under equi-confounding}\label{sec:identification_2}
We have recently shown that the no-unmeasured-confounding assumption can be relaxed to permit particular confounding structures \citep{boyerIdentificationEstimationVaccine2026}: when the net effect of all unmeasured factors $U$ is equivalent on the ratio scale for test-positive and test-negative illness given $X$,
\begin{equation}\label{eqn:orec}
     \dfrac{\Pr(Y^{v=0} = 1  | V = 1, X)}{\Pr(Y^{v=0} = 1 | V = 0, X)} = \frac{\Pr(Y^{v=0} = -1 | V = 1, X)}{\Pr(Y^{v=0} = -1 | V = 0, X)},
\end{equation}
a condition we call multiplicative (odds-ratio) \emph{equi-confounding} \cite{tchetgenUniversalDifferenceinDifferencesCausal2023,parkUniversalDifferenceinDifferencesApproach2023}. Combined with consistency, no causal effect of the focal vaccine on test-negative illness, and a modified positivity assumption, $OR_{T=1}(X)$ then identifies the causal conditional risk ratio \emph{among the vaccinated},
\begin{equation}
RR_V(X) \;\equiv\; 
\frac{\Pr\!\left(Y^{v=1}=1 \mid V = 1, X\right)}
{\Pr\!\left(Y^{v=0}=1 \mid V = 1, X\right)},
\end{equation}
which equals $RR(X)$ when $X$ captures all relevant effect modifiers (assumptions and proofs in Appendix~\ref{sec:app_identification_2}).

This builds on the insight that, given the no-effect restriction, test-negative illness is functionally a negative control outcome \citep{lipsitchNegativeControlsTool2010a} for which bias correction via difference-in-differences is possible under equi-confounding on a suitable scale \citep{soferNegativeOutcomeControl2016,tchetgenUniversalDifferenceinDifferencesCausal2023,parkUniversalDifferenceinDifferencesApproach2023}. The TND is then a form of outcome-dependent sampling of the target outcome and an associated negative control. 

\subsection{Graphical intuition for equi-confounding}
Figure~\ref{fig:equiconfounding} illustrates which parts of the data-generating process are most consequential for equi-confounding. Relative to Figure~\ref{fig:base-dag}, $U$ may now affect $V$ directly ($U\rightarrow V$), and gray arrows mark the pathways where equi-confounding imposes parametric restrictions.

The equi-confounding assumption~\eqref{eqn:orec} factors into two sufficient components: (A) $U$ has equivalent multiplicative effects on symptomatic focal ($I_+$) and non-focal ($I_-$) illness, and (B) $U$ does not differentially modify test-seeking ($T$) for focal versus non-focal illness. Graphically, the gray arrows $U \rightarrow I_+$ and $U \rightarrow I_-$ must carry the same multiplicative effect, and $U$ must not modify $I_+ \rightarrow T$ relative to $I_- \rightarrow T$ (though the background probability of test-seeking may still differ). We state these conditions formally in Appendix~\ref{sec:app_identification_2}.

\section{Two views on the purpose of control selection in a TND}\label{sec:two_views}

The two identification arguments above---(i) no unmeasured confounding and ``control exchangeability'' (Section \ref{sec:identification_1}), and (ii) equi-confounding via assuming test-negative illness is a negative control outcome (Section \ref{sec:identification_2})---imply different views on the role of  controls in the TND and, therefore, different visions for multiplex-informed control selection. The first has dominated the TND literature, but the second is arguably more consistent with the original motivation for the design and better suited to control selection under multiplex testing; we adopt it for the remainder of the paper (although we discuss practical implications under View 1 in Appendix~\ref{sec:app_two_views_diverge}). Box~\ref{box:two-views} contrasts the two, which we describe in turn below.

\subsection{View 1: Controls represent vaccination rates in the source population}
Under the conventional interpretation, largely drawn from case-control literature, controls approximate the vaccination distribution (or propensity score) in the \emph{source population} that gave rise to cases \citep{wacholder1992selection}. This ``sampling'' view suggests a relatively conservative stance on control refinement: defining controls using post-baseline outcomes that are themselves affected by causes of vaccination or testing can introduce selection bias---for example, by opening collider paths through $T$ or through pathogen-specific symptom severity---leading to controls that are no longer ``representative'' of the source population. Multiplex testing is therefore used \emph{defensively}, to exclude clearly invalid controls (pathogens plausibly affected by the focal vaccine, or with sharply different testing pathways), rather than to build ever-finer pathogen-specific control strata.

\subsection{View 2: Controls are a negative control outcome for bias correction}
Under the equi-confounding approach, controls are not merely a sampling device but an \emph{auxiliary (negative control) outcome} used to correct bias due to unmeasured confounding---similar to the indirect cohort design \cite{broomePneumococcalDiseasePneumococcal1980}. Under this perspective, a valid control (i) is not causally affected by vaccination, yet (ii) shares the confounding and selection structure of the focal outcome such that differences in relative vaccination rates reflect the extent of residual confounding. This ``bias-correction'' view makes pathogen choice central rather than incidental.  Multiplex testing creates a \emph{menu of candidate negative controls} that differ in how well they meet these conditions, so the information is used \emph{constructively} to prioritize pathogens that share the dominant unmeasured bias mechanisms affecting the focal outcome while remaining unaffected by vaccination. Pooling several valid controls can improve precision, but a violation for even one pathogen biases the pool, so careful screening is essential.


\begin{tcolorbox}[
  float=ht,
  colback=white,
  colframe=black!70,
  coltitle=white,
  colbacktitle=black!70,
  fonttitle=\bfseries\small,
  title={Box 1: Two views on the role of test-negative controls},
  boxrule=0.6pt,
  arc=2pt,
  left=6pt, right=6pt, top=4pt, bottom=4pt
]
\label{box:two-views}
\small
\renewcommand{\arraystretch}{1.4}%
\begin{tabular}{@{} >{\itshape}p{2.8cm} p{5.1cm} p{5.1cm} @{}}
  & \textbf{View 1: Sampling} & \textbf{View 2: Bias correction} \\[0.3em]
  \arrayrulecolor{black!25}\hline\noalign{\vskip 4pt}
  Identification basis &
    No unmeasured confounding; control exchangeability &
    Equi-confounding; negative control outcome \\
  Role of controls &
    Represent vaccination distribution in the source population &
    Auxiliary outcome sharing the confounding and selection structure of the focal outcome \\
  Primary risk from refinement &
    Selection bias from conditioning on post-baseline information &
    Violation of equi-confounding if control and focal illness have different confounding structures \\
  Use of multiplex data &
    Defensive: exclude clearly invalid pathogens &
    Constructive: select pathogens that best approximate residual confounding \\[0.2em]
  \arrayrulecolor{black!25}\hline\noalign{\vskip 4pt}
  \multicolumn{3}{@{}c@{}}{\textbf{Shared requirement:} Vaccination does not causally affect the chosen control outcome.}
\end{tabular}
\renewcommand{\arraystretch}{1.0}%
\end{tcolorbox}

\section{Implications of multiplex testing}\label{sec:multiplex_implications}

\subsection{``Test-negative illness'' is no longer a single object}
Historically, control illness has been loosely defined, e.g., ``influenza test-negative influenza-like illness'' \citep{jacksonTestnegativeDesignEstimating2013}. Multiplex PCR resolves it into specific pathogens---RSV, rhinovirus, metapneumovirus, adenovirus, seasonal coronaviruses, parainfluenza, and others---each with distinct epidemiology, symptom profiles, and testing pathways \citep{huangMultiplexPCRSystem2018}. When identifying assumptions hold for some of these pathogens but not others, pooling them into a single control group can induce bias. Multiplex testing thus shifts the design question from ``Is the test-negative group valid?'' to ``Which components are valid, and how should they be combined?''

Notationally, the pathogenic source of test-negative illness can be represented by a vector $(I_{-1},\dots,I_{-K})$, where $I_{-k} \in \{0, 1\}$ indicates symptomatic illness due to test-negative pathogen $k$. The observed test-negative group can then be represented as a mixture across pathogens:
$$
I_{-} = \mathbbm{1} \left\{\bigcup_{k=1}^{K} I_{-k} = 1\right\}
$$
Multiplex testing introduces the possibility of detecting two potential types of co-infection: \emph{control--control} co-infections, i.e., those within $I_{-}$, and \emph{focal--control} co-infections, i.e., those between $I_+$ and any element of $I_{-}$. Both are treated in more depth below. We begin with the simple case where there are no co-infections, such as when all co-detections can be resolved to a single pathogenic source of current symptoms.

\subsection{We must consider identifiability under multiple candidate negative controls} 

Absent co-infections, we can accomodate multiplex results by expanding $Y$ into additional mutually exclusive categories $Y \in \{1,0,-1,\ldots,-K\}$ for a symptomatic  episode, where $Y=1$ indicates testing positive for the focal pathogen, $Y=0$ no test, and $Y=-k$ testing positive for (focal-negative) pathogen $k$. The conventional pooled ``test-negative'' control group is then the event $Y\in\{-1,\ldots,-K\}$, whereas pathogen-specific controls compare $Y=1$ versus $Y=-k$.

This notation reveals two implications of multiplex results. First, each pathogen-specific contrast has its own identifying conditions: the ``no vaccine effect on controls'' requirement becomes $\Pr(Y^{v=1}=-k\mid \cdot)=\Pr(Y^{v=0}=-k\mid \cdot)$ for each $k$, and the plausibility of equi-confounding may vary across $k$ depending on patheogen-specific determinants of exposure and care-seeking (see Appendix~\ref{sec:app_pooling} for identification results). Second, pooling produces a mixture estimand that can mask violations concentrated in a subset of pathogens: even when most satisfy assumptions, a single pathogen affected by vaccination or a different testing channel can bias the pool if it forms a sizable share of controls during part of the season. Figure~\ref{fig:multiplex-decomp} depicts this: the ``test-negative group'' is now a mixture of pathogen-specific illnesses and identification conditions plausible for one pathogen may fail for another. 

Accordingly, we recommend reporting pathogen-specific VE estimates as a stability analysis (grouping very low-prevalence pathogens, e.g. by viral family or transmission route, to preserve precision). Let $OR_k(X)$ denote the conditional odds ratio comparing vaccination odds among focal positives ($Y=1$) versus pathogen-$k$ controls ($Y=-k$) within strata of $X$; the pooled $OR_{T=1}(X)$ is then a control-prevalence-weighted average of $\{OR_k(X)\}_{k=1}^K$, weighted by each pathogen's share of the vaccinated control pool (Appendix~\ref{sec:app_pooling}, Equation~\ref{eq:pooled_decomp}). When all controls share the same source of unmeasured confounding and satisfy the identifying assumptions, the $\{OR_k(X)\}$ should agree within sampling variability. We therefore describe an omnibus homogeneity test and---because rejection does not reveal \emph{which} control is at fault---a complementary pairwise control--control diagnostic that equals one in expectation when both controls are valid, so departures from unity localize the offending pathogen (Appendix~\ref{sec:app_pooling}, Equation~\ref{eq:or_jk}). Unlike a traditional TND, multiplex testing thus lets control heterogeneity be examined and addressed through pre-specified exclusions and sensitivity analyses rather than assumed away. 

\subsection{Handling co-detections}\label{sec:coinfection}
Multiplex panels can detect nucleic acid from more than one organism in a single specimen (a \emph{co-detection}), raising the question of whether this reflects genuine co-causation of symptoms (a \emph{co-infection}) or an incidental finding; Appendix~\ref{sec:app_codetection} formalizes the mapping from detection to illness attribution. Prior literature suggests that co-detection rates can be nontrivial (7--27\%) especially in pediatric populations and during peak respiratory virus seasons \cite{krammerTenyearRetrospectiveData2024,yonghaoCharacteristicsRespiratoryCoinfections2026}. When a credible \textit{attribution rule} assigns the episode to a single source---e.g., highest viral load (lowest cycle threshold), best syndromic match, or clinical adjudication---mutual exclusivity is restored and the main-text methods apply unmodified. When no such rule exists or genuine \emph{co-infection} is likely, two cases must be distinguished. 

\textbf{Control--control co-infection} is the co-occurrence of non-focal pathogens. Provided the overlapping controls are both valid, pathogen-specific and pooled estimators remain unbiased. Each $OR_k$ can be estimated from control pathogen $k$ in isolation regardless of co-infection (though combining the estimates, e.g. by inverse-variance weighting, must account for their nonindependence). The pooled estimator is invariant in an even stronger sense: any non-excluding rule merely relabels which $-k$ bucket an episode enters, leaving the union $\bigcup_k \{Y = -k\}$ unchanged (Appendix~\ref{sec:app_control_control}). The choice between estimation strategies is thus mainly about precision with pooling yielding higher effective sample size. 

Alternatively, co-infections may represent latent differences in exposure risk. When this is unrelated to vaccination or care-seeking strategies, the results above apply without change. However, when latent exposure risk is related, and only applies to controls, it can invalidate equi-confounding assumption (see Figure~\ref{fig:exclude_codetect}) in which case excluding co-infection episodes may be preferred.

\textbf{Focal--control co-infection}, the co-occurence of focal and non-focal pathogens, is more consequential: under the conventional rule that assigns any focal-positive episode to the case group (by necessity when tests are singleplex), unvaccinated individuals---who face higher focal-infection risk---are preferentially removed from the control pool, so the standard pooled estimator overstates VE, and no change of attribution rule removes the bias. A difference-in-differences estimator that contrasts focal- and control-pathogen risk ratios cancels common multiplicative confounding and recovers the target VE under equi-confounding. We formalise the bias mechanism and identification result in Appendix~\ref{sec:app_focal_control} and evaluate a modified (stacked) Poisson estimator in simulation in Appendix~\ref{sec:app_coinf_sim}.

\subsection{Handling pan-negative episodes}\label{sec:pan_negatives}
A related challenge is \emph{pan-negative} episodes: symptomatic individuals testing negative for the focal pathogen \emph{and} every other pathogen on the panel. These may reflect (i) off-panel infections (bacterial, fungal, or novel viral pathogens), (ii) measurement error (e.g. false negatives), or (iii) a non-infectious source of symptoms (provided symptom screen is effective). Including pan-negatives risks importing controls that violate identifying assumptions, while excluding them may alter the confounding structure if pan-negativity relates to vaccination or care-seeking. When pan-negatives arise purely from measurement error, estimators that \emph{exclude} them still identify $RR_V(X)$ under perfect specificity and non-differential sensitivity (Appendix~\ref{sec:app_imperfect_test}). We recommend excluding them by default and reporting analyses with pan-negative controls as a pre-specified sensitivity analysis.

\section{Distinguishing between exposure proxies and testing proxies}\label{sec:proxy_taxonomy}
Under the equi-confounding view of control selection, multiplex assays present several candidate negative control outcomes, each of which may communicate different information about unmeasured confounding. A useful conceptual distinction is between pathogens that serve as ``exposure proxies'' versus ``testing proxies.'' Figure~\ref{fig:multiplex-proxies} reprises the multiplex DAG to emphasize that different test-negative pathogens may align more closely with unmeasured factors related to exposure ($U_1$), unmeasured factors related to testing ($U_2$), or both. Note, the exposure/testing distinction here refers to the \emph{primary} association: $U_1$ also affects testing indirectly via pathogen-specific illness, but its dominant role is in exposure risk. We discuss each type of proxy in turn below, along with examples of pathogens, and implications for control selection.

\begin{figure}[!ht]
\centering
\begin{tikzpicture}[> = stealth, shorten > = 1pt, auto, node distance = 2.5cm, inner sep = 0pt, minimum size = 0.5pt, thick]
  \tikzstyle{every state}=[draw=none, fill=none]
  \node[state] (x) {$X$};
  \node[state] (v) [right of=x] {$V$};
  \node[state] (i) [right of=v] {$I_+$};
  \node[state] (t) [right of=i] {$T$};
  \node[state] (i1) [below of=i] {$I_{-1}$};
  \node[state] (i2) [below=-0.25cm of i1] {$I_{-2}$};
  \node[state] (i3) [below=-0.35cm of i2] {$\vdots$};
  \node[state] (i4) [below=-0.25cm of i3] {$I_{-K}$};
  \node[state] (u) [below of=v] {$U_1$};
  \node[state] (u2) [below of=u] {$U_2$};

  \path[->] (x) edge node {} (v);
  \path[->] (x) edge [out=45, in=135] node {} (i);
  \path[->] (v) edge node {} (i);

  \path[->] (i) edge node {} (t);
  \path[->] (i1) edge node {} (t);
  \path[->] (i2) edge node {} (t);
  \path[->] (i4) edge node {} (t);
  \path[->] (x) edge [out=45, in=135] node {} (t);

  \path[->] (u) edge node {} (x);
  \path[->] (u) edge node {} (v);
  \path[->] (u) edge node {} (i);
  \path[->] (u) edge node {} (i1);
  \path[->] (u) edge node {} (i4);

 \path[->] (u2) edge [out=135, in=225] node {} (v);
 \path[->] (u2) edge [out=315, in=315] node {} (t);
 \path[->] (u2) edge [out=135, in=270] node {} (x);
\end{tikzpicture}
\caption{DAG illustrating classification of test-negative pathogens as proxies for unmeasured confounders of exposure ($U_1$) or testing behavior ($U_2$).}
\label{fig:multiplex-proxies}
\end{figure}
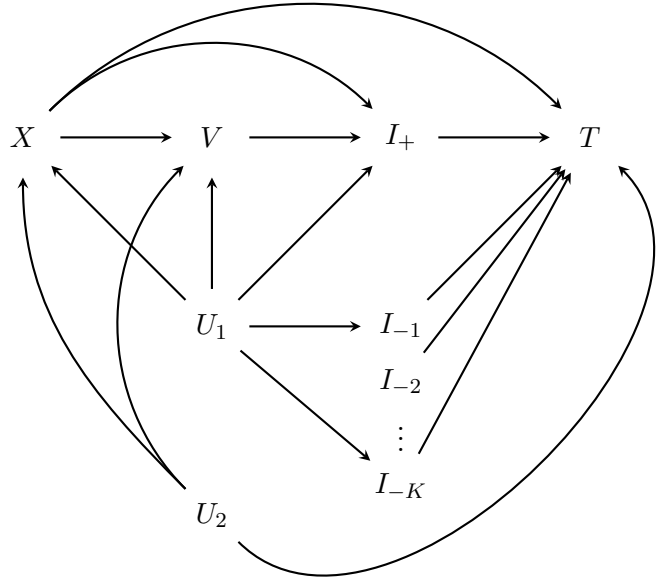

\subsection{Exposure proxies}
Some pathogens share exposure determinants with the focal pathogen (e.g., household crowding, daycare contact, occupation, masking practices) but may be unaffected by the focal vaccine. Detection of such pathogens among test-negative patients can act as a proxy for latent exposure risk. If unmeasured exposure risk confounds $V \rightarrow I^{+}$, then including control pathogens that are similarly exposure-confounded can make equi-confounding more plausible because unmeasured factors that increase focal disease risk also increase control disease risk. 

\paragraph{Examples.}
Table~\ref{tab:proxy-examples} (Panel A) lists examples. For influenza studies, rhinovirus and parainfluenza viruses share the same transmission routes and behavioral determinants of exposure (household crowding, daycare attendance, workplace contact density) without being affected by influenza vaccines. Similarly, for COVID-19 studies, seasonal coronaviruses (OC43, 229E, NL63, HKU1) are among the most closely transmission-aligned pathogens: they spread via similar aerosol/droplet routes and are sensitive to the same indoor-crowding and ventilation-related exposure determinants as SARS-CoV-2, yet are unaffected by mRNA vaccines.

\subsection{Testing proxies}
Other pathogens primarily indicate shared healthcare-seeking and clinical-testing propensity rather than shared exposure risk. Detection of such pathogens among test-negative patients can act as a proxy for latent care-seeking and testing behaviour. If unmeasured determinants of care-seeking $U_2$ confound $V \rightarrow I_+$---for example, when vaccinated and unvaccinated individuals differ in their threshold for seeking clinical testing for a given level of symptomatic illness---then including control pathogens that are similarly testing-confounded can make equi-confounding more plausible. This is often the more pressing threat in modern TND studies, where vaccination status is correlated with healthcare access, insurance, and at-home testing behaviour.

\paragraph{Examples.}
The clearest cases of testing-only proxies are pathogens that produce febrile illness resembling the focal syndrome but spread through non-respiratory routes (Table~\ref{tab:proxy-examples}, Panel B). For instance, enteric viruses (norovirus and rotavirus) and arboviruses (e.g. dengue) can trigger the same symptom-driven clinical encounter as certain respiratory infections---and thus proxy unmeasured care-seeking ($U_2$). Yet because they spread primarily via fecal-oral and vector-borne transmission, they share few of the respiratory exposure structure ($U_1$). 

\begin{table}[ht]
\centering
\caption{Examples of test-negative pathogens classified by their primary proxy role. Exposure proxies share unmeasured determinants of infection risk ($U_1$) with the focal pathogen; testing proxies primarily reflect unmeasured determinants of care-seeking and clinical testing ($U_2$). See Figure~\ref{fig:multiplex-proxies} for the corresponding DAG.}
\label{tab:proxy-examples}
\small
\begin{tabular}{@{}lll@{}}
\toprule
\textbf{Focal pathogen} & \textbf{Candidate control} & \textbf{Shared determinants} \\
\midrule
\multicolumn{3}{@{}l}{\textit{Panel A: Exposure proxies---shared determinants of infection risk ($U_1$)}} \\
\addlinespace[3pt]
Influenza & Rhinovirus/enterovirus & \makecell[tl]{Respiratory/contact transmission; household\\crowding; daycare/school attendance} \\
\addlinespace
Influenza & Parainfluenza viruses & \makecell[tl]{Respiratory droplet transmission; overlapping\\winter seasonality; similar age-contact patterns} \\
\addlinespace
COVID-19 & \makecell[tl]{Seasonal coronaviruses\\(OC43, 229E, NL63, HKU1)} & \makecell[tl]{Aerosol/droplet transmission; sensitivity to\\indoor crowding, ventilation, contact intensity} \\
\addlinespace
COVID-19 & Human metapneumovirus & \makecell[tl]{Respiratory route; winter seasonality; similar\\age-risk profile; congregate-setting exposure} \\
\midrule
\multicolumn{3}{@{}l}{\textit{Panel B: Testing proxies---shared determinants of care-seeking ($U_2$)}} \\
\addlinespace[3pt]
Influenza & Norovirus/rotavirus & \makecell[tl]{Febrile illness triggering same clinical encounter;\\fecal-oral transmission shares $U_2$ (care-seeking)\\but not $U_1$ (respiratory exposure)} \\
\addlinespace
Any respiratory & \makecell[tl]{Arboviruses (e.g., dengue\\on syndromic panels)} & \makecell[tl]{Acute febrile presentation overlapping ILI;\\mosquito-borne transmission orthogonal\\to respiratory exposure determinants} \\
\bottomrule
\end{tabular}
\end{table}

\subsection{Most respiratory pathogens serve as both}
The exposure-proxy/testing-proxy distinction is a useful heuristic, but not absolute: most pathogens carry both exposure and testing signals, and the classification reflects which dimension \emph{dominates} in a given setting. This matters because the respiratory pathogens that dominate multiplex panels, such as rhinovirus, parainfluenza, human metapneumovirus, and seasonal coronaviruses, share \emph{both} the transmission route (hence the exposure determinants in $U_1$) \emph{and} the symptom profile that triggers testing (hence the care-seeking determinants in $U_2$). Under the equi-confounding view, this dual alignment is precisely what makes them strong negative control candidates: they are informative about both sources of unmeasured confounding simultaneously. By contrast, a single-dimension proxy helps only against its matched threat (e.g. testing proxies against differential testing, exposure proxies against differential exposure). Pooling complementary exposure- and testing-ony proxies may reduce bias but approximates equi-confounding only under strong balancing conditions. The practical question for any study is which confounder, $U_1$ (exposure) or $U_2$ (testing), poses the larger unmeasured threat and whether the chosen controls capture it; when both are plausible, joint proxies (respiratory pathogens with similar transmission and presentation) are the most conservative choice.

\section{Three principles for control selection using multiplex testing}\label{sec:three_principles}

We now illustrate the framework with three (non-exhaustive) principles for control selection and their corresponding DAGs (Figure~\ref{fig:violations}).

\subsection{Principle 1: Vaccine irrelevance---exclude pathogens affected by the focal vaccine}
A direct violation occurs when the focal vaccine influences the probability of symptomatic illness for a control pathogen. This could be through cross-protection afforded by vaccine targeted antibodies or viral interference~\cite{fengAssessmentVirusInterference2017} through nonspecific protection or ecological competition. The standard TND contrast then no longer isolates the vaccine's effect on the focal pathogen; instead, it partly compares vaccination odds among cases to odds among controls whose membership is itself affected by vaccination. Under equi-confounding, we have shown previously that this biases VE toward the null when vaccine effects are directionally similar for the focal and control pathogens \citep{boyerIdentificationEstimationVaccine2026}.

Figure~\ref{fig:violations-vaccine-irrelevance} shows this with the dashed $V \rightarrow I_{-2}$ arrow (direct protection) and $V \rightarrow I_{+} \rightarrow I_{-1}$ arrow (indirect protection), which violate the ``no effect on controls'' condition. 

\begin{figure}[p]
\centering
\begin{subfigure}[t]{0.48\textwidth}
\centering
\begin{tikzpicture}[> = stealth, shorten > = 1pt, auto, node distance = 2.3cm, inner sep = 0pt, minimum size = 0.5pt, thick]
  \tikzstyle{every state}=[draw=none, fill=none]
  \node[state] (x) {$X$};
  \node[state] (v) [right of=x] {$V$};
  \node[state] (i) [right of=v] {$I_+$};
  \node[state] (t) [right of=i] {$T$};
  \node[state] (i1) [below of=i] {$I_{-1}$};
  \node[state] (i2) [below=-0.25cm of i1] {$I_{-2}$};
  \node[state] (i3) [below=-0.35cm of i2] {$\vdots$};
  \node[state] (i4) [below=-0.25cm of i3] {$I_{-K}$};
  \node[state] (u) [below of=v] {$U$};

  \path[->] (x) edge node {} (v);
  \path[->] (x) edge [out=45, in=135] node {} (i);
  \path[->] (v) edge node {} (i);

  \path[->] (i) edge node {} (t);
  \path[->] (i1) edge node {} (t);
  \path[->] (i2) edge node {} (t);
  \path[->] (i4) edge node {} (t);
  \path[->] (x) edge [out=45, in=135] node {} (t);

  \path[->] (u) edge node {} (x);
  \path[->] (u) edge node {} (v);
  \path[->] (u) edge node {} (i);
  \path[->] (u) edge node {} (t);
  \path[->] (u) edge node {} (i1);
  \path[->] (u) edge node {} (i2);
  \path[->] (u) edge node {} (i4);
\end{tikzpicture}
\subcaption{TND under multiplex testing\label{fig:multiplex-decomp}}
\end{subfigure}
\begin{subfigure}[t]{0.48\textwidth}
\begin{tikzpicture}[> = stealth, shorten > = 1pt, auto, node distance = 2.3cm, inner sep = 0pt, minimum size = 0.5pt, thick]
  \tikzstyle{every state}=[draw=none, fill=none]
  \node[state] (x) {$X$};
  \node[state] (v) [right of=x] {$V$};
  \node[state] (i) [right of=v] {$I_+$};
  \node[state] (t) [right of=i] {$T$};
  \node[state] (i1) [below of=i] {$I_{-1}$};
  \node[state] (i2) [below=-0.25cm of i1] {$I_{-2}$};
  \node[state] (i3) [below=-0.35cm of i2] {$\vdots$};
  \node[state] (i4) [below=-0.25cm of i3] {$I_{-K}$};
  \node[state] (u) [below of=v] {$U$};

  \path[->] (x) edge node {} (v);
  \path[->] (x) edge [out=45, in=135] node {} (i);
  \path[->] (v) edge node {} (i);
  \path[->] (v) edge [dashed] node {} (i2);

  \path[->] (i) edge node {} (t);
  \path[->] (i) edge [dashed] node {} (i1);
  \path[->] (i1) edge [dashed] node {} (t);
  \path[->] (i2) edge [dashed] node {} (t);
  \path[->] (i4) edge node {} (t);
  \path[->] (x) edge [out=45, in=135] node {} (t);

  \path[->] (u) edge node {} (x);
  \path[->] (u) edge node {} (v);
  \path[->] (u) edge node {} (i);
  \path[->] (u) edge node {} (t);
  \path[->] (u) edge node {} (i1);
  \path[->] (u) edge node {} (i2);
  \path[->] (u) edge node {} (i4);
\end{tikzpicture}
\subcaption{Vaccine irrelevance violation\label{fig:violations-vaccine-irrelevance}}
\end{subfigure}
\hfill
\begin{subfigure}[t]{0.48\textwidth}
\centering
\begin{tikzpicture}[> = stealth, shorten > = 1pt, auto, node distance = 2.3cm, inner sep = 0pt, minimum size = 0.5pt, thick]
  \tikzstyle{every state}=[draw=none, fill=none]
  \node[state] (x) {$X$};
  \node[state] (v) [right of=x] {$V$};
  \node[state] (i) [right of=v] {$I_+$};
  \node[state] (t) [right of=i] {$T$};
  \node[state] (i1) [below of=i] {$I_{-1}$};
  \node[state] (i2) [below=-0.25cm of i1] {$I_{-2}$};
  \node[state] (i3) [below=-0.35cm of i2] {$\vdots$};
  \node[state] (i4) [below=-0.25cm of i3] {$I_{-K}$};
  \node[state] (u) [below of=v] {$U$};
  \node[state] (v2) [below=0.9cm of u] {$V_2$};

  \path[->] (x) edge node {} (v);
  \path[->] (x) edge [out=45, in=135] node {} (i);
  \path[->] (v) edge node {} (i);

  \path[->] (i) edge node {} (t);
  \path[->] (i1) edge node {} (t);
  \path[->] (i2) edge [dashed] node {} (t);
  \path[->] (i4) edge node {} (t);
  \path[->] (x) edge [out=45, in=135] node {} (t);

  \path[->] (u) edge node {} (x);
  \path[->] (u) edge node {} (v);
  \path[->] (u) edge [dashed] node {} (v2);
  \path[->] (u) edge node {} (i);
  \path[->] (u) edge node {} (t);
  \path[->] (u) edge node {} (i1);
  \path[->] (u) edge node {} (i2);
  \path[->] (u) edge node {} (i4);

  \path[->] (v2) edge [dashed] node {} (i2);
\end{tikzpicture}
\subcaption{Other-vaccine entanglement\label{fig:violations-vaccine-entanglement}}
\end{subfigure}
\hfill
\begin{subfigure}[t]{0.48\textwidth}
\centering
\begin{tikzpicture}[> = stealth, shorten > = 1pt, auto, node distance = 2.3cm, inner sep = 0pt, minimum size = 0.5pt, thick]
  \tikzstyle{every state}=[draw=none, fill=none]
  \node[state] (x) {$X$};
  \node[state] (v) [right of=x] {$V$};
  \node[state] (i) [right of=v] {$I_+$};
  \node[state] (t) [right of=i] {$T$};
  \node[state] (i1) [below of=i] {$I_{-1}$};
  \node[state] (i2) [below=-0.25cm of i1] {$I_{-2}$};
  \node[state] (i3) [below=-0.35cm of i2] {$\vdots$};
  \node[state] (i4) [below=-0.25cm of i3] {$I_{-K}$};
  \node[state] (u) [below of=v] {$U$};
  \node[state] (t2) [right of=i2] {$T_2$};

  \path[->] (x) edge node {} (v);
  \path[->] (x) edge [out=45, in=135] node {} (i);
  \path[->] (v) edge node {} (i);

  \path[->] (i) edge node {} (t);
  \path[->] (i1) edge node {} (t);
  \path[->] (i2) edge node {} (t);
  \path[->] (i4) edge node {} (t);
  \path[->] (x) edge [out=45, in=135] node {} (t);

  \path[->] (u) edge node {} (x);
  \path[->] (u) edge node {} (v);
  \path[->] (u) edge node {} (i);
  \path[->] (u) edge node {} (t);
  \path[->] (u) edge node {} (i1);
  \path[->] (u) edge [dashed] node {} (i2);
  \path[->] (u) edge node {} (i4);

  \path[->] (i2) edge [dashed] node {} (t2);
  \path[->] (t2) edge [dashed] node {} (t);
\end{tikzpicture}
\subcaption{Testing-process mismatch\label{fig:violations-at-home-test}}
\end{subfigure}

\caption{DAG extension under multiplex testing, decomposing test-negative illness into pathogen-specific components ($I_{-1},\ldots,I_{-K}$). Panel (a) shows the general multiplex DAG structure. Panels (b) through (d) illustrate common violations motivating pathogen-specific test-negative control selection: (b) vaccine irrelevance violation where $V$ affects a control pathogen; (c) other-vaccine entanglement where a correlated vaccine $V_2$ affects a control pathogen; (d) testing-process mismatch where at-home testing $T_2$ alters clinical testing differentially. Dashed arrows denote pathways that, if present, can compromise identification when pooling test-negative pathogens. $V$ indicates focal vaccination; $I_+$, focal (test-positive) illness; $I_{-k}$, test-negative illness due to pathogen $k$; $T$, clinical testing; $T_2$, at-home testing; $X$, measured covariates; and $U$, unmeasured factors.}
\label{fig:violations}
\end{figure}
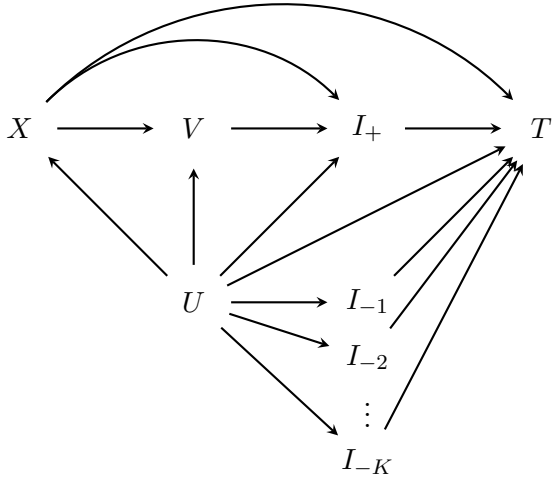
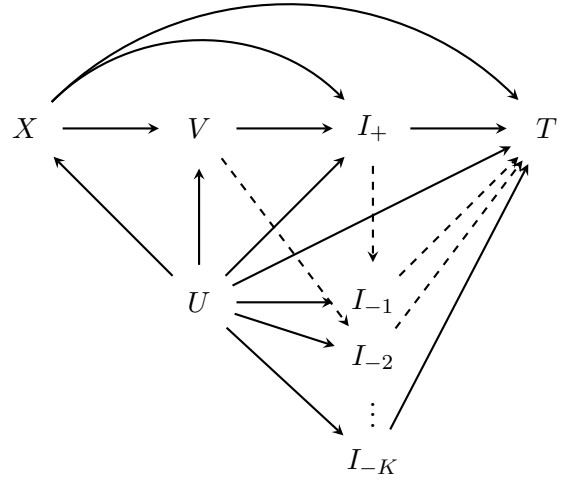
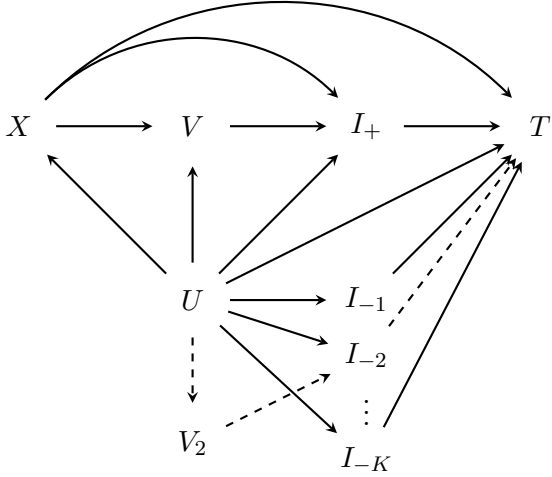
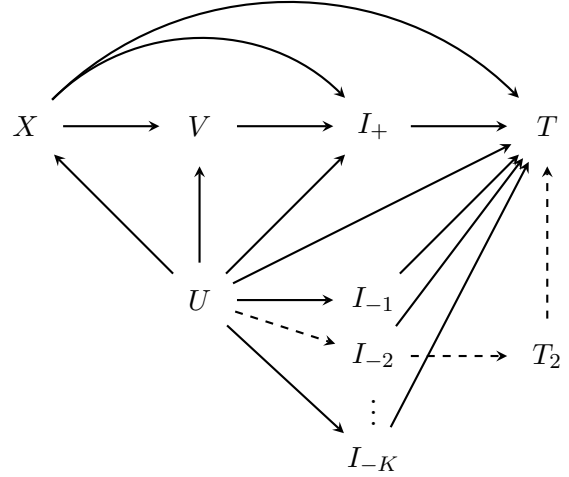

\subsection{Principle 2: Avoid intervention entanglement---exclude or adjust for pathogens with their own preventative intervention}
Suppose a control pathogen has its own vaccine $V_2$ (or another pathogen-specific preventive intervention) that reduces the risk of $I_{-2}$, and uptake of $V_2$ is correlated with focal-vaccine uptake through shared determinants of vaccination behavior (represented by $U$). Conditioning on that pathogen's test-negative illness can then induce bias even when $V$ has no causal effect on it \citep{dollEffectsConfoundingBias2022,payneImpactAccountingCorrelation2023}. As Figure~\ref{fig:violations-vaccine-entanglement} shows, the $U \rightarrow V_2 \rightarrow I_{-2}$ pathway adds an additional source of asymmetry between focal and control outcomes: it does not by itself violate equi-confounding, but it imposes additional balance conditions that make the assumption harder to justify.

\subsection{Principle 3: Testing-process comparability---account for alternative testing pathways (e.g., at-home tests)}
TND studies rely on a clinically tested population, but alternative channels such as at-home rapid tests can make care-seeking and clinical testing pathogen-specific. If a pathogen can be identified at home (via $T_2$) and this alters clinical testing ($T$), the selection mechanism into $T=1$ differs by pathogen \citep{qasmiehMagnitudePotentialBiases2024}. Figure~\ref{fig:violations-at-home-test} depicts one such structure, motivating exclusion of affected pathogens or restriction/stratification by testing era and setting.

More broadly, heterogeneous reasons for testing---at-home tests, mandatory screening, contact tracing, and surveillance-driven testing---each induce distinct selection mechanisms and target distinct VE estimands. \textcite{yuTestNegativeDesignsVarious2026} and \textcite{lipsitchInterpretingVaccineEfficacy2021a} show that the standard odds-ratio estimator is generally biased when individuals with different reasons for testing are pooled, and propose stratification by reason for testing. This complements the pathogen-level perspective developed here: just as control pathogens may satisfy or violate the identifying assumptions to differing degrees (Figure~\ref{fig:violations}), so may different testing pathways within a given pathogen. Investigators should therefore consider both \emph{which pathogens} and \emph{which testing pathways} to include, and stratify or restrict accordingly. 

\section{Simulation study}\label{sec:sim}

We conducted a simulation study to evaluate how multiplex-informed control selection affects bias of test-negative vaccine effectiveness estimators under common deviations from idealized identifying conditions.

\subsection{Data-generating process}
For each replicate we generated a cohort of $n = 10{,}000$ individuals with covariates $X \sim N(0,1)$ and an unmeasured confounder $U \sim N(0,1)$ that affects both exposure risk and care-seeking. Vaccination was drawn as $\Pr(V=1\mid X,U)=\mathrm{logit}^{-1}(-0.5 + 0.3X + 0.5U)$. The focal infection followed a log-linear model
\[
\Pr(I_+=1\mid V,X,U)=\exp(-3.0 + 0.2X + 0.4U + \beta_V V),\qquad \beta_V = \log(0.5),
\]
giving a true VE of 50\%. The $K=5$ test-negative pathogens were generated analogously,
\[
\Pr(I_{-k}=1\mid V,X,U)=\exp(\gamma_{0k} + \gamma_{Xk} X + 0.4U),
\]
with intercepts $\gamma_{0k}\in[-3.5,-2.9]$ and $\gamma_{Xk}\in[0.17,0.25]$ inducing realistic heterogeneity in baseline prevalence. The shared $U$ coefficient ensures equi-confounding holds exactly in the baseline scenario \citep{boyerIdentificationEstimationVaccine2026}; infections were then resolved to mutually exclusive outcomes via multinomial sampling. Clinical testing was drawn as $\Pr(T=1\mid V, X, U, \text{symptomatic})=\mathrm{logit}^{-1}(1.0 + 0.1X + 0.3U)\times\mathbbm{1}\{\text{symptomatic}\}$, where only symptomatic individuals can be tested. The shared $U$ coefficient across pathogens ensures equi-confounding in the testing process (equi-selection). This yielded roughly 300 focal cases and 1{,}300--1{,}700 controls per replicate.

\subsection{Scenarios}
We evaluated nine scenarios. Scenarios 1--4 correspond to the violations in Figure~\ref{fig:violations}; Scenarios 5--7 stress-test additional identification issues formalised in Appendix~\ref{sec:app_sims}.

\begin{enumerate}
\item \textbf{Baseline (idealized validity):} all control pathogens satisfy vaccine irrelevance, share the focal pathogen's confounding structure, and have comparable testing processes.
\item \textbf{Vaccine irrelevance violation:} the focal vaccine has 30\% cross-protective effectiveness against control pathogen 2 ($\beta_V^{(2)} = \log(0.7)$).
\item \textbf{Other-vaccine entanglement:} control pathogen 2 has its own 80\% effective vaccine $V_2$ strongly correlated with $V$ through shared dependence on $U$, $\Pr(V_2=1\mid X,U)=\mathrm{logit}^{-1}(0.3X + 2.0U)$.
\item \textbf{At-home testing:} individuals with control pathogen 2 may use at-home rapid tests with differential uptake by vaccination status (70\% vs.\ 40\%); positive at-home tests reduce clinical testing probability by 80\%.
\item \textbf{Pan-negative controls} (5a, valid; 5b, biased): pan-negative episodes are added to the control pool; under 5b the focal vaccine also lowers pan-negative episode rates, violating exchangeability (Appendix~\ref{sec:app_panneg_sim}).
\item \textbf{Co-detections} (6a, control--control; 6b, focal--control): the multiplex panel may report multiple positives per episode, either among controls only (6a) or between the focal pathogen and a control (6b); both scenarios additionally evaluate a modified-Poisson difference-in-differences estimator (Appendix~\ref{sec:app_coinf_sim}).
\item \textbf{Imperfect multiplex testing:} the focal and control assays have pathogen-specific sensitivities below one; we compare pooled estimators with and without pan-negatives in the control pool (Appendix~\ref{sec:app_imperfect_sim}).
\end{enumerate}

\subsection{Estimators and performance measures}
Among tested episodes ($T=1$), the \textit{pooled} estimator regresses case status on $V$ and $X$ using all $K$ control pathogens, and the \textit{screened} estimator does so after excluding the problematic pathogen (pathogen 2 in Scenarios 2--4; pan-negative episodes in Scenario 7). The two coincide in Scenarios 1, 5, and 6. Scenario 6 additionally reports the pairwise difference-in-differences estimator described in Appendix~\ref{sec:app_coinf_sim}. We ran 1{,}000 replicates per scenario and report bias, RMSE, and 95\% Wald CI coverage (robust SE). We also computed pathogen-specific estimates $\{\widehat{VE}_k\}_{k=1}^K$ and a Wald test for their homogeneity with cluster-robust standard errors (clustering on individuals to account for the shared case group); see Appendix~\ref{sec:app_pooling} for details.

\subsection{Results}
\input{\resultspath tables/simulation_results.tex}

Table~\ref{tab:sim_results} summarises results across all scenarios. In the baseline, both estimators are essentially unbiased (mean VE 49.5\%, coverage 95.7\%) and the homogeneity test rejects in 4.9\% of replicates, consistent with correct size (Table~\ref{tab:homogeneity}). Scenarios 2 and 4 produce the largest pooled bias ($-3.1$ and $-2.6$~pp; coverage 93\%); excluding pathogen 2 restores near-nominal performance in both ($|\text{bias}|\leq 0.7$~pp, coverage $\geq 94\%$). Scenario 3 (other-vaccine entanglement) yields a more modest pooled bias of $-1.2$~pp because entanglement operates through only one of five control pathogens. The homogeneity test detects these violations in 52\%, 17\%, and 48\% of replicates in Scenarios 2, 3, and 4 respectively. Figure~\ref{fig:pathogen_specific} shows pathogen-specific VE estimates, illustrating how the problematic pathogen (pathogen 2) exhibits systematically different estimates under each violation scenario, providing a diagnostic signal that could guide exclusion decisions in applied settings.

Results for the additional Scenarios 5--7 confirm the corresponding identification results derived in Appendix~\ref{sec:app_sims}.
\begin{figure}[ht]
\centering
\includegraphics[width=0.9\textwidth]{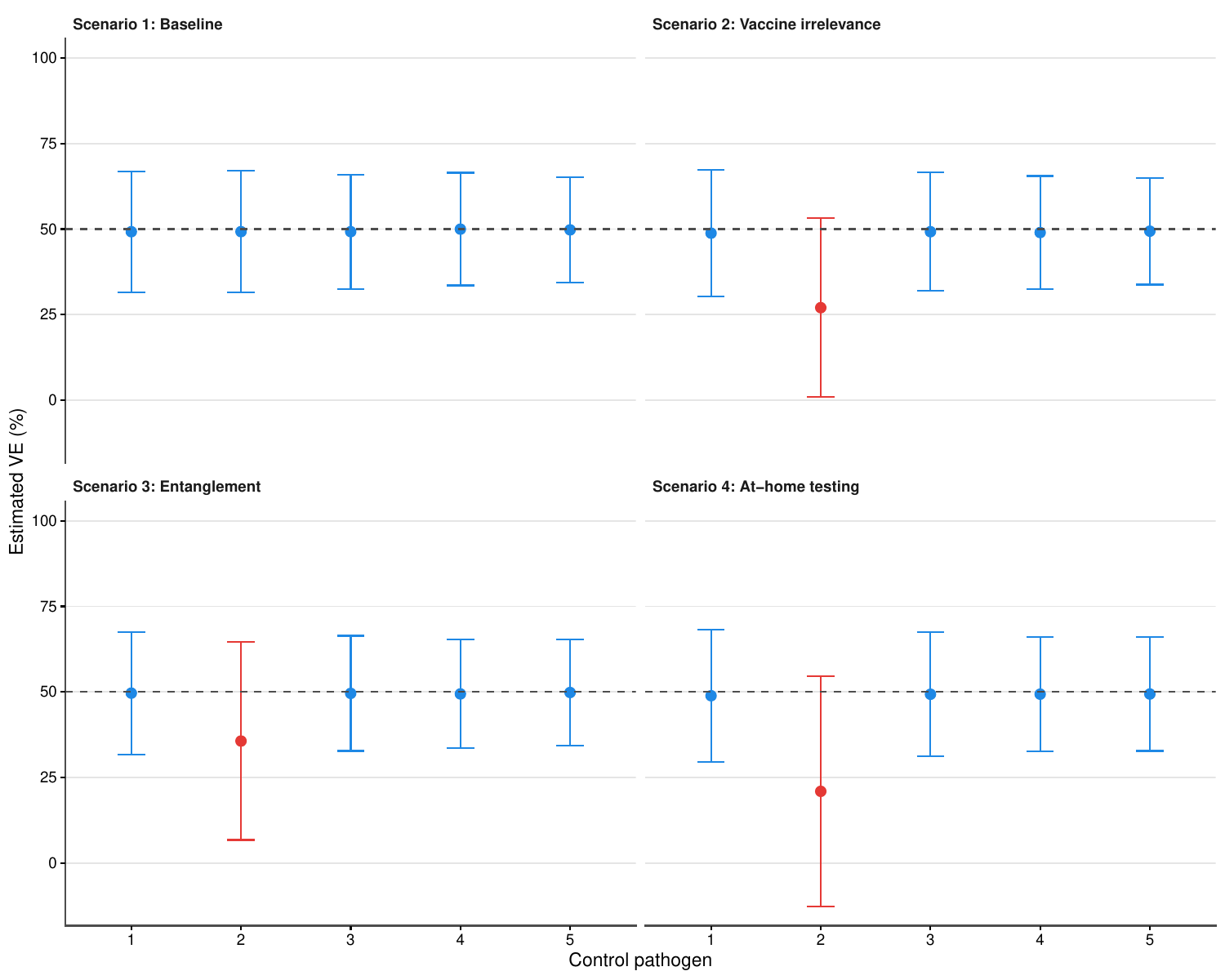}
\caption{Pathogen-specific vaccine effectiveness estimates across scenarios. Pathogen 2 (highlighted) exhibits systematically biased estimates under each assumption violation scenario (2--4), while remaining consistent with other pathogens under the baseline scenario. Dashed horizontal line indicates true VE of 50\%.}
\label{fig:pathogen_specific}
\end{figure}

\section{Discussion}\label{sec:discussion}
Multiplex respiratory testing enables more purposeful TND control selection through pathogen screening, yet no formal guidance exists for it. We introduce a framework grounded in the control-dependent identifying assumptions of the TND---most importantly that vaccination does not affect the control illness and that unmeasured confounding relates focal and non-focal illness in a structured way \citep{sullivanTheoreticalBasisTestNegative2016,lewnardMeasurementVaccineDirect2018,boyerIdentificationEstimationVaccine2026}. The accompanying DAGs (Figures~\ref{fig:tnd-dags}--\ref{fig:violations}) make explicit which pathways motivate exclusion or adjustment and yield three practical principles: exclude pathogens affected by the focal vaccine; exclude or adjust for pathogens with their own vaccines; and account for pathogen-specific testing pathways.

Our simulations show that violations concentrated in even a single control pathogen can substantially bias pooled TND estimates: when the focal vaccine violated vaccine irrelevance for one of five controls, or when at-home testing induced differential selection by vaccination status, the pooled estimator was biased with below-nominal coverage (Table~\ref{tab:sim_results}); excluding the offending pathogen restored near-nominal performance in both cases. The Wald homogeneity test detected these violations as a diagnostic but requires adequate sample size and careful choice of decision thresholds.

Pathogen-specific selection introduces a bias--variance tradeoff: excluding controls reduces bias when violations are present but shrinks the effective sample size. In our simulations, dropping one pathogen of five retained precision comparable to the pooled estimator, though aggressive exclusion could materially inflate variance when case counts are limited. Combining the homogeneity test with the pairwise diagnostic (Appendix~\ref{sec:app_pooling}) offers a data-driven way to identify which pathogen to drop in large samples, but pre-specifying the control set on biological grounds is generally preferable. We provide a step-by-step workflow in Appendix~\ref{sec:app_workflow}.

Our framework complements recent work on heterogeneous reasons for testing. \textcite{yuTestNegativeDesignsVarious2026} show that pooling individuals tested for different reasons (symptoms, mandatory screening, contact tracing) can bias VE because each pathway involves a different selection mechanism and estimand. This heterogeneity operates at the population level and ours at the pathogen level, but the two interact: a pathogen detected mainly through mandatory screening (e.g., workplace surveillance) rather than symptom-driven testing can enrich controls with a differently confounded population. When testing indication is recorded, investigators should consider stratifying by both pathogen and reason for testing.

Our work also connects to recently proposed ``platform TNDs'' that use a common surveillance source and universal multiplex testing to evaluate multiple vaccines at once \citep{chuaUseTestnegativeControls2020a}, where each vaccine's controls are the test-negatives for pathogens it does not target; our principles apply directly to selecting each arm's control pathogens.

Several issues merit further attention. First, multiplex panels vary in sensitivity, specificity, and composition across sites and time, and this heterogeneity may interact with control selection and should be documented. Second, the framework relies on substantive knowledge (e.g., plausibility of cross-protection or vaccine correlation) that may be uncertain, so sensitivity-analysis frameworks tailored to multiplex controls are an important next step. Third, equi-confounding, while a tractable identification route, is a strong assumption that may be implausible in many settings; developing multiplex-informed diagnostics for its violation is a promising direction. It is closely related to proximal causal inference \citep{tchetgenIntroductionProximalCausal2024}, which formalizes when proxy variables can recover causal effects under more general unmeasured confounding structures, and explicit proximal--TND connections exist \citep{liDoubleNegativeControl2023}; extending them to the multiplex setting, where each pathogen-specific control may satisfy proximal conditions to differing degrees, is a natural next step. Finally, the practical value of multiplex-informed selection depends on the prevalence and severity of violations in real data, so empirical evaluation on large multiplex surveillance databases would complement our theory and simulations.



\newpage 
\printbibliography

\newpage 

\appendix

\section{Appendix: Identification results for the conditional risk ratio via the test-negative odds ratio}
\numberwithin{equation}{section} 
\setcounter{equation}{0}         
\renewcommand{\theequation}{\thesection.\arabic{equation}}

\numberwithin{figure}{section} 
\setcounter{figure}{0}         
\renewcommand{\thefigure}{\thesection\arabic{figure}}

\numberwithin{table}{section} 
\setcounter{table}{0}         
\renewcommand{\thetable}{\thesection\arabic{table}}

\subsection{Notation and sampling}
Recall the observed outcome $Y\in\{-1,0,1\}$ where $Y=1$ denotes test-positive illness ($I_+=1$ and $T=1$), $Y=-1$ denotes test-negative illness ($I_-=1$ and $T=1$), and $Y=0$ denotes no recorded test outcome ($T=0$). For each $v\in\{0,1\}$, let $Y^{v}$ denote the counterfactual outcome under intervention that sets $V = v$. Throughout the appendix we abbreviate $Y^{v=0}$ and $Y^{v=1}$ (used in the main text) as $Y^{0}$ and $Y^{1}$; the two notations refer to the same potential outcomes.

The test-negative odds ratio among the tested is
\begin{equation}\label{eq:app_or_t1}
OR_{T=1}(X)\equiv
\frac{\Pr(Y=1\mid T=1,V=1,X)\,\Pr(Y=-1\mid T=1,V=0,X)}
{\Pr(Y=-1\mid T=1,V=1,X)\,\Pr(Y=1\mid T=1,V=0,X)}.
\end{equation}

\paragraph{A useful algebraic identity.}
Because $T=1$ implies $Y\in\{-1,1\}$, within strata of $(V,X)$ we have
\begin{equation*}\label{eq:app_odds_invariance}
\frac{\Pr(Y=1\mid T=1,V=v,X)}{\Pr(Y=-1\mid T=1,V=v,X)}
=
\frac{\Pr(Y=1\mid V=v,X)}{\Pr(Y=-1\mid V=v,X)} \qquad \text{for } v\in\{0,1\},
\end{equation*}
since
$\Pr(Y=y\mid T=1,V=v,X)=\Pr(Y=y\mid V=v,X)/\Pr(T=1\mid V=v,X)$ for $y\in\{-1,1\}$,
and the normalizing factor cancels in the ratio.
Therefore,
\begin{equation}\label{eq:app_or_as_uncond}
OR_{T=1}(X)
=
\frac{
\Pr(Y=1\mid V=1,X)/\Pr(Y=-1\mid V=1,X)
}{
\Pr(Y=1\mid V=0,X)/\Pr(Y=-1\mid V=0,X)
}.
\end{equation}

\subsection{Identification under no unmeasured confounding (Section 2.4)}\label{sec:app_identification_1}
The target estimand in the main text is the conditional causal risk ratio
\begin{equation}\label{eq:app_target_rr}
RR(X)\equiv \frac{\Pr(Y^{1}=1\mid X)}{\Pr(Y^{0}=1\mid X)}.
\end{equation}
We state assumptions in a form compatible with the simplified DAG in Figure~\ref{fig:base-dag}.

\begin{enumerate}[label=(A\arabic*),ref=A\arabic*]
\item \label{ass:A1} \textit{Consistency}: If $V=v$, then $Y=Y^{v}$ almost surely.

\item \label{ass:A2} \textit{No unmeasured confounding}: For the medically attended illness outcome,
\begin{equation*}\label{eq:app_exchangeability}
Y^{v}\ \perp\!\!\!\perp\ V \mid X \qquad \text{for } v\in\{0,1\}.
\end{equation*}

\item \label{ass:A3} \textit{No causal effect on test-negative illness}:
\begin{equation*}\label{eq:app_noeffect_controls}
\Pr(Y^{1}=-1\mid X)=\Pr(Y^{0}=-1\mid X),
\end{equation*}
i.e., intervening on $V$ does not change the (counterfactual) risk of test-negative illness.

\item \label{ass:A4} \textit{Positivity}: For all $x$ in the support of $X$, $\Pr(V=v\mid X=x)>0$ for $v\in\{0,1\}$.
\end{enumerate}

Alternatively, \textcite{jiangTNDDREfficientDoubly2023} state the assumptions in terms of conditional exchangeability with respect to the test-positive outcome and ``control exchangeability'' with respect to the (observed) test-negative outcome, where the latter embeds the no effect on test-negative illness assumption and makes the connection to ``random sampling'' of controls explicit, i.e.
\begin{enumerate}[label=(A\arabic*$^*$),ref=A\arabic*$^*$,start=2]
  \item \label{ass:A2star} \textit{Conditional exchangeability for the test-positive outcome}: $$\mathbbm{1}(Y^{v} = 1)\ \perp\!\!\!\perp\ V \mid X \qquad \text{for } v\in\{0,1\}.$$
  \item \label{ass:A3star} \textit{Control exchangeability for the test-negative outcome}: $$\mathbbm{1}(Y = -1)\ \perp\!\!\!\perp\ V \mid X \qquad \text{for } v\in\{0,1\}.$$
\end{enumerate}
We note that, under Assumption \ref{ass:A3star}, we have 
\begin{equation*}
  \Pr(V = v \mid Y = -1, X) = \Pr(V = v \mid X) \quad \text{ for } v\in\{0,1\},
\end{equation*}
which is the formal statement of the idea that the test-negative controls are effectively a conditionally random sample of vaccination status (or unbiased estimate of the propensity score) in the source population as claimed in Section~\ref{sec:two_views} of the main text.

\begin{proposition}\label{prop:identification_nouc}
Under Assumptions \ref{ass:A1}--\ref{ass:A4}, the conditional causal risk ratio $RR(X)$ is identified by the TND odds ratio $OR_{T=1}(X)$ using the observed data $(Y, V, X, T = 1)$ among the tested, i.e. 
\begin{equation*}\label{eq:app_id_nouc}
RR(X) = OR_{T=1}(X).
\end{equation*}
\end{proposition}

\begin{proof}
Starting with $OR_{T=1}(X)$. By the identity in \eqref{eq:app_or_as_uncond} and consistency (Assumption~\ref{ass:A1}),
\begin{align*}
OR_{T=1}(X)
&=
\frac{
\Pr(Y=1\mid V=1,X)/\Pr(Y=-1\mid V=1,X)
}{
\Pr(Y=1\mid V=0,X)/\Pr(Y=-1\mid V=0,X)
}\\
&=
\frac{
\Pr(Y^{1}=1\mid V=1,X)/\Pr(Y^{1}=-1\mid V=1,X)
}{
\Pr(Y^{0}=1\mid V=0,X)/\Pr(Y^{0}=-1\mid V=0,X)
}.
\end{align*}
By exchangeability (Assumption~\ref{ass:A2}),
$\Pr(Y^{v}=y\mid V=v,X)=\Pr(Y^{v}=y\mid X)$ for $y\in\{-1,1\}$ and $v\in\{0,1\}$.
Thus
\[
OR_{T=1}(X)
=
\frac{
\Pr(Y^{1}=1\mid X)/\Pr(Y^{1}=-1\mid X)
}{
\Pr(Y^{0}=1\mid X)/\Pr(Y^{0}=-1\mid X)
}.
\]
Finally, by the no-effect-on-controls restriction (Assumption~\ref{ass:A3}),
$\Pr(Y^{1}=-1\mid X)=\Pr(Y^{0}=-1\mid X)$ so these terms cancel, giving
\[
OR_{T=1}(X)=\frac{\Pr(Y^{1}=1\mid X)}{\Pr(Y^{0}=1\mid X)}=RR(X),
\]
which proves \eqref{eq:app_id_nouc}.
\end{proof}

\subsection{Identification under equi-confounding (Section 2.5)}\label{sec:app_identification_2}
This subsection formalizes the weaker alternative in Section~\ref{sec:identification_2}, in which
test-negative illness is used as a negative control outcome and an ``equi-confounding'' restriction
replaces full exchangeability. However, we must first modify the target estimand. 

\paragraph{Target estimand under equi-confounding.}
Define the conditional risk ratio among the vaccinated,
\begin{equation}\label{eq:app_rr_vacc}
RR_V(X)\equiv \frac{\Pr(Y^{1}=1\mid V=1,X)}{\Pr(Y^{0}=1\mid V=1,X)}.
\end{equation}
As noted in the main text, if $X$ is sufficiently rich to eliminate relevant effect modification, then $RR_V(X)=RR(X)$.

We state assumptions in a form compatible with the simplified DAG in Figure~\ref{fig:equiconfounding}.

\begin{enumerate}[label=(A\arabic*),ref=A\arabic*,resume]
\item \label{ass:A5} \textit{Consistency}: If $V=v$, then $Y=Y^{v}$ almost surely. Same as Assumption~\ref{ass:A1}.

\item \label{ass:A6} \textit{Equi-confounding on the multiplicative scale}: For all $x$ in the support of $X$,
\begin{equation}\label{eq:app_equiconf}
\frac{\Pr(Y^{0}=1\mid V=1,X=x)}{\Pr(Y^{0}=1\mid V=0,X=x)}
=
\frac{\Pr(Y^{0}=-1\mid V=1,X=x)}{\Pr(Y^{0}=-1\mid V=0,X=x)}.
\end{equation}

\item \label{ass:A7} \textit{No causal effect on test-negative illness}: The restriction on the focal vaccine holds within vaccination strata:
\begin{equation}\label{eq:app_noeffect_controls_strata}
\Pr(Y^{1}=-1\mid V=1,X)=\Pr(Y^{0}=-1\mid V=1,X).
\end{equation}

\item \label{ass:A8} \textit{Overlap}: For $y\in\{-1,1\}$ and $v\in\{0,1\}$, let $\mathcal{S}_y(v)$ denote the support of the joint law of $(Y^0=y,V=v,X)$.
Assume that, for $v\in\{0,1\}$,
\begin{equation}\label{eq:app_overlap_support_main_alt}
\mathcal{S}_{1}(1)\subseteq \mathcal{S}_{1}(0)
\qquad\text{and}\qquad
\mathcal{S}_{1}(v)\subseteq \mathcal{S}_{-1}(v).
\end{equation}
Equivalently, whenever $\Pr(Y^0=1,V=1,X=x)>0$, it also holds that
$\Pr(Y^0=1,V=0,X=x)>0$ and $\Pr(Y^0=-1,V=v,X=x)>0$ for the relevant vaccination status $v$.
This condition ensures that the conditional odds defining $OR_{T=1}(X)$ are well-defined on the support of the target
estimand $RR_V(X)$.
\end{enumerate}

\begin{proposition}\label{prop:identification_equiconf}
Under Assumptions \ref{ass:A5}--\ref{ass:A8}, the conditional causal risk ratio among the vaccinated $RR_V(X)$ is identified by the TND odds ratio $OR_{T=1}(X)$ using the observed data $(Y, V, X, T = 1)$ among the tested, i.e.
\begin{equation}\label{eq:app_id_equiconf}
OR_{T=1}(X)=RR_V(X).
\end{equation}
\end{proposition}

\begin{proof}
Starting with $OR_{T=1}(X)$. By the identity in \eqref{eq:app_or_as_uncond} and consistency (Assumption~\ref{ass:A5}):
\[
OR_{T=1}(X)
=
\frac{
\Pr(Y^{1}=1\mid V=1,X)/\Pr(Y^{1}=-1\mid V=1,X)
}{
\Pr(Y^{0}=1\mid V=0,X)/\Pr(Y^{0}=-1\mid V=0,X)
}.
\]
By the no-effect-on-controls restriction (Assumption~\ref{ass:A7}),
$$\Pr(Y^{1}=-1\mid V=1,X)=\Pr(Y^{0}=-1\mid V=1,X),$$
so
\[
OR_{T=1}(X)
=
\frac{
\Pr(Y^{1}=1\mid V=1,X)/\Pr(Y^{0}=-1\mid V=1,X)
}{
\Pr(Y^{0}=1\mid V=0,X)/\Pr(Y^{0}=-1\mid V=0,X)
}.
\]
Now apply the equi-confounding restriction (Assumption~\ref{ass:A6}):
\[
\frac{\Pr(Y^{0}=-1\mid V=1,X)}{\Pr(Y^{0}=-1\mid V=0,X)}
=
\frac{\Pr(Y^{0}=1\mid V=1,X)}{\Pr(Y^{0}=1\mid V=0,X)}.
\]
Substituting into the previous display yields cancellation of $\Pr(Y^{0}=1\mid V=0,X)$ and gives
\[
OR_{T=1}(X)
=
\frac{\Pr(Y^{1}=1\mid V=1,X)}{\Pr(Y^{0}=1\mid V=1,X)}
=
RR_V(X),
\]
which proves \eqref{eq:app_id_equiconf}.
\end{proof}

\subsection{Extension to the marginal risk ratio}\label{sec:app_marginal}
The identification results above apply to the conditional risk ratios $RR(X)$ or $RR_V(X)$, but can be extended to the marginal (population-average) risk ratio
\[
\overline{RR} \;\equiv\; \frac{\Pr(Y^{1}=1)}{\Pr(Y^{0}=1)},
\]
or the marginal risk ratio among the vaccinated,
\[
\overline{RR}_V = \Pr(Y^{1}=1\mid V=1)/\Pr(Y^{0}=1\mid V=1),
\]
under the same assumptions. However, these are no longer identified by the standard TND odds ratio estimator $OR_{T=1}(X)$ (except in the trivial case where there is no effect modification). Instead, alternative identifying expressions are required, as discussed in \textcite{schnitzerEstimandsEstimationCOVID192022} and \textcite{boyerIdentificationEstimationVaccine2026}. 

Despite this, because our multiplex-informed control selection principles operate at the $X$-conditional level, they apply equally to any weighted average of the conditional estimand, including the marginal risk ratios. Thus, even when the target estimand is a marginal risk ratio rather than a conditional risk ratio, multiplex-informed control selection can still be used to protect identifying assumptions and mitigate bias.

\section{Appendix: Pathogen-specific identification, estimation, pooling, and homogeneity testing}\label{sec:app_pooling}

This appendix formally defines the pathogen-specific VE estimators arising from multiplex TND data, describes methods for combining them into a pooled estimate, and presents tests of homogeneity that can serve as a diagnostic for assumption violations.

\subsection{Setup}
Consider a TND sample of tested individuals ($T=1$) in which a multiplex assay classifies each episode as focal-pathogen positive ($Y=1$) or as positive for one of $K$ non-focal pathogens ($Y=-k$, $k=1,\ldots,K$). Let $\mathcal{S} \subseteq \{1,\ldots,K\}$ denote the set of control pathogens retained after the screening steps described in the main text, with $|\mathcal{S}| = K^*$. We maintain mutual exclusivity of infections for now although discuss extensions to co-infection as well as possible inclusion of pan-negatives in Section~\ref{sec:app_codetection}. 

For each retained control pathogen $k \in \mathcal{S}$, define the pathogen-specific log odds ratio
\begin{equation}\label{eq:beta_k}
\beta_k \;=\; \log OR_k \;=\; \log \frac{\Pr(Y=1 \mid T=1, V=1, X)\;\Pr(Y=-k \mid T=1, V=0, X)}{\Pr(Y=-k \mid T=1, V=1, X)\;\Pr(Y=1 \mid T=1, V=0, X)},
\end{equation}
with corresponding vaccine effectiveness $VE_k = 1 - \exp(\beta_k)$. Note that, under odds ratio invariance, this is equivalent to 
\begin{equation*}
\beta_k \;=\; \log \frac{\Pr(Y=1 \mid Y \in \{1,-k\}, V=1, X)\;\Pr(Y=-k \mid Y \in \{1,-k\}, V=0, X)}{\Pr(Y=-k \mid Y \in \{1,-k\}, V=1, X)\;\Pr(Y=1 \mid Y \in \{1,-k\}, V=0, X)}.
\end{equation*}

\subsection{Pathogen-specific identification.}
The identification results of Appendix~\ref{sec:app_identification_1} (no unmeasured confounding) and Appendix~\ref{sec:app_identification_2} (equi-confounding) extend directly to each retained pathogen $k\in\mathcal{S}$, with $Y=-1$ replaced by $Y=-k$ throughout. Specifically, under the pathogen-specific versions of consistency, positivity/overlap, and the no-causal-effect restriction
\begin{equation}\label{eq:no_causal_effect_k}
\Pr(Y^{1}=-k\mid V=1,X)=\Pr(Y^{0}=-k\mid V=1,X),
\end{equation}
together with the pathogen-specific equi-confounding restriction
\begin{equation}\label{eq:equiconf_k}
\frac{\Pr(Y^{0}=1\mid V=1,X)}{\Pr(Y^{0}=1\mid V=0,X)} \;=\; \frac{\Pr(Y^{0}=-k\mid V=1,X)}{\Pr(Y^{0}=-k\mid V=0,X)},
\end{equation}
the same argument used to prove Proposition~\ref{prop:identification_equiconf} gives $OR_k(X) = RR_V(X)$. Under the stronger no-unmeasured-confounding view (Appendix~\ref{sec:app_identification_1}), the same logic yields $OR_k(X) = RR(X)$. Crucially, the NCO restriction \eqref{eq:no_causal_effect_k} and the pathogen-specific equi-confounding restriction \eqref{eq:equiconf_k} may hold for some $k$ and fail for others, which is precisely what justifies pre-specified screening of the control set $\mathcal{S}$ and motivates the homogeneity diagnostic developed below.

\subsection{Pathogen-specific estimators}

In practice, $\beta_k$ is estimated by logistic regression on the subset containing all focal-pathogen cases and controls from pathogen $k$:
\begin{equation}\label{eq:logistic_k}
\mathrm{logit}\;\Pr(Y=1 \mid V, X,\, Y \in \{1,-k\}) \;=\; \alpha_k + \beta_k V + \bm{\gamma}_k^\top X.
\end{equation}
Let $\hat{\beta}_k$ and $\hat{\sigma}_k^2 = \widehat{\mathrm{Var}}(\hat{\beta}_k)$ denote the maximum likelihood estimate and its estimated variance (from the model-based or sandwich variance estimator). 

\subsection{Pooled estimators}

\paragraph{Standard pooled estimator.}
The conventional TND estimator pools all $K^*$ retained control pathogens into a single control group:
\begin{equation}\label{eq:pooled}
\mathrm{logit}\;\Pr\bigl(Y=1 \mid V, X,\, Y \in \{1\} \cup \{-k : k \in \mathcal{S}\}\bigr) \;=\; \alpha_{\text{pool}} + \beta_{\text{pool}} V + \bm{\gamma}_{\text{pool}}^\top X.
\end{equation}
This estimator treats all retained controls as exchangeable; it is efficient when the identifying assumptions hold uniformly across $\mathcal{S}$, but it can be biased when assumptions hold for some pathogens but not others (as demonstrated in our simulation study).

\paragraph{Decomposition of the pooled odds ratio.}
To see why violations concentrated in a single pathogen can contaminate the pooled estimator, consider the population-level pooled odds ratio with covariates suppressed for clarity. Let $a_v = \Pr(Y=1\mid V=v,T=1)$ denote the case probability among the tested at vaccination level $v$ and $c_{k,v} = \Pr(Y=-k\mid V=v,T=1)$ the corresponding pathogen-$k$ control probability. Then $OR_k = (a_1 c_{k,0})/(a_0 c_{k,1})$ and the pooled odds ratio is
\[
OR_{T=1} \;=\; \frac{a_1\sum_{k\in\mathcal{S}} c_{k,0}}{a_0\sum_{k\in\mathcal{S}} c_{k,1}}.
\]
Rearranging via $a_1 c_{k,0} = OR_k \cdot a_0 c_{k,1}$ gives the exact decomposition
\begin{equation}\label{eq:pooled_decomp}
OR_{T=1} \;=\; \sum_{k\in\mathcal{S}} w_k\, OR_k, \qquad w_k \;=\; \frac{c_{k,1}}{\sum_{j\in\mathcal{S}} c_{j,1}},
\end{equation}
i.e., the pooled odds ratio is exactly a convex combination of the pathogen-specific odds ratios, with weights equal to each pathogen's share of the \emph{vaccinated} control pool. When all $OR_k$ are equal, $OR_{T=1}$ equals the common value; when one pathogen is contaminated (say $OR_2 \neq RR_V$), its contribution to $OR_{T=1}$ scales with $w_2$, making pooled-estimator bias proportional to the prevalence of the violating pathogen. This formalises the prevalence-weighting claim made in the main text and motivates the prevalence sweep in Appendix~\ref{sec:app_sweeps}. With covariate adjustment, the same intuition holds within strata of $X$.

\paragraph{Inverse-variance weighted estimator.}
When pathogen-specific estimates are available, an alternative is to combine them via inverse-variance weighting. Under the assumption that all retained pathogens target the same $\beta \equiv \log OR$, the combined estimator is
\begin{equation}\label{eq:ivw}
\hat{\beta}_{\text{IVW}} \;=\; \frac{\sum_{k \in \mathcal{S}} w_k \,\hat{\beta}_k}{\sum_{k \in \mathcal{S}} w_k}, \qquad w_k = \hat{\sigma}_k^{-2},
\end{equation}
with estimated variance
\begin{equation}\label{eq:ivw_var}
\widehat{\mathrm{Var}}(\hat{\beta}_{\text{IVW}}) \;=\; \left(\sum_{k \in \mathcal{S}} w_k\right)^{-1}.
\end{equation}
This estimator is asymptotically efficient under homogeneity ($\beta_1 = \cdots = \beta_{K^*}$) and independent estimates. In the TND setting, however, pathogen-specific estimates are \emph{not} independent because the case group ($Y=1$) is shared across all contrasts. This positive correlation means that \eqref{eq:ivw_var} understates the true variance, and a corrected version is needed.

\paragraph{Accounting for shared cases.}
To properly account for the correlation induced by shared cases, we construct a stacked dataset in which each case appears $K^*$ times (once per pathogen-specific stratum) while each control appears once. A joint model is fit on this stacked dataset:
\begin{equation}\label{eq:stacked}
\mathrm{logit}\;\Pr(Y=1 \mid V, X, k) \;=\; \alpha_k + \beta V + \bm{\gamma}_k^\top X,
\end{equation}
where pathogen-specific intercepts $\alpha_k$ absorb differences in baseline prevalence across control pathogens, but a common vaccine coefficient $\beta$ is imposed. The variance of $\hat{\beta}$ is estimated using cluster-robust (sandwich) standard errors with clustering on individual, which correctly accounts for the within-person correlation arising from repeated inclusion of cases \citep{liang1986longitudinal}. This approach is analogous to a meta-analysis with correlated effect sizes but is implemented within a single regression framework, avoiding the need to separately estimate the inter-study correlation.

\subsection{Test of homogeneity across control pathogens}
Before pooling, it is useful to test whether the pathogen-specific vaccine effects are homogeneous. Heterogeneity can signal that identifying assumptions are violated for a subset of control pathogens.

\paragraph{Hypotheses.}
\begin{equation}\label{eq:h0}
H_0: \beta_1 = \beta_2 = \cdots = \beta_{K^*} \qquad \text{vs.} \qquad H_1: \beta_j \neq \beta_k \text{ for some } j \neq k.
\end{equation}

\paragraph{Wald test with cluster-robust variance.}
To test \eqref{eq:h0}, we fit an expanded model on the stacked dataset that allows pathogen-specific vaccine effects:
\begin{equation}\label{eq:interaction}
\mathrm{logit}\;\Pr(Y=1 \mid V, X, k) \;=\; \alpha_k + \beta_1 V + \sum_{k=2}^{K^*} \delta_k \,(V \times \mathbbm{1}\{\text{pathogen} = k\}) + \bm{\gamma}_k^\top X.
\end{equation}
Here $\delta_k = \beta_k - \beta_1$ captures the deviation of pathogen $k$'s vaccine effect from pathogen 1 (the reference), and the null hypothesis \eqref{eq:h0} is equivalent to $H_0: \bm{\delta} = \mathbf{0}$, where $\bm{\delta} = (\delta_2, \ldots, \delta_{K^*})^\top$.

Let $\hat{\bm{\delta}}$ denote the vector of estimated interaction coefficients and let $\hat{\bm{\Sigma}}$ denote the corresponding $(K^*-1) \times (K^*-1)$ block of the cluster-robust variance--covariance matrix (clustered on individual). The Wald statistic is
\begin{equation}\label{eq:wald}
W \;=\; \hat{\bm{\delta}}^\top \,\hat{\bm{\Sigma}}^{-1}\, \hat{\bm{\delta}} \;\xrightarrow{d}\; \chi^2_{K^*-1} \quad \text{under } H_0.
\end{equation}
We reject $H_0$ at level $\alpha$ if $W > \chi^2_{K^*-1,\,1-\alpha}$.

\paragraph{Interpretation.}
The cluster-robust variance $\hat{\bm{\Sigma}}$ is essential because the stacked dataset duplicates cases across strata, creating within-individual correlation. Using a naive (model-based) variance instead would overstate the precision of the interaction terms and inflate the type I error rate of the test.

Rejection of homogeneity does not by itself identify \emph{which} pathogen is responsible for the heterogeneity. In practice, investigators can supplement the omnibus test with pairwise comparisons or visual inspection of the $\{\widehat{VE}_k\}$ estimates (as in Figure~\ref{fig:pathogen_specific}) to guide exclusion decisions. Importantly, a failure to reject $H_0$ does not guarantee that all assumptions hold uniformly: the test has limited power when the sample size per pathogen is small or when violations are diffuse across multiple pathogens.

\paragraph{Additional simulation results.}
We implement the homogeneity test across the 7 simulation scenarios described in the main text as a potential data-driven method for identifying which pathogens to screen. Results are shown in Table~\ref{tab:homogeneity}. In Scenario 1, when all pathogens are valid test-negative controls, the test rejects at the nominal rate (4.9\%). In Scenarios 2 through 4, when a single offending pathogen is included among the pooled controls, the homogeneity test has between roughly 15\% and 55\% power to reject, with differences driven by the absolute numbers of controls carrying the offending pathogen as well as the strength of the violation. The remaining scenarios track the identification results: the test retains near-nominal size when all retained controls are valid (Scenario 5a, 5.7\%; Scenarios 6a--6b, $\approx 5\%$), whereas it rejects at elevated rates when a pan-negative control violates vaccine irrelevance (Scenario 5b, 46.5\%) or when heterogeneous assay sensitivity contaminates the pooled control set (Scenario 7, 71.8\%). 

\input{\resultspath tables/homogeneity_test.tex}

\paragraph{Relationship to the combined estimator of \textcite{yuTestNegativeDesignsVarious2026}.}
The inverse-variance weighted estimator in \eqref{eq:ivw} and the stacked-regression estimator in \eqref{eq:stacked} are analogous to the combined stratified estimator proposed by \textcite{yuTestNegativeDesignsVarious2026}, who show that stratified VE estimates from different \emph{reasons for testing} (symptoms, mandatory screening, contact tracing) can be combined when the underlying estimands coincide. Our framework applies the same logic, but instead of stratifying by the reason for testing, we stratify by the \emph{identity of the control pathogen}. The justification is similar---if all retained control pathogens satisfy the same identifying assumptions, each pathogen-specific estimator targets the same VE, and combining them improves precision. The homogeneity test in \eqref{eq:wald} provides an empirical check of this premise, just as \textcite{yuTestNegativeDesignsVarious2026} propose testing equality of their stratified VE parameters before combining. 


\paragraph{Pairwise control--control diagnostic.}
To localize heterogeneity detected by the omnibus Wald test, consider any pair $(j,k)$ of retained control pathogens and define the pairwise vaccination-odds contrast
\begin{equation}\label{eq:or_jk}
OR_{j|k}(X) \;=\; \frac{\Pr(V = 1 \mid Y = -j, X)}{\Pr(V = 0 \mid Y = -j, X)} \,\bigg/\, \frac{\Pr(V = 1 \mid Y = -k, X)}{\Pr(V = 0 \mid Y = -k, X)}.
\end{equation}
Under vaccine irrelevance for both $j$ and $k$ and pathogen-specific equi-confounding (i.e., Equation~\eqref{eq:equiconf_k} applied to each), the multiplicative confounding factor entering numerator and denominator is the same, so $OR_{j|k}(X) = 1$ in expectation for all $X$. A systematic departure of $\widehat{OR}_{j|k}$ from unity, after accounting for sampling variability, flags any of (a) failure of vaccine irrelevance for $j$ or $k$, (b) differential confounding strength between $j$ and $k$, perhaps through some biological or testing-channel interaction, violating equi-confounding.

Operationally, we can fit a logistic regression of $V$ on an indicator for ``$j$ rather than $k$'' (and $X$) within the subset $\{Y \in \{-j,-k\}\}$, and inspect whether the indicator coefficient is significantly non-zero. This can be computed for all $\binom{|\mathcal{S}|}{2}$ retained pairs and arranged as a symmetric matrix; a single control $k^*$ producing $|\log \widehat{OR}_{k^*|k}| > 0$ against most other controls is stronger evidence of a localized violation than scattered noise across the matrix. The pairwise diagnostic complements the omnibus test in two ways: it localizes which controls differ, and it can reveal structure that is obscured in a single omnibus p-value. As with any negative-control diagnostic, $\widehat{OR}_{j|k} \approx 1$ remains compatible with either both controls being valid or both violating assumptions in nearly the same way.

Figure~\ref{fig:pairwise_diagnostic} illustrates the diagnostic in simulation. In the baseline scenario every off-diagonal cell is close to unity, whereas under each violation scenario the row and column corresponding to the offending control (pathogen 2) deviate systematically from one, localizing the violation to a single control.

\begin{figure}[hp]
\centering
\includegraphics[width=0.95\textwidth]{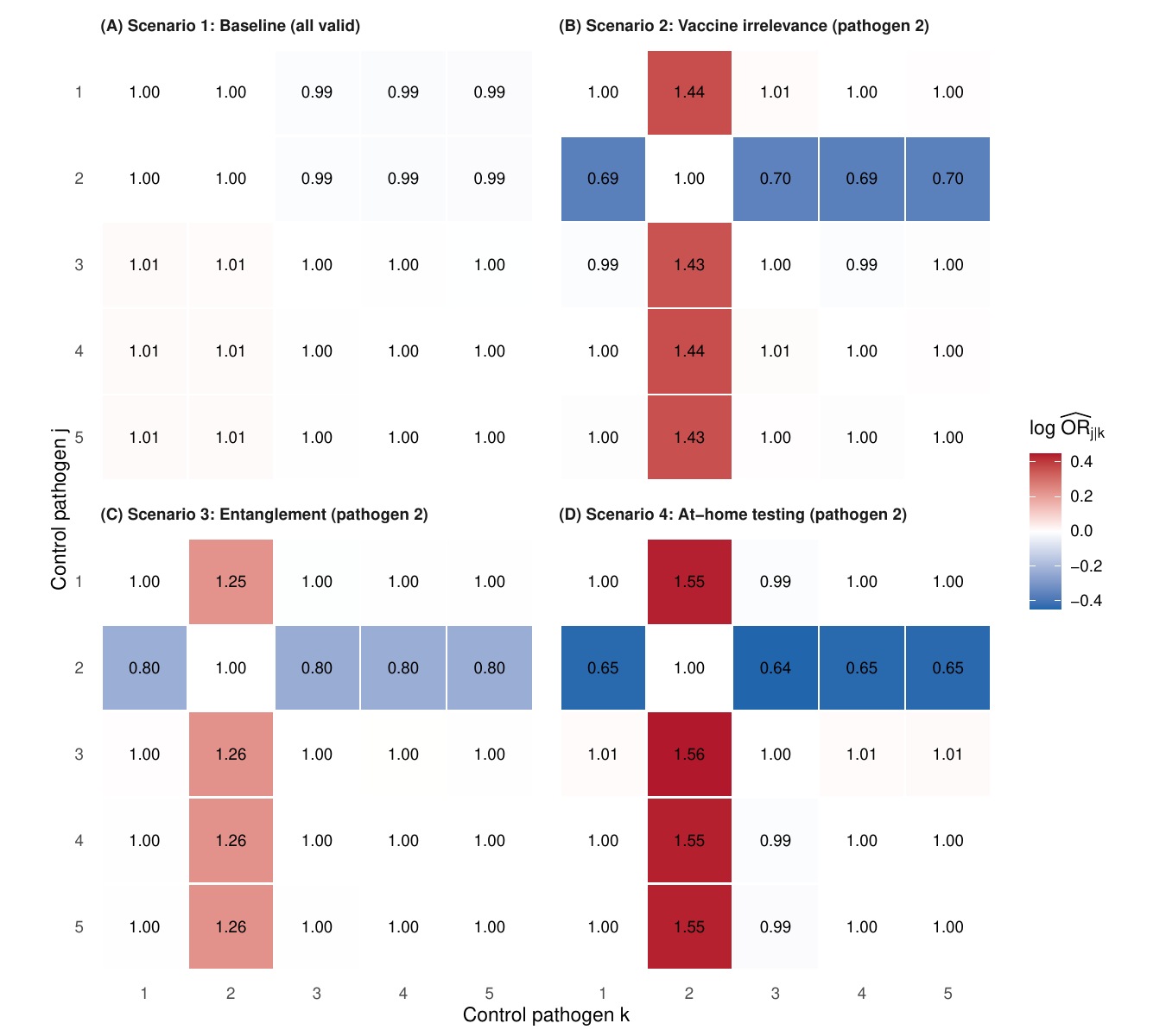}
\caption{Pairwise control--control vaccination-odds contrasts $\widehat{OR}_{j|k}$ (Equation~\ref{eq:or_jk}), averaged across simulation replicates, for the baseline scenario and the three violation scenarios. Each cell reports $\widehat{OR}_{j|k}$ for control pathogens $j$ (row) and $k$ (column); cells are shaded by $\log \widehat{OR}_{j|k}$ (blue $<0$, red $>0$, white $=0$). Under the maintained assumptions all off-diagonal cells are $\approx 1$ (baseline, panel A); under each violation the row/column of the offending control (pathogen 2) departs from unity, localizing the violation.}
\label{fig:pairwise_diagnostic}
\end{figure}

\section{Appendix: Formal notation for co-detections and pan-negative episodes}\label{sec:app_codetection}

This appendix introduces notation that separates pathogen \emph{detection} from illness \emph{attribution}, formalizes co-infections and pan-negative episodes, defines an attribution mapping, and presents a DAG that visualizes the relationships among these quantities.

\subsection{Detection versus attributed illness}

Multiplex PCR can identify the presence of multiple pathogens in a single sample. However, co-detection (i.e., the presence of nucleic acid from two or more pathogens in a single sample) does not, by itself, imply co-infection (i.e., symptomatic illness caused by infection of two or more pathogens in a single host) as pathogen detection can reflect contamination, asymptomatic carriage, or residual RNA from a previous infection. Resolving whether co-detection reflects co-infection requires an intermediate \emph{attribution} step in which each detected pathogen is evaluated for its likelihood of causing the observed symptoms. To reason precisely about co-infections (as well as pan-negative episodes), we now formally separate detection via the raw assay output from attribution.

Let $D_{-k} \in \{0,1\}$ indicate that pathogen $k$ is \emph{detectably present} (that is, sufficient nucleic acid is present to return a positive PCR result) and let $D_+$ be the corresponding indicator for the focal pathogen. Because detectable presence is a biological state (reflecting current or recent infection, asymptomatic carriage, or residual RNA), $D_{-k}$ is defined for all individuals regardless of testing status; the observed multiplex result is the vector $\mathbf{D}_{-} \mid T = 1$, where
$$
\mathbf{D}_{-} = (D_{-1}, \ldots, D_{-K}).
$$
Crucially, $D_{-k} = 1$ does not imply $I_{-k} = 1$: detection may reflect incidental carriage or residual RNA from a previous (possibly asymptomatic) infection rather than clinically meaningful symptomatic illness caused by pathogen $k$. Conversely, $I_{-k} = 1$ requires both detectable presence ($D_{-k} = 1$) and that pathogen $k$ is judged to be the cause of symptoms.

\subsection{Co-detection indicator}

Define the co-detection indicator
\begin{equation}\label{eq:codetect}
C_{-} = \mathbbm{1}\!\left\{\sum_{k=1}^{K} D_{-k} > 1\right\},
\end{equation}
which flags episodes in which more than one non-focal pathogen is detected on the multiplex panel. When $C_{-} = 1$, the mapping from $\mathbf{D}$ to a single attributed illness $I_{-k}$ is ambiguous and requires an explicit attribution rule (see below).

Similarly, we can define a focal--control co-detection indicator
\begin{equation}\label{eq:codetect_focal}
C_+ \;=\; \mathbbm{1}\!\left\{D_+ = 1,\; \textstyle\sum_{k=1}^{K} D_k \geq 1 \right\},
\end{equation}
which flags episodes in which the focal pathogen co-occurs with at least one non-focal pathogen (but possibly more) on the panel. Again when $C_{+} = 1$, the mapping to either a single $I_{-k}$ or $I_{+}$ is ambiguous and requires an explicit attribution rule.

\subsection{Attribution mapping}

Moving from the raw detection vector to pathogen-attributed illness requires both testing and an \emph{attribution rule}. For each pathogen $j \in \{+, -1, \ldots, -K\}$ define the tri-level outcome
\begin{equation}\label{eq:trilevel}
Y_j = \begin{cases}
  +1 & \text{tested and attributed positive for pathogen } j \;(T = 1,\; I_j = 1),\\
  -1 & \text{tested and attributed negative for pathogen } j \;(T = 1,\; I_j = 0),\\
  \phantom{-}0 & \text{not tested } (T = 0),
\end{cases}
\end{equation}
and collect these into the outcome vector $\mathbf{Y} = (Y_+, Y_{-1}, \ldots, Y_{-K})$. Each component shares the tri-level structure of the scalar outcome $Y$ of the main text but is now defined separately for every pathogen on the panel. The attribution rule is the mapping
\begin{equation}\label{eq:attribution}
\alpha:\; \{0,1\}^{K+1} \times \{0,1\} \;\longrightarrow\; \{-1,\, 0,\, +1\}^{K+1}, \qquad \mathbf{Y} = \alpha(D_+, \mathbf{D}_{-}, T),
\end{equation}
which takes the detection vector and testing indicator and maps them to the outcome vector. As discussed in the main text, three common choices of $\alpha$ are assigning the episode with the highest viral load (lowest cycle threshold value), assigning the most consistent with presenting syndrome, or to the pathogen identified through formal clinical adjudication.

This vector representation separates two phenomena:
\begin{itemize}[nosep]
  \item \textbf{Co-detection} is a property of the \emph{detection} vector: $C_{-} = 1$ when $\sum_k D_{-k} > 1$ (Equation~\eqref{eq:codetect}), or $C_{+} = 1$ when the focal pathogen is detected alongside a control (Equation~\eqref{eq:codetect_focal}). It records that the assay returned more than one positive but says nothing about how illness is attributed.
  \item \textbf{Genuine co-infection} is a property of the \emph{outcome} vector: it occurs when the attribution rule assigns positive illness status to more than one pathogen, i.e., $\sum_j \mathbbm{1}\{Y_j = +1\} > 1$. A co-detected episode need not be a genuine co-infection, as an attribution rule may resolve it to a single $Y_j = +1$, but a genuine co-infection is necessarily co-detected.
\end{itemize}

\paragraph{Collapse to the scalar outcome.} When at most one component of $\mathbf{Y}$ equals $+1$ (i.e., when the attributed illnesses are mutually exclusive), the vector collapses without loss of information to the single categorical outcome used elsewhere in the paper:
\begin{equation}\label{eq:collapse}
Y = \begin{cases}
  +1 & Y_+ = 1 \quad (\text{focal case}),\\
  -k & Y_{-k} = 1 \quad (\text{control attributed to pathogen } k),\\
  \phantom{-}0 & T = 0 \quad (\text{all } Y_j = 0),
\end{cases}
\end{equation}
The mutual-exclusivity assumption maintained in Sections~\ref{sec:multiplex_implications}--\ref{sec:three_principles} is precisely the statement that this collapse is possible, so that every result stated for the scalar $Y$ applies unchanged. Genuine co-infection is the case in which the collapse fails and the full vector $\mathbf{Y}$ must be retained; this is the situation analysed in Sections~\ref{sec:app_control_control} and \ref{sec:app_focal_control}.

\subsection{Control--control co-infection}\label{sec:app_control_control}

Consider first genuine \emph{control--control} co-infection: an episode in which two or more non-focal pathogens are attributed positive, $Y_{-j} = 1$ and $Y_{-k} = 1$ for distinct $j, k$, while the focal pathogen is negative ($Y_+ = -1$). The outcome vector $\mathbf{Y}$ does not collapse to a single label, but unlike the focal--control case of Section~\ref{sec:app_focal_control} the focal arm $\Pr(Y_+ = 1 \mid V, X, T = 1)$ of every TND contrast is untouched, so the only question is how such an episode enters the control pool.

\paragraph{Pathogen-specific estimators.} Each $OR_k(X)$ uses the focal arm and the single control arm $\{Y_{-k} = 1\}$, irrespective of whether other controls are also positive in the same episode. Provided every retained control $k \in \mathcal{S}$ satisfies vaccine irrelevance and equi-confounding (Appendix~\ref{sec:app_pooling}), a co-infected episode simply appears in the $\{Y_{-k} = 1\}$ arm of each control it carries, and $OR_k(X) = RR_V(X)$ continues to hold pathogen-by-pathogen.

\paragraph{Pooled estimator.} A co-infected episode contributes to more than one bucket $\{Y_{-k} = 1\}$ at once, so the pooled control sums $\sum_{k \in \mathcal{S}} \Pr(Y_{-k} = 1 \mid V, X)$ count it once per involved control. This re-weighting is immaterial under the common equi-confounding restriction, in which a single ratio $\rho(X)$ relates the $V = 1$ and $V = 0$ control probabilities for every $k \in \mathcal{S}$ (Appendix~\ref{sec:app_pooling}): each term in the numerator and denominator sums---including the doubly counted mass from co-infected episodes---carries the same factor $\rho(X)$, so the ratio of the two pooled-control sums is $\rho(X)$ regardless of how many controls an episode activates. Writing $r_+(X) = \Pr(Y_+ = 1 \mid V = 1, X)/\Pr(Y_+ = 1 \mid V = 0, X)$ for the focal risk ratio,
$$
OR_{\text{pool}}(X) \;=\; \frac{r_+(X)}{\rho(X)} \;=\; RR_V(X),
$$
exactly as in the mutually exclusive case. This is the invariance noted in the main text: any non-excluding resolution of the co-infection---assigning the episode to a single bucket by a $V$-independent rule, or leaving it in several---leaves the $V$-conditional distribution of the pooled control set unchanged.

Control--control co-infection is therefore a precision question, not an identification one, so long as every retained control is valid; it becomes a bias issue only when one of the co-infecting controls itself violates equi-confounding or vaccine irrelevance, in which case the offending pathogen---not the co-infection---is at fault.

\paragraph{Diagnostics.} The formal definitions of the omnibus homogeneity test (Equation~\eqref{eq:wald}) and pairwise control--control diagnostic (Equation~\eqref{eq:or_jk}) are given in Appendix~\ref{sec:app_pooling}. Under explicit detection notation, both are implemented after applying a pre-specified attribution rule $\alpha(D_+,\mathbf{D}_{-},T)$. For the omnibus diagnostic, fit the stacked focal-vs-control model using attributed pathogen categories $\{Y_{-k} = 1: k\in\mathcal{S}\}$ exactly as in Section~\ref{sec:app_pooling}. For pairwise comparisons, restrict to episodes with $(Y_{-j}=1 \text{ or } Y_{-k}=1)$ to define the two-category subset. This re-expression changes bookkeeping, not interpretation: departures from homogeneity or from $OR_{j|k}(X)=1$ still indicate that at least one retained control behaves differently from the others.

\paragraph{An example of when to exclude control--control co-infections} Consider the case where co-infection is mainly a reflection of latent immune status, where individuals with a certain immune profile are more susceptible to being infected with two or more pathogens simultaneously, but this doesn't extend to the focal pathogen. If this latent immune status is also related to vaccination status, it would be a source of non-equi-confounding and therefore invalidate it as potential control. Excluding co-infected episodes effectively conditions on individuals with lower latent immune susceptibility, and potentially mitigates the resulting bias (but also changes estimand). Figure~\ref{fig:exclude_codetect} illustrates this scenario. We could imagine an equivalent scenario where co-infection reflects latent exposure risk, with certain high risk individuals being more likely to encounter multiple pathogens (excepting the focal pathogen) simultaneously. 

\begin{figure}[!ht]
  \centering
\begin{tikzpicture}[> = stealth, shorten > = 1pt, auto, node distance = 2.3cm, inner sep = 0pt, minimum size = 0.5pt, thick]
  \tikzstyle{every state}=[draw=none, fill=none]
  \node[state] (x) {$X$};
  \node[state] (v) [right of=x] {$V$};
  \node[state] (i) [right of=v] {$I_+$};
  \node[state] (t) [right of=i] {$T$};
  \node[state] (i1) [below of=i] {$I_{-1}$};
  \node[state] (i2) [below=-0.25cm of i1] {$I_{-2}$};
  \node[state] (i3) [below=-0.35cm of i2] {$\vdots$};
  \node[state] (i4) [below=-0.25cm of i3] {$I_{-K}$};
  \node[state] (u) [below of=v] {$U$};
  \node[state] (v2) [below=0.9cm of u] {$U_{-}$};

  \path[->] (x) edge node {} (v);
  \path[->] (x) edge [out=45, in=135] node {} (i);
  \path[->] (v) edge node {} (i);

  \path[->] (i) edge node {} (t);
  \path[->] (i1) edge node {} (t);
  \path[->] (i2) edge [dashed] node {} (t);
  \path[->] (i4) edge [dashed] node {} (t);
  \path[->] (x) edge [out=45, in=135] node {} (t);

  \path[->] (u) edge node {} (x);
  \path[->] (u) edge node {} (v);
  \path[->] (v2) edge [dashed, out=135, in=225] node {} (v);
  \path[->] (u) edge node {} (i);
  \path[->] (u) edge node {} (t);
  \path[->] (u) edge node {} (i1);
  \path[->] (u) edge node {} (i2);
  \path[->] (u) edge node {} (i4);

  \path[->] (v2) edge [dashed] node {} (i2);
  \path[->] (v2) edge [dashed] node {} (i4);
\end{tikzpicture}
\caption{An example of co-infection caused by latent immune status or latent exposure risk specific to non-focal pathogens violating equi-confounding. $U_{-}$ represents the latent factor affecting non-focal pathogens that is also related to focal vaccination.}\label{fig:exclude_codetect}
\end{figure}
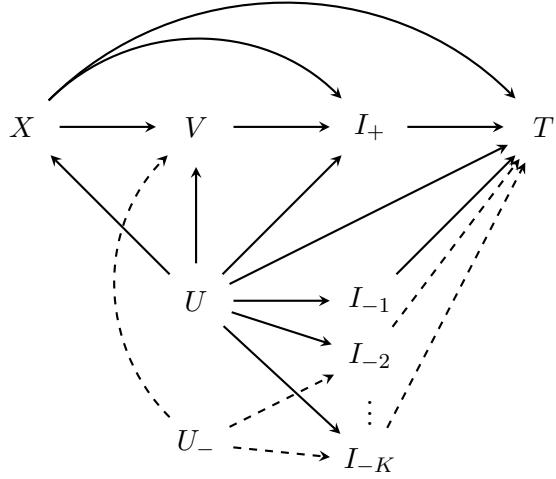

\subsection{Focal--control co-infection}\label{sec:app_focal_control}

Genuine focal--control co-infection is the event in which the focal pathogen and a control pathogen are \emph{both} attributed positive, $Y_+ = 1$ and $Y_{-k} = 1$ for some $k$. When single-plex PCR testing was the norm, such an episode would typically be forced into the case group ($Y = +1$ via Equation~\eqref{eq:collapse}) because that was the only pathogen that could be identified. This is innocuous when focal--control co-infection is rare, but creates a distinct identification problem when it is not, because the depletion mechanism it induces does not rely on any violated assumption: it is driven by the causal effect of $V$ on focal infection itself, which is the estimand of interest. 

\paragraph{Bias mechanism.} Work within the tested population ($T = 1$) and let $\pi_+(v, X) = \Pr(Y_+ = 1 \mid V = v, X)$ and $\pi_k(v, X) = \Pr(Y_{-k} = 1 \mid V = v, X)$ denote the probabilities that the focal pathogen and control pathogen $k$, respectively, are attributed positive; assume the two attributed-positive events are conditionally independent given $(V, X, U)$. Under the default scalar attribution rule, an episode contributes to the control pool for pathogen $k$ only when the focal pathogen is attributed negative and control $k$ positive, i.e., $Y_+ = -1$ and $Y_{-k} = 1$; genuine co-infections ($Y_+ = 1,\; Y_{-k} = 1$) are instead routed to the case group. The probability of contributing to that pool is therefore $(1 - \pi_+(v, X))\,\pi_k(v, X)$, which is smaller for $v = 0$ than for $v = 1$ whenever the focal vaccine reduces focal infection risk ($\pi_+(0, X) > \pi_+(1, X)$). The control pool is thus differentially depleted of unvaccinated episodes, and the conventional pooled TND estimator is biased upward even if vaccine irrelevance and equi-confounding hold for $k$.

\paragraph{Identification of VE under focal--control co-infection.} Restricting attention to a single control pathogen $k$, consider the pairwise subset $\mathcal{P}_k = \{Y_+ = 1 \text{ or } Y_{-k} = 1\}$ of tested episodes, where we now allow $Y_+ = Y_{-k} = 1$ simultaneously. 

Our target is the conditional causal risk ratio among the vaccinated
\begin{equation}
  \mathrm{RR}_{V+}(X) = \dfrac{\Pr(Y^{v=1}_+ = 1 \mid V = 1, X)}{\Pr(Y^{v=0}_+ = 1 \mid V = 1, X)}.
\end{equation}
Define the conditional log risk ratios
\begin{align}
\log \mathrm{RR}_+^{(k)}(X) &= \log \Pr(Y_+ = 1 \mid V = 1, X, \mathcal{P}_k) - \log \Pr(Y_+ = 1 \mid V = 0, X, \mathcal{P}_k), \label{eq:rrplus}\\
\log \mathrm{RR}_k^{(k)}(X) &= \log \Pr(Y_{-k} = 1 \mid V = 1, X, \mathcal{P}_k) - \log \Pr(Y_{-k} = 1 \mid V = 0, X, \mathcal{P}_k), \label{eq:rrk}
\end{align}
and define the difference-in-differences contrast
\begin{equation}\label{eq:did}
\delta_k(X) \;=\; \log \mathrm{RR}_+^{(k)}(X) - \log \mathrm{RR}_k^{(k)}(X).
\end{equation}

Assume:
\begin{enumerate}[label=(C\arabic*),ref=C\arabic*]
  \item \label{ass:C1} \textit{Consistency}: If $V = v$, then $Y_{+} = Y^v$ and $Y_{-k}^v = Y_{-k}$ almost surely.
  \item \label{ass:C2} \textit{Vaccine irrelevance for control $k$}: $Y_{-k}^{v=1} = Y_{-k}^{v=0}$ almost surely, so that any observed $\log\mathrm{RR}_k^{(k)}$ reflects only confounding, not a causal effect of $V$ on $Y_{-k}$.
  \item \label{ass:C3} \textit{Equi-confounding on the multiplicative scale}: there exists a function $\gamma(X)$ such that within $\mathcal{P}_k$,
  $$
  \log \frac{\Pr(Y_+^{v=0} = 1 \mid V = 1, X)}{\Pr(Y_+^{v=0} = 1 \mid V = 0, X)} \;=\; \log \frac{\Pr(Y_{-k}^{v=0} = 1 \mid V = 1, X)}{\Pr(Y_{-k}^{v=0} = 1 \mid V = 0, X)} \;=\; \gamma(X).
  $$
  That is, unmeasured confounding induces the same multiplicative shift in risk for focal and control under the no-vaccine counterfactual. Note that this condition is equivalent to odds ratio equi-confounding 
  $$
  \log \dfrac{\dfrac{\Pr(Y_+^{v=0} = 1 \mid V = 1, X)}{\Pr(Y_+^{v=0} = 0 \mid V = 1, X)}}{\dfrac{\Pr(Y_+^{v=0} = 1 \mid V = 0, X)}{\Pr(Y_+^{v=0} = 0 \mid V = 0, X)}} \;=\; \log \dfrac{\dfrac{\Pr(Y_{-k}^{v=0} = 1 \mid V = 1, X)}{\Pr(Y_{-k}^{v=0} = 0 \mid V = 1, X)}}{\dfrac{\Pr(Y_{-k}^{v=0} = 1 \mid V = 0, X)}{\Pr(Y_{-k}^{v=0} = 0 \mid V = 0, X)}} \;=\; \gamma(X)
  $$
  as $Y_+= 0$ when $Y_{-k}= 0$ as both refer to same untested state. 
  \item \label{ass:C4} \textit{Overlap}: Whenever $\Pr(Y_+^0=1,V=1,X=x)>0$, it also holds that
  $\Pr(Y^0_+=1,V=0,X=x)>0$ and $\Pr(Y^0_{-k}=1,V=v,X=x)>0$ for the relevant vaccination status $v$.
\end{enumerate}

Under Assumptions~\ref{ass:C1}--\ref{ass:C4}, the standardised log-risk-ratios decompose as
\begin{align*}
\log \mathrm{RR}_+^{(k)} &= \log \mathrm{RR}_{V+}(X) + \gamma(X), \\
\log \mathrm{RR}_k^{(k)} &= \gamma(X),
\end{align*}
where $\log \mathrm{RR}_{V+}(X)$ is the causal log-risk-ratio for the focal pathogen. The confounding term $\gamma(X)$ enters both expressions identically and cancels in the difference, giving
$$
\delta_k(X) \;=\; \log \mathrm{RR}_{V+}(X),
$$
so $\delta_k(X)$ identifies the marginal causal log-risk-ratio for the focal pathogen and hence VE on the risk-ratio scale.

Two remarks. First, restricting to the pairwise subset $\mathcal{P}_k$ does not introduce a separate selection bias: the two working models are fit on the same subset, and any constant offset induced by the conditioning enters $\log\mathrm{RR}_+^{(k)}$ and $\log\mathrm{RR}_k^{(k)}$ identically under Assumption~\ref{ass:C2} and therefore cancels in $\delta_k$, similar to algebraic identity in Equation~\ref{eq:app_odds_invariance}. Second, estimation of the difference-in-differences contrast in Equation~\ref{eq:did} is no longer possible via the traditional logistic regression estimator. Instead, one can fit separate working models for the focal and control pathogens as suggested below. Pooling across control pathogens via a stacked common-$\delta$ regression is discussed alongside the simulation results.

\paragraph{Pathogen-specific estimators.}
For each control pathogen $k$, we can estimate $\delta_k(X)$ by fitting separate working models for the focal and control pathogens on the pairwise subset $\mathcal{P}_k$:
\begin{align*}
  \Pr(Y_+ = 1 \mid V, X, \mathcal{P}_k) &= p_+^{(k)}(V, X), \\
  \Pr(Y_{-k} = 1 \mid V, X, \mathcal{P}_k) &= p_k^{(k)}(V, X),
\end{align*}
and then computing the difference-in-differences contrast
\begin{equation*}\hat{\delta}_k(X) = \log \dfrac{\hat{p}_+^{(k)}(1, X)}{\hat{p}_+^{(k)}(0, X)} - \log \dfrac{\hat{p}_k^{(k)}(1, X)}{\hat{p}_k^{(k)}(0, X)}.\end{equation*}
Under Assumptions~\ref{ass:C1}--\ref{ass:C4}, $\hat{\delta}_k(X)$ is a consistent estimator of the causal log-risk-ratio $\log \mathrm{RR}_{V+}(X)$ for the focal pathogen. Standard errors can be obtained via the delta method.

For example, one can fit log binomial regression models for the focal and control pathogens separately on the pairwise subset $\mathcal{P}_k$:
\begin{align*}
\mathrm{log}\;\Pr(Y_+ = 1 \mid V, X, \mathcal{P}_k) &= \alpha_+^{(k)} + \beta_+^{(k)} V + \bm{\gamma}_+^{(k)\top} X, \\
\mathrm{log}\;\Pr(Y_{-k} = 1 \mid V, X, \mathcal{P}_k) &= \alpha_k^{(k)} + \beta_k^{(k)} V + \bm{\gamma}_k^{(k)\top} X,
\end{align*}
in which case the difference-in-differences contrast is computed as $\hat{\delta}_k(X) = \hat{\beta}_+^{(k)} - \hat{\beta}_k^{(k)}$,
because the log-risk-ratio for each model is simply the coefficient on $V$. In practice we fit these two working models jointly as a single arm-stratified log-link (modified) Poisson regression and obtain the standard error of $\hat{\delta}_k(X)$---which must account for the covariance between $\hat{\beta}_+^{(k)}$ and $\hat{\beta}_k^{(k)}$ induced by the shared subset---from a cluster-robust sandwich covariance clustered on individual (Appendix~\ref{sec:app_coinf_sim}).

\paragraph{Pooled estimator.}
To combine information across all retained control pathogens $k \in \mathcal{S}$ while targeting a single focal effect, we stack the per-pathogen pairwise subsets $\mathcal{P}_k$ into one dataset---each contributing a focal arm ($Y_+$) and a control arm ($Y_{-k}$)---and fit a single modified-Poisson working model with pathogen- and arm-specific nuisance terms but a \emph{common} focal-versus-control difference-in-differences parameter $\delta$:
\begin{equation*}
\log \Pr(Y = 1 \mid V, X, k, \mathrm{arm}) = \alpha_{k,\mathrm{arm}} + \bm{\gamma}_{k,\mathrm{arm}}^\top X + \eta_k V + \delta\, \big(V \times \mathbbm{1}\{\mathrm{arm}=+\}\big),
\end{equation*}
where $\eta_k$ is the control-arm log-risk-ratio for pathogen $k$ (pure confounding under vaccine irrelevance) and $\hat\delta$ is the pooled estimate of $\log \mathrm{RR}_{V+}(X)$. Under the maintained assumptions every $\delta_k$ equals the same causal focal log-risk-ratio, so imposing a common $\delta$ is appropriate. The variance of $\hat\delta$ is obtained from a cluster-robust sandwich covariance clustered on individual, which accounts for both the within-pathogen focal/control correlation and the cross-pathogen correlation induced by the shared focal cases. We prefer this stacked estimator to (i) contrasting the focal outcome against the \emph{pooled} control $Y_- = \mathbbm{1}\{\bigcup_k Y_{-k}=1\}$, which can be biased because the log-risk-ratio of a high-prevalence union outcome no longer equals the common confounding offset, and to (ii) inverse-variance weighting of the $\hat\delta_k$, whose variance would understate the strong positive cross-pathogen correlation.

\subsection{Pan-negatives}

Define the pan-negative indicator
\begin{equation}\label{eq:panneg}
P = \mathbbm{1}\!\left\{Y_+ = -1,\; \bigcap_{k=1}^{K^*} Y_{-k}=-1\right\},
\end{equation}
identifying symptomatic, tested individuals for whom neither the focal pathogen nor any non-focal pathogen on the panel is detected (or attributed as pathogenic source after applying attribution rule). As mentioned in the main text, pan-negative results may reflect: 
\begin{enumerate}[font=\bfseries]
  \item \textbf{A pathogen not included in the multiplex panel.} In this case, the illness is caused by an infectious agent that is not tested for (e.g. a fungus, a bacterium or virus not included in the panel, or a novel pathogen), and the pan-negative result is a true negative for all pathogens on the panel.
  \item \textbf{A non-infectious cause of illness.} In this case, symptoms are caused by a non-infectious agent (e.g. an autoimmune disorder, an allergic reaction, or a metabolic condition) that mimics a respiratory infection, and the pan-negative result is a true negative for all pathogens on the panel.
  \item \textbf{A false negative result for one or more pathogens on the panel.} In this case, the illness is caused by an infectious agent that is included in the multiplex panel, but the test fails to detect it due to imperfect sensitivity (e.g. low viral load, improper sample collection, or other technical errors), and the pan-negative result is a false negative for one or more pathogens on the panel.
\end{enumerate}

One can think of pan-negative results as a distinct control category subject to the same conditions for inclusion in the pooled control set as any other control pathogen: i.e. (1) are they a valid negative control, in the sense that they are unaffected by focal vaccine and relevant to source of unmeasured confounding, and (2) do they satisfy equi-confounding. If these conditions hold, pan-negatives can be included in the pooled control set without biasing the VE estimate for the focal pathogen. The problem is that because the precise identity of the causal agent or source of negative result is unknown, it is difficult to evaluate the plausibility of these conditions. Furthermore, pan-negatives may be a heterogeneous mixture of the three scenarios above, with the relative proportions varying across populations, time periods, or other covariates.

\subsection{A DAG for detection, attribution, and boundary cases}

Figure~\ref{fig:detection-dag} extends the multiplex DAG from the main text to show the relationship between the detection vector $\mathbf{D}$, the attribution mapping $\alpha$, the co-detection indicator $C$, the pan-negative indicator $P$, and the observed test result $Y$.

\begin{figure}[!ht]
\centering
\begin{tikzpicture}[
  > = stealth, shorten > = 1pt, auto,
  node distance = 2.2cm,
  inner sep = 2pt, minimum size = 8pt, thick,
  every state/.style={draw=none, fill=none, font=\small},
  obs/.style={draw=black, rounded corners=2pt, inner sep=4pt, font=\small},
  latent/.style={draw=black, dashed, rounded corners=2pt, inner sep=4pt, font=\small}
]
  \node[state] (v) {$V$};
  \node[state] (u) [below=0.5cm of v] {$U,X$};

  \node[state] (dplus) [right=2cm of v] {$D_+$};
  \node[state] (d1) [below=0.5cm of dplus] {$D_{-1}$};
  \node[state] (dk) [below=1cm of d1] {$D_{-K}$};
  \node[state] (vdots2) at ($(d1)!0.5!(dk)$) {$\vdots$};

  \node[state] (iplus) [right=1.5cm of dplus] {$I_+$};
  \node[state] (i1) [below=0.5cm of iplus] {$I_{-1}$};
  \node[state] (ik) [below=1cm of i1] {$I_{-K}$};
  \node[state] (vdots1) at ($(i1)!0.5!(ik)$) {$\vdots$};

  \node[state] (t) [right=1.5cm of i1] {$T$};

  \node[latent] (alpha) [right=0.8cm of t] {$\alpha$};

  \node[state] (y1) [right=1.5cm of alpha] {$Y_{-1}$};
  \node[state] (yplus) [above=0.5cm of y1] {$Y_+$};
  \node[state] (yk) [below=1cm of y1] {$Y_{-K}$};
  \node[state] (vdots3) at ($(y1)!0.5!(yk)$) {$\vdots$};

  \node[state] (p) [below=1.25cm of alpha] {$P$};

  \path[->] (v) edge (dplus);
  \path[->] (v) edge[out=30, in=150] (iplus);
  \path[->] (u) edge (v);
  \path[->] (u) edge (dplus);
  \path[->] (u) edge (d1);
  \path[->] (u) edge (dk);
  \path[->] (u) edge (iplus);
  \path[->] (u) edge[out=-15, in=195] (i1);
  \path[->] (u) edge[out=-15, in=195] (t);
  \path[->] (u) edge (ik);

  \path[->] (dplus) edge (iplus);
  \path[->] (d1) edge (i1);
  \path[->] (dk) edge (ik);

  \path[->] (iplus) edge (t);
  \path[->] (i1) edge (t);
  \path[->] (ik) edge (t);


  \path[->] (t) edge (alpha);
  \path[->] (dplus) edge[out=30, in=150] (alpha);
  \path[->] (d1) edge[out=15, in=165] (alpha);
  \path[->] (dk) edge[out=330, in=210] (alpha);

  \path[->] (alpha) edge (yplus);
  \path[->] (alpha) edge (y1);
  \path[->] (alpha) edge (yk);
  \path[->] (alpha) edge (p);


\end{tikzpicture}
\caption{DAG showing the relationship between detectable pathogen ($D_+, \mathbf{D}$), symptomatic infection ($I_+, I_{-k}$), and the observed test result ($Y_+, Y_{-k}$) and the pan-negative indicator ($P$) via the attribution mapping ($\alpha$). The dashed rectangle around $\alpha$ indicates that the attribution rule is a design choice and a mapping between random variables rather than a true state variable itself.}
\label{fig:detection-dag}
\end{figure}
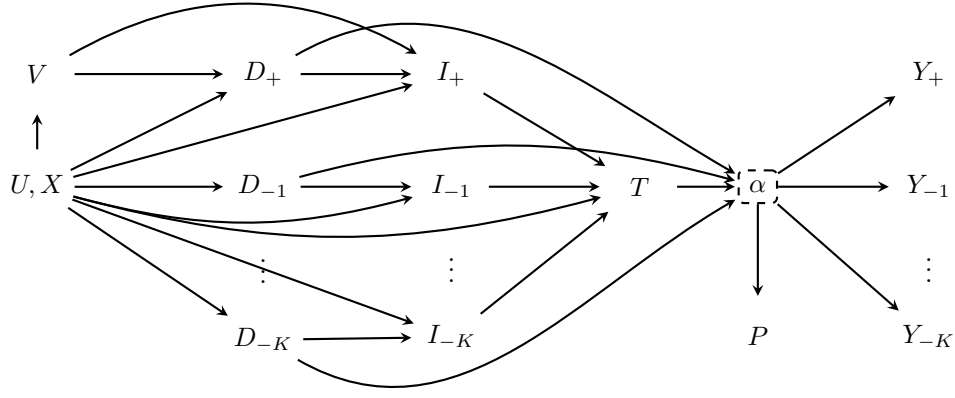

The DAG highlights three key features. First, the detection layer $(D_+, \mathbf{D}_{-})$ is an intermediate between the latent infection process and the observed outcome vector $\mathbf{Y} = (Y_+, Y_{-1}, \ldots, Y_{-K})$: it captures what the assay measures, which may diverge from what is clinically relevant. Second, the co-detection ($C_{-}, C_{+}$) and pan-negative ($P$) indicators are deterministic functions of the detection vector---they do not introduce new causal mechanisms, but they delineate the boundary cases that require explicit handling---while genuine co-infection is the corresponding property of the attributed outcome vector, arising when more than one component of $\mathbf{Y}$ equals $+1$. Third, the attribution mapping $\alpha$ is a \emph{design choice}, not a data-generating quantity: different investigators analyzing the same data may apply different mappings, yielding different analytic samples and different estimands. Making $\alpha$ explicit encourages pre-specification and transparency about how co-detections, co-infections, and pan-negatives are resolved.

\section{Identification under imperfect multiplex testing}\label{sec:app_imperfect_test}

The identification results of Appendices~\ref{sec:app_identification_1}--\ref{sec:app_identification_2} and~\ref{sec:app_pooling} were stated for the true outcome $Y$. In practice, the analyst observes a mismeasured outcome derived from the multiplex detection vector $(D_+, D_{-1}, \ldots, D_{-K})$. Mirroring the attribution mapping of Appendix~\ref{sec:app_codetection}, each assay yields a per-pathogen observed outcome $\widetilde{Y}_j$ ($+1$ if tested and assay-positive, $-1$ if tested and assay-negative, $0$ if untested), collected into the observed vector $\widetilde{\mathbf{Y}} = (\widetilde{Y}_+, \widetilde{Y}_{-1}, \ldots, \widetilde{Y}_{-K})$. Under the mutual exclusivity maintained throughout this appendix, $\widetilde{\mathbf{Y}}$ collapses to the scalar mismeasured outcome $\widetilde{Y}$ exactly as in Equation~\eqref{eq:collapse}, with $\{\widetilde{Y} = +1\} = \{\widetilde{Y}_+ = 1\}$ and $\{\widetilde{Y} = -k\} = \{\widetilde{Y}_{-k} = 1\}$; we work with the scalar form below. Building on the single-pathogen result of \textcite{boyerIdentificationEstimationVaccine2026}, we extend the standard non-differential-misclassification conditions to the multiplex panel and show that the pathogen-specific and pooled multiplex TND estimators continue to identify $RR_V(X)$ under those conditions, and characterise the residual bias incurred when pan-negative episodes are included in the control pool.

\paragraph{Assumptions.}

\begin{enumerate}[label=(D\arabic*),ref=D\arabic*]
  \item \label{ass:D1} \textit{Perfect specificity for every assay.} For each $j \in \{+, -1, \ldots, -K\}$ and every individual $i$,
  $$I_{j,i} = 0 \;\Longrightarrow\; D_{j,i} = 0,$$
  equivalently $D_{j,i} = 1 \Rightarrow I_{j,i} = 1$. This is the assay-by-assay version of the perfect-specificity condition (G1) of \textcite{boyerIdentificationEstimationVaccine2026}.
  \item \label{ass:D2} \textit{Non-differential sensitivity for every assay.} For each $j$, define $se_j(X) = \Pr(D_j = 1 \mid I_j = 1, T = 1, X)$ and assume
  $$\Pr(D_j^v = 1 \mid V = v', I_j^v = 1, T = 1, X) \;=\; se_j(X) \qquad \forall\, v, v' \in \{0,1\}.$$
  Sensitivities may vary across assays and across covariate strata but not across vaccination strata within a stratum of $X$.
\end{enumerate}
Implicit in Assumption~\ref{ass:D2} is that $D_j$ depends on $V$ only through $(I_j, X)$, the standard non-differential-misclassification structure.

\begin{lemma}\label{lemma:imperfect}
Under Assumptions~\ref{ass:D1}, \ref{ass:D2}, and the mutual-exclusivity of $I_+$ and $\{I_{-k}\}$ maintained in Sections~\ref{sec:multiplex_implications}--\ref{sec:three_principles}, for every $v \in \{0,1\}$, $k \in \{1,\ldots,K\}$, and $X$,
\begin{align}
\Pr(\widetilde{Y} = +1 \mid V = v, X) &\;=\; se_+(X)\,\Pr(Y = +1 \mid V = v, X), \label{eq:imperfect_lemma_plus}\\
\Pr(\widetilde{Y} = -k \mid V = v, X) &\;=\; se_{-k}(X)\,\Pr(Y = -k \mid V = v, X). \label{eq:imperfect_lemma_k}
\end{align}
\end{lemma}

\begin{proof}
For \eqref{eq:imperfect_lemma_plus}, the event $\{\widetilde{Y} = +1\}$ requires $T = 1$ and $D_+ = 1$; by Assumption~\ref{ass:D1} applied to the focal assay, $D_+ = 1 \Rightarrow I_+ = 1$, so $\{\widetilde{Y} = +1\} = \{T = 1, D_+ = 1, I_+ = 1\}$. Factorising,
$$\Pr(\widetilde{Y} = +1 \mid V, X) = \Pr(D_+ = 1 \mid I_+ = 1, T = 1, V, X)\,\Pr(I_+ = 1, T = 1 \mid V, X),$$
and by Assumption~\ref{ass:D2} the first factor equals $se_+(X)$ and does not depend on $V$, while the second factor is $\Pr(Y = +1 \mid V, X)$ by definition.

For \eqref{eq:imperfect_lemma_k}, the event $\{\widetilde{Y} = -k\}$ requires $T = 1$, $D_+ = 0$, and $D_{-k} = 1$. By Assumption~\ref{ass:D1} applied to assay $k$, $D_{-k} = 1 \Rightarrow I_{-k} = 1$; by mutual exclusivity $I_{-k} = 1 \Rightarrow I_+ = 0$; and by Assumption~\ref{ass:D1} applied to the focal assay, $I_+ = 0 \Rightarrow D_+ = 0$. The condition $D_+ = 0$ is therefore automatic given $D_{-k} = 1$, and $\{\widetilde{Y} = -k\} = \{T = 1, D_{-k} = 1, I_{-k} = 1\}$. Factorising and applying Assumption~\ref{ass:D2} for assay $k$ gives \eqref{eq:imperfect_lemma_k}.
\end{proof}

\begin{proposition}\label{prop:imperfect_pathogen_specific}
Under Assumptions~\ref{ass:D1}, \ref{ass:D2}, and mutual exclusivity, the pathogen-specific TND odds ratio computed with the observed outcome equals the pathogen-specific TND odds ratio computed with the true outcome:
$$
OR_k^{\text{obs}}(X) \;\equiv\; \frac{\Pr(\widetilde{Y} = +1 \mid V=1, X)\,\Pr(\widetilde{Y} = -k \mid V=0, X)}{\Pr(\widetilde{Y} = -k \mid V=1, X)\,\Pr(\widetilde{Y} = +1 \mid V=0, X)} \;=\; OR_k(X).
$$
\end{proposition}

\begin{proof}
Substitute \eqref{eq:imperfect_lemma_plus} and \eqref{eq:imperfect_lemma_k} into both numerator and denominator: $se_+(X)$ and $se_{-k}(X)$ each appear with equal multiplicity above and below the bar and cancel, leaving the right-hand side equal to $OR_k(X)$.
\end{proof}

Combining Proposition~\ref{prop:imperfect_pathogen_specific} with the pathogen-specific identification result of Appendix~\ref{sec:app_pooling} (Proposition~\ref{prop:identification_equiconf} applied with $Y = -1$ replaced by $Y = -k$ throughout, together with the pathogen-specific equi-confounding restriction \eqref{eq:equiconf_k}) gives $OR_k^{\text{obs}}(X) = RR_V(X)$ \emph{exactly}, with no rarity assumption. This is the multiplex analogue of the \textcite{boyerIdentificationEstimationVaccine2026} result, applied pathogen-by-pathogen.

\begin{proposition}\label{prop:imperfect_pooled}
Let $\mathcal{S}$ denote the retained control set and suppose the pathogen-specific equi-confounding restriction \eqref{eq:equiconf_k} holds for every $k \in \mathcal{S}$ with common multiplicative confounding ratio $\rho(X)$. Then under Assumptions~\ref{ass:D1}, \ref{ass:D2}, mutual exclusivity, and the remaining assumptions of Proposition~\ref{prop:identification_equiconf} applied to each $k \in \mathcal{S}$,
$$
OR_{\text{pool}}^{\text{obs}}(X) \;\equiv\; \frac{\Pr(\widetilde{Y} = +1 \mid V=1, X)\,\sum_{k \in \mathcal{S}}\Pr(\widetilde{Y} = -k \mid V=0, X)}{\sum_{k \in \mathcal{S}}\Pr(\widetilde{Y} = -k \mid V=1, X)\,\Pr(\widetilde{Y} = +1 \mid V=0, X)} \;=\; RR_V(X).
$$
\end{proposition}

\begin{proof}
By \eqref{eq:imperfect_lemma_k} and the common equi-confounding ratio,
$$
\sum_{k \in \mathcal{S}} \Pr(\widetilde{Y} = -k \mid V = 1, X) = \sum_{k \in \mathcal{S}} se_{-k}(X) \Pr(Y = -k \mid V = 1, X) = \rho(X) \sum_{k \in \mathcal{S}} se_{-k}(X) \Pr(Y = -k \mid V = 0, X),
$$
and the corresponding sum at $V = 0$ equals $\sum_{k \in \mathcal{S}} se_{-k}(X) \Pr(Y = -k \mid V = 0, X)$, so the ratio of the two sums is $\rho(X)$. By \eqref{eq:imperfect_lemma_plus}, the focal sensitivities $se_+(X)$ cancel between numerator and denominator. Therefore $OR_{\text{pool}}^{\text{obs}}(X) = r_+(X)/\rho(X)$, where $r_+(X) = \Pr(Y = +1 \mid V=1, X)/\Pr(Y = +1 \mid V=0, X)$. Applying Proposition~\ref{prop:identification_equiconf} to the true outcome gives $r_+(X)/\rho(X) = RR_V(X)$.
\end{proof}

Proposition~\ref{prop:imperfect_pooled} is slightly stronger than the single-pathogen result: heterogeneous control sensitivities $\{se_{-k}(X)\}$ do not bias the pooled multiplex estimator, because the equi-confounding restriction forces them to attach proportionally to the $V=1$ and $V=0$ pooled-control probabilities and cancel in the ratio.

\paragraph{Pan-negative inclusion.}
Pan-negative episodes are the event $\widetilde{Y}_{\text{pan}} = \{T = 1, D_+ = 0, D_{-k} = 0 \;\forall k \in \mathcal{S}\}$. Decomposing by true status under Assumptions~\ref{ass:D1}, \ref{ass:D2}, mutual exclusivity, and the special case in which pan-negativity arises only from measurement error on the panel (no infections by off-panel pathogens or non-infectious sources of symptoms),
\begin{equation}\label{eq:panneg_decomp}
\Pr(\widetilde{Y}_{\text{pan}} \mid V, X) \;=\; (1 - se_+(X))\,\Pr(Y = +1 \mid V, X) \;+\; \sum_{k \in \mathcal{S}} (1 - se_{-k}(X))\,\Pr(Y = -k \mid V, X).
\end{equation}
The missed-control contributions, by the same argument as Proposition~\ref{prop:imperfect_pooled}, carry the common multiplicative confounding ratio $\rho(X)$ and act as additional valid controls. The missed-focal contribution carries the focal multiplicative ratio $r_+(X) = RR_V(X)\rho(X)$, which is case-like rather than control-like.

Augmenting the control denominator with $\widetilde{Y}_{\text{pan}}$ and applying \eqref{eq:imperfect_lemma_plus}--\eqref{eq:panneg_decomp} yields, after cancellation of the detected-control sensitivities as in Proposition~\ref{prop:imperfect_pooled},
\begin{equation}\label{eq:pool_pan_or}
OR_{\text{pool+pan}}^{\text{obs}}(X) \;=\; RR_V(X) \cdot \frac{Q_0(X) + (1 - se_+(X))\,p_0(X)}{Q_0(X) + (1 - se_+(X))\,RR_V(X)\,p_0(X)},
\end{equation}
where $p_0(X) = \Pr(Y = +1 \mid V = 0, X)$ and $Q_0(X) = \sum_{k \in \mathcal{S}}\Pr(Y = -k \mid V = 0, X)$. The correction factor equals one---and pan-negative inclusion is identifying---in any of the following three cases:
\begin{enumerate}[nosep,label=(\roman*)]
  \item \emph{Perfect focal sensitivity}, $se_+(X) = 1$: the missed-focal mass in pan-negatives vanishes and the augmented pool consists entirely of detected and missed controls, all of which carry the equi-confounding ratio.
  \item \emph{No vaccine effect}, $RR_V(X) = 1$: trivially.
  \item \emph{Rare focal infection in the tested}, $p_0(X) \to 0$ and $RR_V(X) p_0(X) \to 0$: the correction factor approaches one, recovering the standard rare-outcome regime.
\end{enumerate}
When focal sensitivity is imperfect, focal infection is not rare among the tested, and there is a real vaccine effect, pan-negative inclusion induces residual bias toward the null whose magnitude is bounded by \eqref{eq:pool_pan_or} and scales with $(1 - se_+(X))\,p_0(X)/Q_0(X)$, the fraction of the augmented control pool composed of missed focal infections. Pre-specified sensitivity analyses that report estimates with and without pan-negative inclusion (recommended in Section~\ref{sec:coinfection}) provide a direct empirical check on the magnitude of this bias.

\section{Practical implications of the two views: when control selection recommendations diverge}\label{sec:app_two_views_diverge}

Section~\ref{sec:two_views} introduced two frameworks for thinking about control selection in a TND: \emph{View 1}, in which controls should represent vaccination rates in the source population (requiring control exchangeability and no unmeasured confounding), and \emph{View 2}, in which controls are auxiliary outcomes used to correct for unmeasured confounding (requiring equi-confounding and vaccine irrelevance). This appendix clarifies when and how these views lead to divergent practical recommendations.

\paragraph{When the views converge.}

For many plausible violations of the TND assumptions, the two views recommend the same action. Consider a control pathogen with its own vaccine $V_2$ that is correlated with the focal vaccine $V$ because both target similar populations (e.g., elderly or high-risk individuals). This is the vaccine-entanglement scenario described in Section~\ref{sec:three_principles}. Under \emph{View 1}, if $V_2$ is not measured and adjusted for, there is an uncovered backdoor path $V \leftarrow U \to V_2 \to I_{-} \to T$ that violates control exchangeability, so the control should be excluded or the analysis conditioned on $V_2$ (if available). Under \emph{View 2}, the same backdoor path violates the equi-confounding assumption because the confounding path $V \leftarrow U \to V_2 \to I_{-}$ has no symmetric path through the focal outcome $I_+$, again recommending exclusion or conditioning on $V_2$. Both views reach the same practical conclusion: remove or adjust for this control unless the cross-vaccine correlation is measured.

\paragraph{When the views diverge: unmeasured equi-confounding.}

The two views diverge fundamentally when there exists a strong unmeasured factor $U$ that is:
\begin{enumerate}[nosep]
  \item Confounded with vaccination in the source population ($U$ is a common cause of $V$ and illness outcomes),
  \item Affects both the focal and control outcomes through the same mechanism (equi-confounding), and
  \item Cannot be measured or adjusted for in the analysis.
\end{enumerate}

Under these conditions, the control pathogen would normally be invalid as a TND control under \emph{View 1} because the vaccination distribution among control outcomes is not representative of the source population vaccination distribution---a violation of the control exchangeability assumption. However, the control's unrepresentativeness is precisely a reflection of the broader unmeasured confounding structure. From the perspective of \emph{View 2}, this makes the control \emph{valuable}, not problematic: the control's distorted vaccination distribution is an accurate representation of how unmeasured confounding generates the residual bias in the focal analysis, allowing the analyst to subtract out this bias structure.

\paragraph{Concrete example: pneumococcal disease by vaccine-targeted and non-vaccine-targeted serotypes.}

To illustrate this divergence concretely, consider analyzing vaccine effectiveness (VE) against invasive pneumococcal disease caused by vaccine-targeted strains (the focal outcome) using disease from non-vaccine-targeted strains as a control (a common and sensible within-organism comparison in multiplex systems). Suppose:
\begin{itemize}[nosep]
  \item The pneumococcal vaccine $V$ protects against vaccine-targeted serotypes (focal outcome) but not against non-vaccine-targeted serotypes (control outcome).
  \item An unmeasured factor $U$ (e.g., underlying immune competence, propensity for aggressive preventive care, or chronic illness burden) affects both vaccination receipt and susceptibility to \emph{both} serotypes equally.
  \item Individuals with poor immune competence (high $U$) are both more likely to be vaccinated and more susceptible to invasive pneumococcal disease regardless of serotype.
\end{itemize}

In the source population, suppose $\Pr(V=1 \mid X) = 0.50$ (50\% of the population is vaccinated). Among those vaccinated, suppose $40\%$ have poor immune competence ($U=1$), whereas among unvaccinated, $10\%$ have poor immune competence. This creates selection bias: vaccination is a proxy for poor immune competence in this setting.

In a test-negative design naively comparing cases of vaccine-targeted disease ($I_+ = 1$) to cases of non-vaccine-targeted disease ($I_- = 1$), the vaccination distribution is distorted compared to the source population because:
\begin{itemize}[nosep]
  \item Cases of vaccine-targeted disease ($I_+ = 1$) are enriched with vaccinated individuals (vaccination provides some protection, so unvaccinated individuals are over-represented among these cases, but vaccination also selects for $U=1$ so vaccinated individuals remain well-represented).
  \item Cases of non-vaccine-targeted disease ($I_- = 1$) are also enriched with vaccinated individuals \emph{through the unmeasured confounding}, for the identical reason: $U$ is a common cause of both vaccination and disease.
\end{itemize}

Consequently, cases of non-vaccine-targeted disease have a \emph{higher} proportion vaccinated than the source population (approximately $50\%$ of cases are vaccinated, but within-stratum differences are driven by $U$, not by $V$'s causal effect).

\emph{Under View 1:} This violates control exchangeability because the vaccination distribution among non-vaccine-targeted disease cases is not representative of the source population. An analyst following this view would exclude non-vaccine-targeted strains as controls, reasoning: ``We cannot use this control because the vaccinated are over-represented due to their poor immune status and therefore unrepresentative of the source population.'' More broadly, they may recommend against the TND at altogether due to the unavailability of a crucial confounder: poor immune status.

\emph{Under View 2:} The distorted vaccination distribution among cases is exactly the confounding structure we want to capture. An analyst following this view would include non-vaccine-targeted strains as controls, reasoning: ``Non-vaccine-targeted strains are affected by the same unmeasured confounding as vaccine-targeted strains. The fact that vaccinated individuals are over-represented among non-vaccine-targeted disease cases is a direct reflection of the unmeasured immune competence confounding that also biases the focal analysis. By comparing VE estimates between the two outcomes, we can subtract out the shared bias structure.''

\paragraph{Practical guidance.}

Ultimately, the choice between these approaches depends on the investigator's underlying assumptions about the confounding structure in the source population, the availability of measured confounders, and the availability of suitable control pathogens (and one's ability to correctly identify them). For instance, when extensive covariate information is available and measured confounders can be adequately adjusted for, the analyst may lean toward the View 1 approach. Conversely, when key confounders are unmeasured but suitable control pathogens exist, the View 2 approach may be more appropriate.

\section{Appendix: Additional simulation results}\label{sec:app_sims}

This appendix reports simulation results beyond the four main scenarios described in Section~\ref{sec:sim}. We first present two parameter sweeps that examine how bias varies continuously with control-pathogen prevalence and with violation magnitude across the three violation scenarios. We then describe scenarios that formalise the use of pan-negative episodes as controls (scenarios 5a and 5b, building on Appendix~\ref{sec:app_codetection}), separate the two co-detection regimes of Section~\ref{sec:coinfection} into a control--control--only case (scenario 6a) and a focal--control case with a difference-in-differences estimator (scenario 6b), and verify the Appendix~\ref{sec:app_imperfect_test} identification results under imperfect multiplex testing (scenario 7).

\subsection{Parameter sweeps}\label{sec:app_sweeps}

The main simulation study treats violation parameters as fixed. To understand how sensitive the estimators are to the degree of violation and to the composition of the control pool, we conducted two additional sweeps.

\paragraph{Prevalence sweep.}
For each of the three violation scenarios (scenarios 2--4), we varied the baseline intercept $\gamma_{02}$ of control pathogen 2 while holding all other parameters fixed. Varying $\gamma_{02}$ changes the marginal prevalence of the problematic control pathogen within the tested population and, consequently, its relative weight in the pooled estimator. We evaluated 20 equally spaced grid points spanning a marginal prevalence of approximately 2\%--40\% among tested individuals. For each grid point we ran 500 replicates and computed mean bias, RMSE, and Monte Carlo standard errors.

Figure~\ref{fig:prevalence_sweep} shows bias in estimated VE (in percentage points) as a function of the assumed prevalence of control pathogen 2 for the pooled and screened estimators. Several patterns emerge. First, for the screened estimator---which excludes pathogen 2---bias is near zero and stable across the entire prevalence range, confirming that the exclusion strategy remains valid regardless of how common the problematic pathogen is. Second, for the pooled estimator, bias scales with the weight of the violated pathogen: at very low prevalences of pathogen 2, the contamination is diluted by the remaining four control pathogens and bias is minimal; as pathogen 2 becomes more prevalent, its share of the control pool increases and bias grows. Third, the rate of increase differs by violation type: the vaccine irrelevance violation and at-home testing violation produce monotone bias growth, while entanglement-induced bias grows more slowly because $V_2$ is correlated but not identical to $V$, so the confounding transmitted through $V_2$ is partially diluted. The prevalence sweep underscores that the severity of bias from a single violated control depends strongly on its prevalence relative to the other controls---a consideration particularly relevant in settings where one pathogen dominates the season.

\begin{figure}[hp]
\centering
\includegraphics[width=0.95\textwidth]{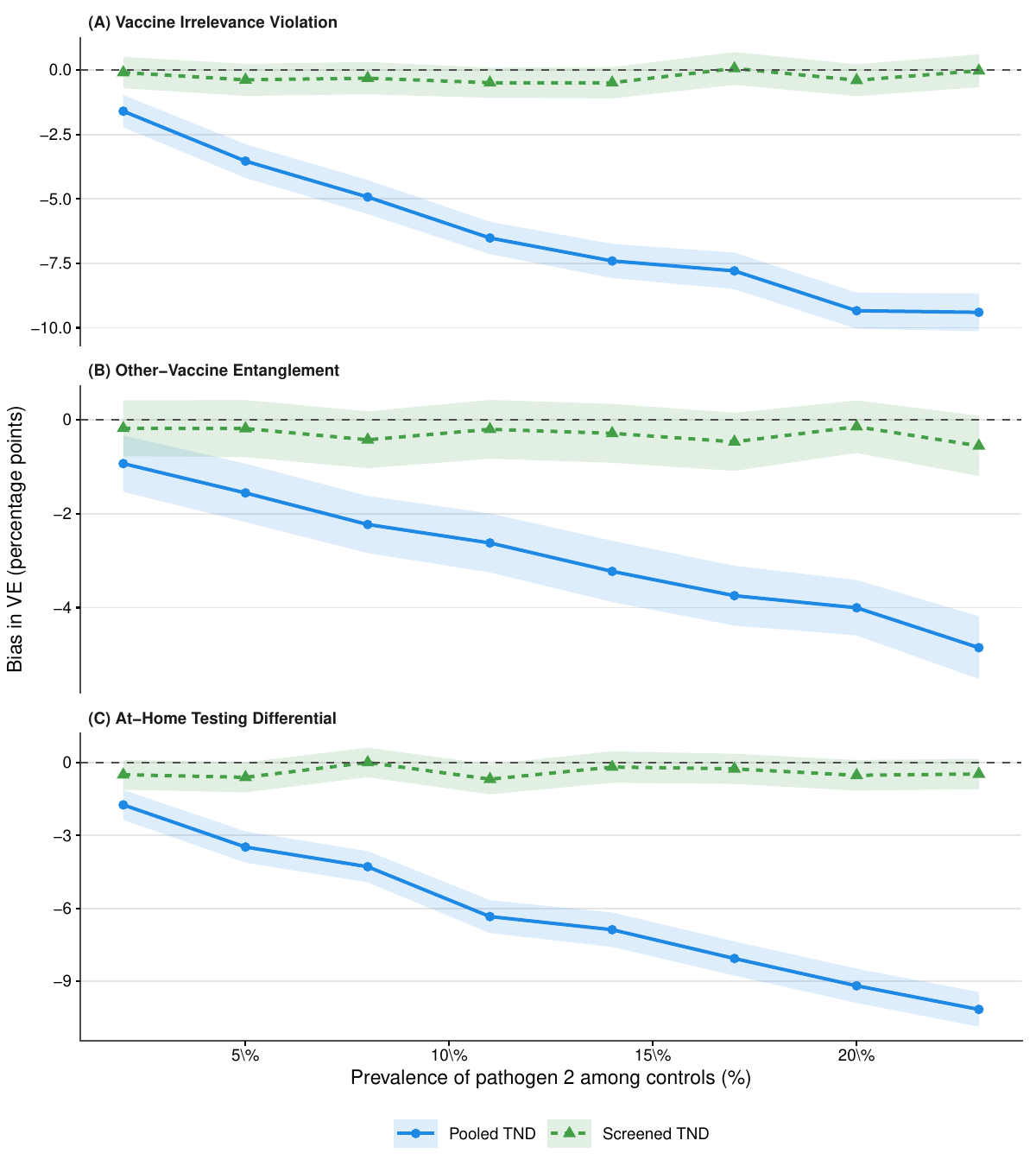}
\caption{Bias in estimated vaccine effectiveness (percentage points) as a function of the prevalence of control pathogen 2 (the pathogen violating the relevant assumption) for each of the three violation scenarios. Ribbons show $\pm 1.96$ Monte Carlo standard errors. The screened estimator (which excludes pathogen 2) is near-zero throughout; pooled estimator bias grows with pathogen prevalence. True VE\,=\,50\%.}
\label{fig:prevalence_sweep}
\end{figure}

\paragraph{Violation-strength sweep.}
We also examined how bias grows as the magnitude of each violation increases from zero. The violation-strength parameter is scenario-specific:
\begin{itemize}[nosep]
  \item \textbf{Vaccine irrelevance (scenario 2):} the cross-protection log-OR $\beta_V^{(2)}$ is varied from $0$ (no violation) to $\log(0.5)$ (same protection as against the focal pathogen).
  \item \textbf{Other-vaccine entanglement (scenario 3):} the $U$-coupling coefficient in $\Pr(V_2=1\mid X,U) = \mathrm{logit}^{-1}(\cdot + c_U \cdot U)$ is varied from $c_U = 0$ (no correlation between $V$ and $V_2$) to $c_U = 3.0$ (strong correlation).
  \item \textbf{At-home testing (scenario 4):} the differential at-home testing gap $\Delta p = p_{V\!=\!1} - p_{V\!=\!0}$ is varied from $0$ (no differential testing) to $0.5$ (large differential).
\end{itemize}
For each parameter value we ran 500 replicates. Figure~\ref{fig:violation_sweep} shows the resulting bias curves. All three violation types show a near-zero pooled-estimator bias when the violation parameter is at zero, confirming that the baseline scenario is correctly specified. Bias then increases roughly monotonically with violation strength for all three scenarios. The screened estimator remains close to zero throughout, as expected. The violation-strength sweep is useful for calibrating expectations in applied settings: if investigators can bound the plausible range of a violation parameter using external data or study design considerations, Figure~\ref{fig:violation_sweep} conveys how large the resulting bias is likely to be.

\begin{figure}[hp]
\centering
\includegraphics[width=0.95\textwidth]{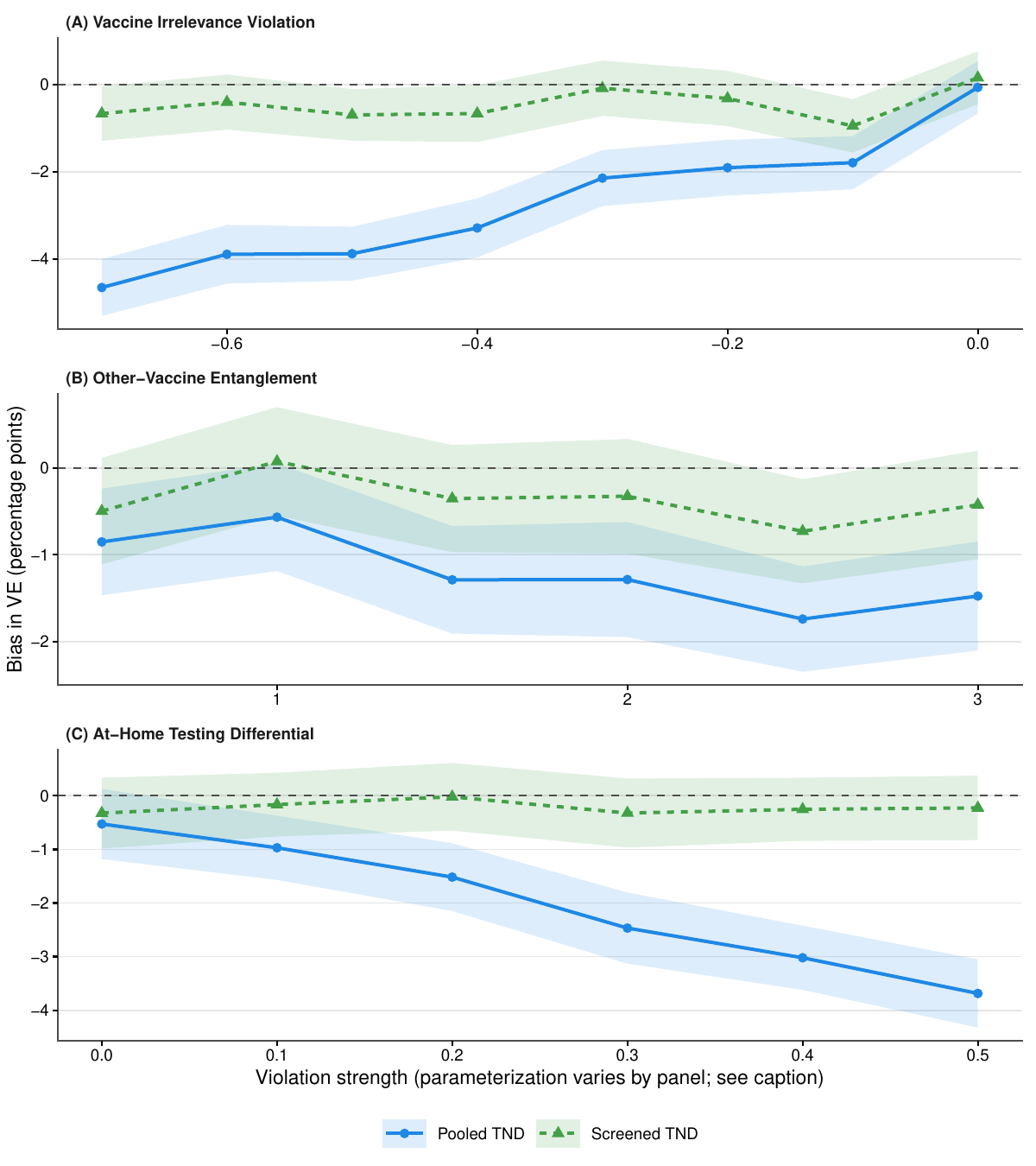}
\caption{Bias in estimated vaccine effectiveness (percentage points) as a function of violation magnitude for each of the three violation scenarios. The x-axis parameterisation differs by panel: cross-protection log-OR for vaccine irrelevance (panel A), $U$-coupling coefficient for entanglement (panel B), and differential at-home testing probability gap for at-home testing (panel C). Ribbons show $\pm 1.96$ Monte Carlo standard errors. True VE\,=\,50\%.}
\label{fig:violation_sweep}
\end{figure}

\subsection{Scenarios 5a and 5b: Pan-negative controls}\label{sec:app_panneg_sim}

The formal notation in Appendix~\ref{sec:app_codetection} defines pan-negative episodes as those in which a tested, symptomatic individual tests negative for all pathogens on the multiplex panel (Equation~\ref{eq:panneg}). While such episodes are typically discarded in standard TND analyses, they represent a potentially large pool of well-characterized negative controls whose inclusion could increase precision. The question is whether using pan-negatives as controls introduces bias. Section~\ref{sec:coinfection} discusses this conceptually; here we evaluate it through simulation.

We generated pan-negative episodes by introducing a \emph{background illness} process: a symptomatic tested-negative illness with probability
\[
\Pr(I_{\mathrm{bg}} = 1 \mid V, X, U) = \exp(-3.2 + 0.2X + 0.4U + \delta_V V),
\]
where $\delta_V$ governs whether the vaccine affects the background illness. Background illness causes clinical testing through the same testing model as pathogen-specific illness. All individuals with $I_{\mathrm{bg}} = 1$ and $T = 1$ who tested negative for all $K$ focal and control pathogens are designated as pan-negative controls with outcome $Y = -(K+1)$.

\paragraph{Scenario 5a (Pan-negatives valid).}
We set $\delta_V = 0$, so the vaccine has no effect on the background illness generating pan-negative episodes. Under this specification, pan-negatives satisfy vaccine irrelevance (the key validity condition for a TND control) and also share the same confounding structure as the focal pathogen through $U$. Including them as an additional control group should therefore not introduce bias.

\paragraph{Scenario 5b (Pan-negatives biased).}
We set $\delta_V = \log(0.7)$, so the vaccine reduces the probability of background illness by 30\%. Vaccinated individuals are less likely to present with pan-negative illness, making them under-represented in the pan-negative control group. This mimics a situation where the vaccine provides some non-specific protection (or is correlated with general health seeking behaviour), violating vaccine irrelevance for the pan-negative control group.

Figure~\ref{fig:pan_negative} shows pathogen-specific VE estimates across the $K$ control pathogens and the pan-negative group for both scenarios. Under scenario 5a, the pan-negative point estimate is consistent with the regular control pathogens and close to the true VE of 50\%, confirming that valid pan-negatives can be safely included. Under scenario 5b, the pan-negative point estimate is systematically biased downward (estimated VE $\approx 25\%$ versus the true $50\%$), reflecting the under-representation of vaccinated individuals among pan-negatives when the vaccine reduces background illness.

These results suggest that pan-negative controls are useful precisely when the vaccine has no plausible effect on the non-specific respiratory illness process generating them---a condition that should be assessed on biological grounds. When that condition is in doubt, a simple diagnostic is available: a pan-negative-specific pathogen-specific VE estimate that departs substantially from the estimates for established control pathogens is a signal of violation.

\begin{figure}[hp]
\centering
\includegraphics[width=0.85\textwidth]{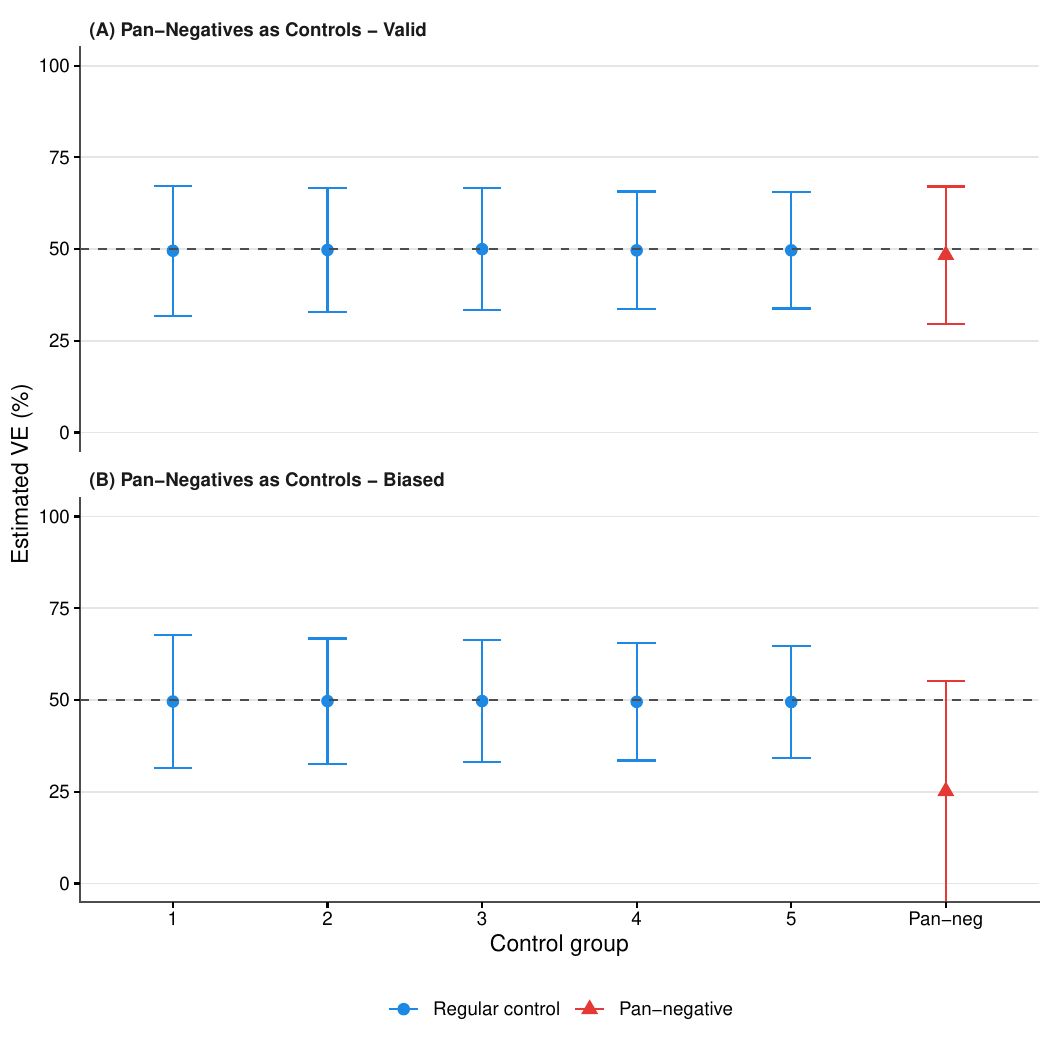}
\caption{Pathogen-specific vaccine effectiveness estimates for control pathogens 1--$K$ and the pan-negative group ($K+1$, shown in red) under scenarios 5a (pan-negatives valid; $\delta_V=0$) and 5b (pan-negatives biased; $\delta_V=\log(0.7)$). Error bars show $\pm 1.96$ simulation SD. Dashed line indicates true VE of 50\%. Under scenario 5b the pan-negative estimate departs from the other control pathogens, providing a diagnostic signal of violation.}
\label{fig:pan_negative}
\end{figure}

\subsection{Scenarios 6a and 6b: Co-detections}\label{sec:app_coinf_sim}

The main simulation study uses multinomial sampling to resolve infections to mutually exclusive outcomes, precluding co-detections by construction. Section~\ref{sec:coinfection} distinguishes two qualitatively different co-detection regimes: \emph{control--control} co-detections, which create an attribution ambiguity within the control pool but do not threaten identification of focal VE, and \emph{focal--control} co-detections, which asymmetrically deplete the control pool and bias the standard pooled estimator. Appendix~\ref{sec:app_focal_control} formalises the bias mechanism and identification result for the focal--control case. The two scenarios below evaluate each regime in turn.

\paragraph{Scenario 6a: control--control co-detections only.}
We retained the multinomial focal-vs-not-focal resolution from the main scenarios (so focal infection is mutually exclusive of control infection by construction), then drew each control pathogen independently among focal-negative individuals. Control pathogens can therefore co-occur with each other in the same episode, but the focal pathogen is never co-detected with a control. When multiple controls are detected on the same episode, we randomly attribute the episode to one of the positives uniformly (random attribution is a no-op in the multinomial scenarios 1--4 and 5a--b but is the natural symmetric rule when control--control co-detection is possible). The focal-pathogen intercept was kept at the baseline value (focal marginal infection probability $\approx 5\%$, preserving the rare-outcome regime in which $\mathrm{OR}\approx\mathrm{RR}$); only the control intercepts were elevated to $-2.0$ (each $\approx 14\%$ among focal-negatives) so the per-pair control--control co-detection rate is non-trivial. Under these settings the standard pathogen-specific and pooled estimators should be approximately unbiased, consistent with the Section~\ref{sec:coinfection} argument that control--control co-detection alone is not an identification threat under the maintained assumptions.

\paragraph{Scenario 6b: focal--control co-infection.}
We replaced the multinomial resolution with independent Bernoulli draws for each pathogen, allowing the focal pathogen to be co-detected with any subset of controls. To generate a non-trivial focal--control co-infection rate, we used the same elevated baseline intercepts as in scenario~6a (focal marginal infection probability $\approx 19\%$ and control marginal probability $\approx 14\%$ each), giving an expected focal--control co-infection rate of approximately $19\%\times14\%\approx 2.7\%$ per pair and roughly $57\%$ of focal-positive individuals carrying at least one control co-infection. Under the default attribution rule, focal-positive co-infected individuals are assigned to the case group and removed from the control pool, and because unvaccinated individuals face higher focal infection risk they are disproportionately depleted; the standard pooled TND estimator is therefore biased upward.

\paragraph{Modified-Poisson difference-in-differences estimator.}
For each control pathogen $k$, we restrict to individuals who tested positive for either the focal or the control pathogen ($I_+ = 1$ or $I_{-k} = 1$, including co-infected individuals in both arms) and stack the focal and control outcomes into a single dataset with an arm indicator. We fit a log-link (``modified'') Poisson working model that gives each arm its own intercept, $V$ slope, and $X$ slope:
\begin{align*}
\log\Pr(I_+ = 1 \mid V, X) &= a_0^+ + \beta_+ V + \bm{a}_X^{+\top} X, \\
\log\Pr(I_{-k} = 1 \mid V, X) &= a_0^k + \beta_k V + \bm{a}_X^{k\top} X.
\end{align*}
so that the within-arm coefficient on $V$ is directly the conditional log-risk-ratio (correctly specified because the data-generating process is log-linear in $X$ with no $V\times X$ interaction). The pairwise difference-in-differences contrast is
\[
\hat\delta_k = \hat\beta_+ - \hat\beta_k,
\]
the conditional log-risk-ratio for the focal pathogen. Modified Poisson (a log-link Poisson working model with robust standard errors) targets the risk ratio directly, so there is no odds-ratio/risk-ratio scale gap regardless of outcome prevalence. Because both arms are fit on the same individuals, we obtain $\mathrm{Var}(\hat\delta_k)$ from a cluster-robust (individual-clustered) sandwich covariance over the stacked fit---equivalent to an independence-working-correlation Poisson GEE---which captures $\mathrm{Cov}(\hat\beta_+,\hat\beta_k)$. This matters because the two log-risk-ratios are negatively correlated, so the naive independence combination $\sqrt{\mathrm{Var}(\hat\beta_+)+\mathrm{Var}(\hat\beta_k)}$ understates the variance and under-covers. We pool across the $K$ control pathogens with a single stacked modified-Poisson fit that imposes a common $\delta$ while retaining pathogen-specific nuisance terms (arm intercepts, arm $X$ slopes, and control-arm $V$ effects); its cluster-robust SE additionally accounts for the cross-pathogen correlation induced by the shared focal cases. We apply the estimator to both scenarios for parallelism, though it is principally motivated by scenario~6b.

Appendix~\ref{sec:app_focal_control} shows that under vaccine irrelevance for each control $k$ and equi-confounding on the log-risk-ratio scale, $\hat\delta_k$ identifies the causal focal log-risk-ratio and hence VE on the risk-ratio scale; the construction is analogous to the ``double-negative-control'' difference-in-differences in proximal causal inference \citep{soferNegativeOutcomeControl2016,tchetgenIntroductionProximalCausal2024}. The standard TND OR (Appendix~\ref{sec:app_identification_2}) targets an odds ratio; at the effect sizes considered here the two scales are close, so the estimators are directly comparable and the DiD's advantage is that it removes the depletion bias rather than any change of scale.

\paragraph{Results.}
Figure~\ref{fig:coinfection} compares pathogen-specific standard TND OR estimates to per-pathogen modified-Poisson DiD estimates for each of the $K=5$ control pathogens under both scenarios, together with the stacked common-$\delta$ pooled DiD estimate within each panel. Under scenario~6a (panel A), the standard pathogen-specific and pooled estimators are approximately unbiased, confirming that control--control co-detection alone is not an identification threat under the maintained assumptions; the DiD is shown for parity and carries a small upward bias---an artifact of the mutually-exclusive construction, in which the control arm is the complement of the focal arm within each pairwise subset at non-negligible control prevalence---so in this regime the standard estimator is preferred. Under scenario~6b (panel B), the standard pathogen-specific and pooled estimators are biased upward by the asymmetric control-pool depletion described above, whereas the modified-Poisson DiD---both per-pathogen and the stacked common-$\delta$ pooled estimate---recovers VE close to the truth. Because the pooled DiD's cluster-robust GEE standard error captures the strong positive correlation among the pathogen-specific contrasts (which share the focal cases), its confidence intervals attain near-nominal coverage, unlike a naive inverse-variance combination that would understate the variance.

\begin{figure}[hp]
\centering
\includegraphics[width=0.85\textwidth]{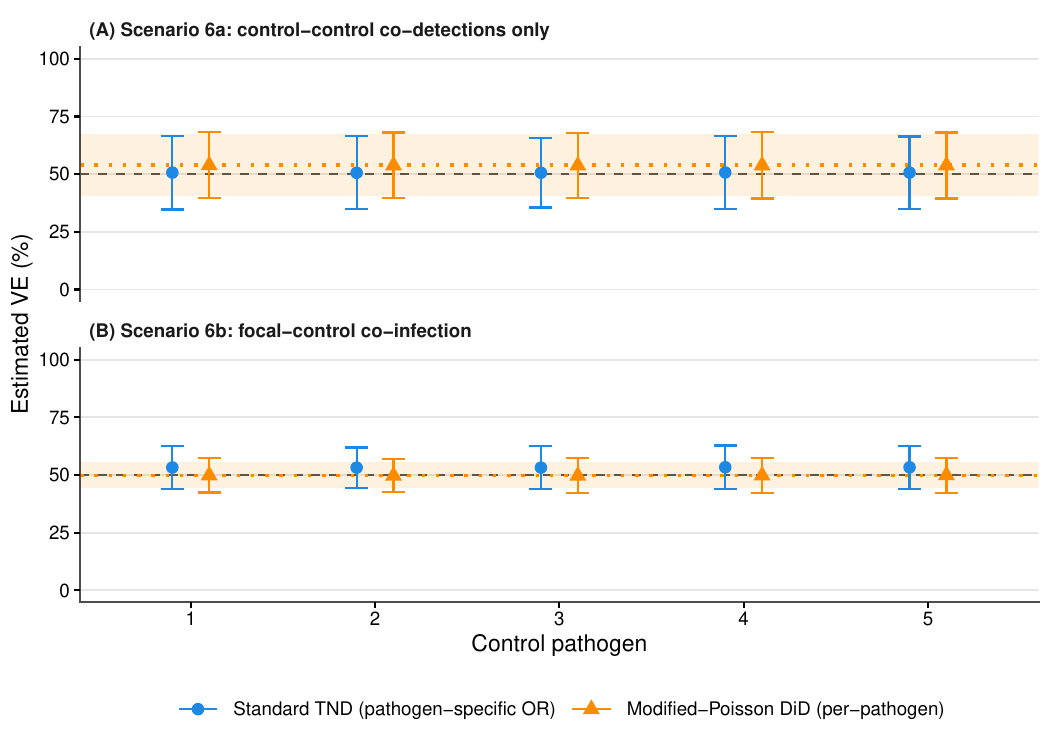}
\caption{Per-pathogen vaccine effectiveness estimates under scenarios 6a (control--control co-detections only, panel A) and 6b (focal--control co-infection, panel B). Blue circles: standard pathogen-specific TND logistic OR. Orange triangles: per-pathogen modified-Poisson DiD estimate. Orange dotted line and shaded band: stacked common-$\delta$ pooled DiD estimate $\pm 1.96$ simulation SD within each panel. Dashed grey line: true VE of 50\%. Error bars show $\pm 1.96$ simulation SD. The standard estimator is unbiased under control--control co-detection (A)---where it is preferred and the DiD is shown only for parity---but biased under focal--control co-infection (B), where the modified-Poisson DiD recovers the truth.}
\label{fig:coinfection}
\end{figure}

These results motivate including a modified-Poisson DiD estimator as a sensitivity analysis in any TND study where focal--control co-infection is plausible. Whether this estimator is preferred to exclusion of co-detected episodes (the strategy discussed in Appendix~\ref{sec:app_codetection}) depends on the prevalence of co-infections, the cost of the induced sample-size loss under exclusion, and the plausibility of the symmetry assumption underlying the DiD cancellation.

\subsection{Scenario 7: Imperfect multiplex testing and pan-negative inclusion}\label{sec:app_imperfect_sim}

Appendix~\ref{sec:app_imperfect_test} shows that, under perfect specificity and heterogeneous non-differential sensitivity for the focal and control assays, (i)~the pathogen-specific TND odds ratio identifies the focal vaccine risk ratio (Proposition~\ref{prop:imperfect_pathogen_specific}), (ii)~the pooled multiplex odds ratio identifies the same target after the control sensitivities cancel through the equi-confounding restriction (Proposition~\ref{prop:imperfect_pooled}), but (iii)~augmenting the pooled control denominator with measurement-error pan-negatives biases the estimator toward the null because the pan-negative group contains a contribution from missed focal infections that is case-like rather than control-like (Equation~\ref{eq:pool_pan_or}). Scenario~7 evaluates these three predictions in a setting in which pan-negatives arise \emph{only} from missed detections of true infections, so the bias in (iii) is fully attributable to the measurement-error mechanism.

\paragraph{Data-generating process for scenario 7.}
We retained the baseline DGP (mutually exclusive focal/control infection, no testing-irrelevance violations, no co-infection) and disabled the background-illness layer entirely by setting the background-illness intercept to $-100$ and the vaccine effect on background illness to zero, so $\Pr(I_{\mathrm{bg}}=1)\approx 0$ and pan-negatives can arise only from missed detections. Multiplex sensitivities were specified as $se_+ = 0.5$ for the focal assay and $se_{-k} \in \{0.90, 0.93, 0.95, 0.97, 0.99\}$ for the five control assays; all assays had perfect specificity, so $D_j = 0$ whenever $I_j = 0$. The low focal sensitivity routes many true focal-positive individuals into the pan-negative pool, while the high control sensitivities keep the pan-negative pool from being diluted by missed (valid) controls. To make the resulting bias clearly visible, we additionally lowered the control-pathogen prevalences (shifting each control intercept by $-1.0$ on the log scale, to roughly $1$--$2\%$) so that the detected-control pool $Q_0(X)$ is comparable in size to the focal prevalence $p_0(X)$; this enlarges the missed-focal share $(1 - se_+)\,p_0/Q_0$ of the augmented control pool (Equation~\ref{eq:pool_pan_or}) while keeping the focal pathogen rare (so $OR \approx RR$ and the reference estimators remain unbiased). Conditional on being tested, an individual's observed multiplex result $\widetilde{Y}$ may therefore differ from the true outcome $Y$: a true focal-positive individual contributes to the focal-positive sample with probability $se_+$ and to the pan-negative sample with probability $1 - se_+$, and analogously for each control. The focal pathogen was kept at its baseline (rare) prevalence, so the true focal VE remains 50\%.

\paragraph{Estimators.}
We applied three estimators to the analytic dataset (which now includes the pan-negative episodes alongside the detected cases and controls): the pathogen-specific TND logistic OR for each of the five regular control pathogens; the \emph{pooled estimator excluding pan-negatives}, which pools only detected control episodes; and the \emph{pooled estimator including pan-negatives}, which adds the pan-negative episodes to the control denominator. All three were fit with the same logistic-OR specification adjusting for $X$.

\paragraph{Results.}
Figure~\ref{fig:measurement_error} shows the three estimators. The five pathogen-specific estimates are approximately unbiased and centered close to the true VE of 50\%, consistent with Proposition~\ref{prop:imperfect_pathogen_specific}: heterogeneous sensitivities $se_{-k}$ rescale numerator and denominator of each pathogen-$k$ OR identically and cancel. The pooled estimator excluding pan-negatives is similarly close to the truth, confirming Proposition~\ref{prop:imperfect_pooled} and demonstrating that the standard pooled multiplex estimator is robust to heterogeneous non-differential sensitivity provided pan-negatives are excluded. The pooled estimator including pan-negatives is biased toward the null, reflecting the case-like contribution of missed focal infections in the pan-negative group identified in Equation~\ref{eq:pool_pan_or}. The magnitude of this bias depends on the focal-to-control prevalence ratio and on the gap between $se_+$ and the $\{se_{-k}\}$; for the parameter values used here it is substantial (on the order of $8$ percentage points of VE) and consistent across replicates, and as expected the pan-negative--inclusive estimator is the only one whose nominal $95\%$ coverage is materially affected.

\begin{figure}[hp]
\centering
\includegraphics[width=0.85\textwidth]{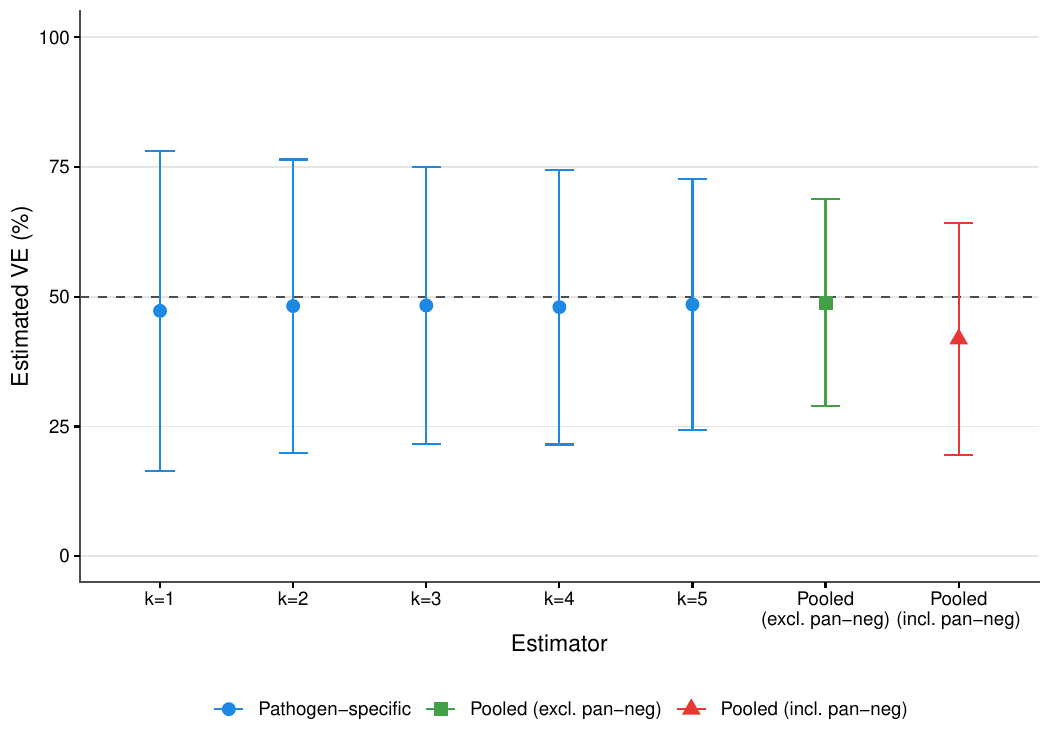}
\caption{Estimator behaviour under scenario~7 (imperfect multiplex testing; pan-negatives arise only from missed detections). Blue circles: pathogen-specific TND logistic OR for each of the five regular control pathogens (k=1,\ldots,5). Green squares: pooled estimator restricted to detected controls (Proposition~\ref{prop:imperfect_pooled}). Red triangles: pooled estimator that augments the control denominator with pan-negative episodes (Equation~\ref{eq:pool_pan_or}). Dashed grey line: true VE of 50\%. Error bars show $\pm 1.96$ simulation SD. The pathogen-specific and pooled-excluding estimators are approximately unbiased; the pan-negative--inclusive pooled estimator is biased toward the null.}
\label{fig:measurement_error}
\end{figure}

The practical implication is that, under imperfect multiplex testing with perfect specificity, pan-negative episodes should not be added to the control pool of a pooled TND analysis unless their measurement-error origin can be ruled out. If background-illness pan-negatives are the intended target (scenario~5a), the missed-detection contribution will contaminate them in any real multiplex setting in proportion to $1 - se_+$ and the relative focal prevalence, and the resulting bias is toward the null regardless of the direction of the focal vaccine effect.

\section{Appendix: A pragmatic control selection workflow}\label{sec:app_workflow}

This appendix provides a step-by-step workflow for pre-specifying test-negative control selection in multiplex TND studies. The steps operationalize the causal principles developed in the main text (vaccine irrelevance, other-intervention independence, testing-process comparability; Section~\ref{sec:three_principles}), apply the exposure-proxy/testing-proxy distinction (Figure~\ref{fig:multiplex-proxies}, Table~\ref{tab:proxy-examples}), and incorporate the co-infection and pan-negative recommendations developed in Section~\ref{sec:coinfection} and Appendices~\ref{sec:app_codetection} and~\ref{sec:app_imperfect_test}. Ideally, each step should be completed and documented \emph{before} the primary analysis; data-driven revisions should be limited to pre-specified sensitivity analyses.

Table~\ref{tab:workflow-summary} summarizes the workflow; detailed guidance follows.

\begin{table}[ht]
\centering
\caption{Summary of the pre-specified control selection workflow for multiplex TND studies. Each step should be documented in the study protocol or statistical analysis plan before data analysis begins.}
\label{tab:workflow-summary}
\small
\begin{tabular}{@{}clp{9.5cm}@{}}
\toprule
\textbf{Step} & \textbf{Action} & \textbf{Key considerations} \\
\midrule
0 & Define estimand & Specify focal pathogen, outcome severity, target causal contrast (primary: conditional $RR_V(X)$), target population, and study period \\
\addlinespace
1 & Define episodes and assay & Specify symptom screen and testing window; document per-pathogen assay sensitivity/specificity; retain control--control co-detections; default to excluding pan-negative episodes \\
\addlinespace
2 & Screen for vaccine irrelevance & Exclude controls plausibly affected by the focal vaccine through cross-protection, interference, or replacement \\
\addlinespace
3 & Screen for other-vaccine entanglement & Identify controls with correlated vaccines or coupled preventive programs; adjust for $V_2$ if measured, otherwise exclude \\
\addlinespace
4 & Apply the proxy lens & Prefer controls aligned on both the exposure ($U_1$) and testing ($U_2$) dimensions; flag for exclusion or stratification any control with substantial misalignment on either dimension, including pathogens with divergent clinical testing pathways \\
\addlinespace
5 & Choose estimator and pooling strategy & Default: pooled odds-ratio with pathogen-specific estimates and the pairwise control--control falsification check as diagnostics; if focal--control co-detection is non-negligible, use the modified-Poisson difference-in-differences estimator pooled via a stacked common-$\delta$ GEE \\
\addlinespace
6 & Specify analyses & Pre-specify the primary control set and sensitivity analyses: add-back of excluded pathogens, period stratification, pathogen-specific homogeneity test, pan-negative inclusion sensitivity, and DiD comparison \\
\bottomrule
\end{tabular}
\end{table}

\subsection*{Step 0: Define the scientific question and estimand}
\begin{itemize}[leftmargin=1.2em]
\item Specify the focal pathogen and the clinical severity of the outcome (e.g., any symptomatic illness, medically attended illness, hospitalization, ICU admission).
\item Specify the primary causal estimand: the conditional risk ratio among the vaccinated, $RR_V(X)$ (Appendix~\ref{sec:app_identification_1}). Indicate whether a marginal summary is also of interest.
\item Define the target population (e.g., community-dwelling adults $\geq 65$ years, children 6 months--17 years) and the study period.
\end{itemize}

\subsection*{Step 1: Define the analytic population, episodes, and assay}
\begin{itemize}[leftmargin=1.2em]

\item Document the per-pathogen sensitivity and specificity of the multiplex assay where available; these inform whether the pan-negative pool can be treated as reliably free of focal-pathogen infection (Appendix~\ref{sec:app_imperfect_test}).
\item \textbf{Co-detection handling.} Decide on an attribution rule for resolving co-detected episodes (Appendix~\ref{sec:app_codetection}). If true co-infections identified after attribution mapping, distinguish \emph{control--control} co-infection (multiple non-focal positives in one specimen) from \emph{focal--control} co-infection (focal positive together with one or more controls). For control--control co-infection, episodes can be retained if they do not violate identification assumptions: i.e., under equi-confounding and vaccine irrelevance the pooled estimator is invariant to attribution rules, and exclusion sacrifices precision and can introduce selection bias (Section~\ref{sec:coinfection}). For focal--control co-infection, the standard pooled estimator is biased through asymmetric depletion of the control pool; this is handled at the estimator level in Step 5 rather than by exclusion at this stage.
\item \textbf{Pan-negative episodes.} Default to \emph{excluding} pan-negative episodes from the control pool. Including them additionally requires that the focal vaccine not affect pan-negative episode rates (extended vaccine irrelevance to the pan-negative pool) and that pan-negatives episodes do not otherwise violate equi-confounding assumption. Both are difficult to verify given unknown identity of the pathogens contributing to the pan-negative pool and are further compromised by imperfect multiplex sensitivity contaminating the pan-negative bucket with undetected focal-pathogen cases (Appendices~\ref{sec:app_imperfect_test} and~\ref{sec:app_panneg_sim}). Include pan-negatives only with explicit biological justification, and always report both inclusion and exclusion estimates as a sensitivity analysis (Step 6).
\end{itemize}

\subsection*{Step 2: Screen for vaccine irrelevance}
\begin{itemize}[leftmargin=1.2em]
\item Review the biological literature to identify control pathogens plausibly affected by the focal vaccine through cross-protection, pathogen interference, or ecological replacement (Figure~\ref{fig:violations-vaccine-irrelevance}).
\item Exclude such pathogens from the primary control set, and document the rationale, hypothesized direction, and supporting evidence for each exclusion.
\end{itemize}

\subsection*{Step 3: Screen for other-intervention entanglement}
\begin{itemize}[leftmargin=1.2em]
\item Identify control pathogens with their own vaccines or coupled preventive programs whose uptake is plausibly correlated with focal vaccine uptake (Figure~\ref{fig:violations-vaccine-entanglement}) but specific to the control pathogen. Common examples include RSV prophylaxis (nirsevimab in infants, RSV vaccines in older adults) and COVID-19 vaccines when the focal pathogen is influenza, especially when these are co-administered.
\item If the correlated vaccine is recorded in the study data, adjust for $V_2$ (e.g., include it as a covariate). If it is unmeasured, exclude the affected pathogen from the primary control set.
\item Where co-administration patterns vary sharply across calendar time (e.g., before vs.\ after a co-administration recommendation), consider period-stratified analyses.
\end{itemize}

\subsection*{Step 4: Apply the proxy lens to the remaining controls}
\begin{itemize}[leftmargin=1.2em]
\item Equi-confounding requires that the focal and control pathogens share \emph{both} unmeasured determinants of infection risk ($U_1$) and unmeasured determinants of care-seeking and clinical testing ($U_2$); failure on either dimension can break the test-negative cancellation. The proxy lens (Section~\ref{sec:three_principles}, Table~\ref{tab:proxy-examples}) classifies each candidate control by which dimension it plausibly covers: \emph{exposure proxies} share $U_1$, \emph{testing proxies} share $U_2$, and the ideal control plausibly covers both.
\item Prefer controls aligned on both dimensions (e.g., a respiratory pathogen with overlapping transmission route, comparable severity profile, and a clinical presentation that triggers the same care-seeking pathway as the focal pathogen).
\item When a candidate aligns on only one dimension, retain it only if there is substantive reason to believe the unmatched dimension is not a meaningful source of confounding in your setting and document that reasoning. Otherwise exclude the pathogen or relegate it to a stratification or sensitivity scheme. Two common failure modes deserve specific attention:
  \begin{itemize}[leftmargin=1em]
  \item \emph{$U_1$ misalignment:} controls with clearly different exposure determinants (e.g., fecal--oral or vector-borne controls in a respiratory study) cannot cancel unmeasured exposure-related confounding of $V \rightarrow I_+$.
  \item \emph{$U_2$ misalignment:} controls whose detection is routed through a meaningfully different clinical pathway from the focal pathogen---most commonly through at-home rapid testing or surveillance-program detection rather than symptom-driven clinical encounters (Figure~\ref{fig:violations-at-home-test})---can break $U_2$ alignment by differentially channelling positive cases away from the TND sample by vaccination status.
  \end{itemize}
\end{itemize}

\subsection*{Step 5: Choose the estimator and pooling strategy}
\begin{itemize}[leftmargin=1.2em]
\item \textbf{Default estimator.} Use the pooled test-negative odds-ratio estimator (logistic regression of case status on $V$ and $X$) on the retained control set, and report pathogen-specific estimates $\{\widehat{VE}_k\}$ alongside as a diagnostic. This hybrid strategy preserves power for the primary estimate while keeping pathogen-level heterogeneity visible.
\item \textbf{When focal--control co-infection is biologically plausible and non-negligible} (per the assay/biology assessment in Step 1): replace the pooled odds-ratio estimator with the difference-in-differences estimator, e.g. modified-Poisson or log-binomial regressions pooled via a stacked common-$\delta$ GEE with cluster-robust standard errors (Appendix~\ref{sec:app_codetection}; Scenarios 6a--b in Appendix~\ref{sec:app_coinf_sim}). This estimator removes the bias from asymmetric control-pool depletion at the cost of wider confidence intervals.
\item \textbf{Diagnostics to report alongside the primary estimate:}
  \begin{itemize}[leftmargin=1em]
  \item Wald homogeneity test across pathogen-specific estimates with cluster-robust standard errors (Appendix~\ref{sec:app_pooling}).
  \item Pairwise control--control vaccination contrast as a falsification check (Section~\ref{sec:coinfection}, Appendix~\ref{sec:app_codetection}): under the maintained assumptions this contrast should equal unity (in expectation), and a substantial deviation suggests equi-confounding has broken down within the control pool.
  \end{itemize}
\end{itemize}

\subsection*{Step 6: Primary analysis and pre-specified sensitivity analyses}
\begin{itemize}[leftmargin=1.2em]
\item \textbf{Primary analysis:} use the control set surviving Steps 2--4 with the estimator selected in Step 5.
\item \textbf{Pre-specified sensitivity analyses:}
  \begin{enumerate}[leftmargin=1.5em]
  \item Add back excluded pathogens one class at a time and assess stability of the VE estimate.
  \item Stratify by calendar periods aligned with major changes in testing availability or vaccine roll-out (e.g., introduction of at-home rapid tests; new co-administration recommendations).
  \item Report pathogen-specific VE estimates with the homogeneity test from Step 5.
  \item Compare results with and without pan-negative episodes in the control pool.
  \item If a correlated vaccine $V_2$ is measured, compare adjusted vs.\ unadjusted estimates for the affected pathogens.
  \item If the standard pooled odds-ratio estimator is used as the primary, additionally report the modified-Poisson DiD estimator from Step 5 to bound the magnitude of any focal--control co-detection bias.
  \end{enumerate}
\item Document all pre-specified and post-hoc decisions in the statistical analysis plan or supplementary materials.
\end{itemize}

\end{document}

%% file: results/tables/simulation_results.tex
\begin{table}[ht]
\centering
\small
\setlength{\tabcolsep}{4pt}
\caption{Simulation results: bias and RMSE of TND vaccine effectiveness estimators across the seven simulation scenarios.}
\label{tab:sim_results}
\begin{tabular}{llrrrrrr}
\toprule
Scenario & Estimator & VE (\%) & Bias & RMSE & Cov. & Cases & Controls \\
\midrule
1. Baseline & Pooled & 49.5 & -0.5 & 6.7 & 95.7 & 305 & 1656 \\
1. Baseline & Screened & 49.5 & -0.5 & 6.7 & 95.7 & 305 & 1656 \\
\addlinespace
2. Vaccine irrelevance & Pooled & 46.9 & -3.1 & 7.8 & 93.4 & 305 & 1623 \\
2. Vaccine irrelevance & Screened & 49.3 & -0.7 & 6.9 & 95.0 & 305 & 1382 \\
\addlinespace
3. Entanglement & Pooled & 48.8 & -1.2 & 7.0 & 95.6 & 308 & 1527 \\
3. Entanglement & Screened & 49.8 & -0.2 & 6.7 & 96.0 & 308 & 1391 \\
\addlinespace
4. At-home testing & Pooled & 47.4 & -2.6 & 8.0 & 92.6 & 304 & 1538 \\
4. At-home testing & Screened & 49.5 & -0.5 & 7.4 & 94.4 & 304 & 1380 \\
\addlinespace
5a. Pan-neg (valid) & Pooled & 49.8 & -0.2 & 6.6 & 95.9 & 305 & 1892 \\
5a. Pan-neg (valid) & Screened & 49.8 & -0.2 & 6.6 & 95.9 & 305 & 1892 \\
\addlinespace
5b. Pan-neg (biased) & Pooled & 47.8 & -2.2 & 7.3 & 95.3 & 306 & 1855 \\
5b. Pan-neg (biased) & Screened & 47.8 & -2.2 & 7.3 & 95.3 & 306 & 1855 \\
\addlinespace
6a. Co-det (ctrl--ctrl) & Diff-in-diff & 53.9 & 3.9 & 7.9 & 91.8 & 246 & 3793 \\
6a. Co-det (ctrl--ctrl) & Pooled & 50.6 & 0.6 & 7.3 & 94.6 & 246 & 3793 \\
6a. Co-det (ctrl--ctrl) & Screened & 50.6 & 0.6 & 7.3 & 94.6 & 246 & 3793 \\
\addlinespace
6b. Co-det (focal--ctrl) & Diff-in-diff & 49.9 & -0.1 & 2.9 & 94.4 & 1434 & 3076 \\
6b. Co-det (focal--ctrl) & Pooled & 53.3 & 3.3 & 4.6 & 84.6 & 1434 & 3076 \\
6b. Co-det (focal--ctrl) & Screened & 53.3 & 3.3 & 4.6 & 84.6 & 1434 & 3076 \\
\addlinespace
7. Imperfect testing & Pooled & 41.8 & -8.2 & 14.0 & 89.1 & 158 & 793 \\
7. Imperfect testing & Screened & 48.9 & -1.1 & 10.2 & 94.7 & 158 & 606 \\
\bottomrule
\end{tabular}
\\[0.5em]
\parbox{\linewidth}{\footnotesize Note: True VE = 50\%. VE, Bias, RMSE and Cov.\ (95\% CI coverage) in percentage points. Cases/Controls are means across replicates. Diff-in-diff (scenarios 6a--b) is the modified-Poisson difference-in-differences estimator with a stacked common-$\delta$ GEE and cluster-robust SE (Appendix~\ref{sec:app_coinf_sim}). In scenario 7, ``Screened'' denotes the pooled estimator restricted to detected (non-pan-negative) controls (Appendix~\ref{sec:app_imperfect_sim}).}
\end{table}

%% file: results/tables/homogeneity_test.tex
\begin{table}[ht]
\centering
\caption{Homogeneity Test Results: Wald Test with Cluster-Robust Standard Errors}
\label{tab:homogeneity}
\begin{tabular}{lrrr}
\toprule
Scenario & Mean Wald & SD Wald & Rejection Rate (\%) \\
\midrule
1. Baseline & 3.98 & 2.76 & 4.9 \\
2. Vaccine irrelevance & 10.35 & 5.36 & 51.5 \\
3. Other-vaccine entanglement & 5.89 & 4.06 & 16.6 \\
4. At-home testing & 10.11 & 5.48 & 47.6 \\
5b. Pan-negatives (biased) & 11.50 & 5.75 & 46.5 \\
5a. Pan-negatives (valid) & 5.06 & 3.22 & 5.7 \\
6b. Co-detections (focal-control) & 4.08 & 2.82 & 5.1 \\
6a. Co-detections (control-control) & 4.09 & 2.79 & 5.2 \\
7. Imperfect multiplex testing & 15.40 & 6.75 & 71.8 \\
\bottomrule
\end{tabular}
\\[0.5em]
\footnotesize{Note: Wald test for homogeneity of vaccine effects across control pathogens,}
\footnotesize{with cluster-robust standard errors to account for repeated cases.}
\footnotesize{Rejection rate is the proportion rejecting $H_0$ at $\alpha = 0.05$.}
\footnotesize{Under $H_0$ (Scenario 1), rejection rate should be approximately 5\%.}
\end{table}